%% file: main.tex
\documentclass{lmcs} 
\pdfoutput=1

\usepackage{lastpage}
\lmcsdoi{22}{3}{3}
\lmcsheading{}{\pageref{LastPage}}{}{}%
{Sep.~03,~2025}{Jul.~21,~2026}{}


\usepackage{hyperref}
\theoremstyle{plain} 

\usepackage{adjustbox}

\newcommand{\macrospath}{./macros}

\input{\macrospath/ben-macros-preamble}
\setboolean{lipicsstyle}{true}
\setboolean{withproofs}{false}
\input{\macrospath/ben-macros-paper}
\input{\macrospath/ben-macros-diagrams}

\input{local-macros}

\begin{document}

\title{Mirroring Call-by-Need, or Values Acting Silly}

%
\author[B. Accattoli]{Beniamino Accattoli\lmcsorcid{0000-0003-4944-9944}}[a]
\author[A. Lancelot]{Adrienne Lancelot\lmcsorcid{0009-0009-5481-5719}}[b]

\address{Inria \& LIX, Ecole Polytechnique, UMR 7161, France}	
\email{beniamino.accattoli@inria.fr}  

\address{Inria \& LIX, Ecole Polytechnique, UMR 7161, France\\ \and Université Paris Cité, CNRS, IRIF, F-75013, Paris, France}	
\email{adrienne.lancelot@inria.fr}  





\begin{abstract}
  \noindent 	
Call-by-need evaluation for the $\lambda$-calculus can be seen as merging the best of  call-by-name and call-by-value, namely the wise erasing behaviour of the former and the wise duplicating behaviour of the latter. To better understand  how duplication and erasure can be combined, we design a degenerated calculus, dubbed \emph{call-by-silly}, that is symmetric to call-by-need in that it merges the worst of call-by-name and call-by-value, namely  silly duplications by-name and silly erasures by-value. 

We validate the design of the call-by-silly calculus via rewriting properties and multi types. In particular, we mirror the main theorem about call-by-need---that is, its operational equivalence with call-by-name---showing that call-by-silly and call-by-value induce the same contextual equivalence. This fact shows the blindness with respect to efficiency of call-by-value contextual equivalence. We also define a call-by-silly \emph{strategy} and a call-by-silly abstract machine implementing the strategy. Moreover, we measure the number of steps taken by the strategy via tight multi types. Lastly, we prove that the call-by-silly strategy computes evaluation sequences of maximal length in the calculus.

\end{abstract}

\maketitle

\section{Introduction}
\label{sect:intro}
Plotkin's call-by-value and Wadsworth's call-by-need (untyped) $\l$-calculi were introduced in the 1970s as more application-oriented variants of the ordinary call-by-name $\l$-calculus \cite{Wad:SemPra:71,DBLP:journals/tcs/Plotkin75}. The simpler call-by-value (shortened to \cbv) calculus has found a logical foundation in formalisms related to classical logic or linear logic \cite{DBLP:conf/icfp/CurienH00,phdlaurent,DBLP:journals/lisp/Levy06,DBLP:conf/ppdp/EhrhardG16}. 

The foundation of call-by-need (\cbneed) is less developed, particularly its logical interpretation.  The duality related to classical logic can accommodate \cbneed, as shown by Ariola et al. \cite{DBLP:conf/tlca/AriolaHS11,DBLP:conf/flops/AriolaDHNS12}, but it does not provide an explanation for it. Within linear logic, \cbneed is understood as a sort of \emph{affine} \cbv, according to Maraist et al. \cite{DBLP:journals/tcs/MaraistOTW99}. Such an interpretation however is unusual, because it does not match exactly with cut-elimination in linear logic, as for call-by-name (\cbn)~and~\cbv, and it is rather connected with \emph{affine logic}.

\myparagraph{\cbneed Optimizes \cbn} The main foundational theorem for \cbneed is the \emph{operational equivalence of \cbn and \cbneed} by Ariola et al. \cite{DBLP:conf/popl/AriolaFMOW95}, that is, the fact that they induce the same contextual equivalence, despite being based on different evaluation mechanisms. The result formalizes that \cbneed is a semantic-preserving optimization of \cbn. 

Kesner gives an elegant semantic proof of this result \cite{DBLP:conf/fossacs/Kesner16}. She shows that the \cbn multi type system by de Carvalho \cite{Carvalho07,DBLP:journals/mscs/Carvalho18} (considered before his seminal work also by Gardner \cite{DBLP:conf/tacs/Gardner94} and Kfoury \cite{DBLP:journals/logcom/Kfoury00}), which characterizes termination for \cbn (that is, $\tm$ is typable $\Leftrightarrow$ $\tm$ is \cbn terminating), characterizes termination also for \cbneed. Multi types  are also known as \emph{non-idempotent intersection types}, and 
have strong ties with linear logic.

\myparagraph{Duplication and Erasure, Wise and Silly.} The linear logic interpretation of \cbn and \cbv  provides some operational insights about \cbneed. By bringing to the fore duplication and erasure, linear logic underlines a second symmetry between \cbn and \cbv, which is independent of classical logic as it can already be observed in an intuitionistic setting. The idea is that \cbn is \emph{wise} with respect to erasure, because it never evaluates arguments that might be erased (for instance $(\la\var\Id)\Omega \Rew{\cbn} \Id$, where $\Id$ is the identity and $\Omega$ is the looping combinator), while it is \emph{silly} with respect to duplications, as it repeats possibly many times the evaluation of arguments that are used at least once (for instance, $(\la\var\var\var)(\Id\Id) \Rew{\cbn} \Id\Id(\Id\Id)$). Symmetrically, \cbv is silly with respect to erasures, as it reduces in arguments that are not going to be used (for instance looping on $(\la\var\Id)\Omega$), but it is wise for duplications, as it reduces only once arguments that shall be used at least once (\eg $(\la\var\var\var)(\Id\Id) \Rew{\cbv} (\la\var\var\var)\Id$).

\myparagraph{\cbneed is Pure Wiseness.} In this framework, one can see \cbneed as merging the best of \cbn and \cbv, i.e. wise erasure and wise duplication. In \cite{DBLP:conf/esop/AccattoliGL19}, this insight  guides Accattoli et al. in extending Kesner's above-mentioned study of \cbneed via  multi types  \cite{DBLP:conf/fossacs/Kesner16}. Starting from existing multi type systems for \cbn and \cbv, they pick out the aspects for \cbn wise erasures and \cbv wise duplication as to build a \cbneed multi type system. Their system characterizes \cbneed termination, as does the \cbn system in \cite{DBLP:conf/fossacs/Kesner16}. Additionally, it tightly characterizes \cbneed evaluation \emph{quantitatively}:  from type derivations  one can read the \emph{exact} number of \cbneed evaluation steps, which is not possible in the \cbn system.

\myparagraph{Pure Silliness.} This paper studies an unusual---and at first counter-intuitive---combination of duplication and erasure. The idea is to mix together silly erasure and silly duplication, completing the following diagram of strategies with its new \emph{call-by-silly} (\cbs) corner:
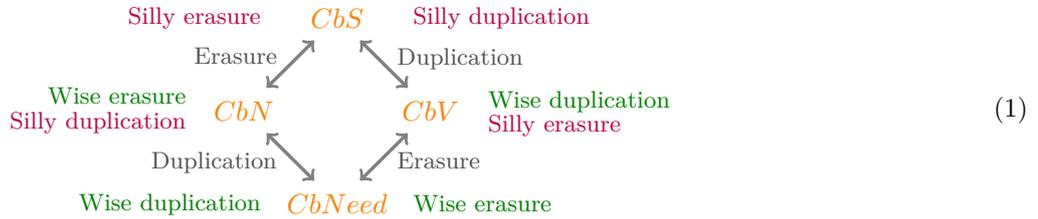
\begin{equation}
\begin{tikzpicture}[ocenter]
\node at (0,0) [anchor=center](silly){\large$\orange{CbS}$};
\node at (silly.center) [left =  25pt](sillyer){\small \purple{Silly erasure}};
\node at (silly.center) [ right =  25pt](sillydup){\small \purple{Silly duplication}};

\node at (silly.center) [below left = 35pt and 35pt, anchor=center](name){\large$\orange{CbN}$};
\node at (name.center) [above left = -1pt and 18pt](nameer){\small \dgreen{Wise erasure}};
\node at (name.center) [below left = -3pt and 18pt](namedup){\small \purple{Silly duplication}};

\node at (silly.center) [below right = 35pt and 35pt, anchor=center](value){\large$\orange{CbV}$};
\node at (value.center) [below right = -2pt and 18pt](valueer){\small \purple{Silly erasure}};
\node at (value.center) [above right = -4pt and 18pt](valuedup){\small \dgreen{Wise duplication}};

\node at (silly.center) [below = 70pt, anchor=center](need){\large$\orange{CbNeed}$};
\node at (need.center) [right =  25pt](needer){\small \dgreen{Wise erasure}};
\node at (need.center) [ below left =  -7pt and 25pt](needdup){\small \dgreen{Wise duplication}};

\draw[shorten <=2pt, shorten >=2pt, draw=gray, line width=0.3ex, <->](silly)to node[above right =-6pt and 1pt]{\small \dgray{Duplication}}(value);
\draw[shorten <=2pt, shorten >=2pt, draw=gray, line width=0.3ex, <->](silly)to node[above left =-4pt and 1pt]{\small \dgray{Erasure}}(name);
\draw[shorten <=2pt, shorten >=2pt, draw=gray, line width=0.3ex, <->](name)to node[below left=-6pt and 1pt]{\small \dgray{Duplication}}(need);
\draw[shorten <=2pt, shorten >=2pt, draw=gray, line width=0.3ex, <->](value)to node[below right=-6pt and 1pt]{\small \dgray{Erasure}}(need);
\end{tikzpicture}
\label{eq:strategies-diagram}
\end{equation}
Designing the \cbsilly calculus and a \cbs evaluation strategy is of no interest for programming purposes, as \cbs is desperately inefficient, by construction. It is theoretically relevant, however, because it showcases how  modularly duplication and erasure can be combined. As we shall discuss, the study of \cbs also contributes to the understanding of \cbv.

\myparagraph{Quantitative Goal and Micro Steps.} We introduce \cbs and provide evidence of its good design via contributions that mirror as much as possible the theory of \cbneed. Our ultimate test is to mirror the quantitative study of \cbneed by Accattoli et al. \cite{DBLP:conf/esop/AccattoliGL19}, via a new system of multi types for \cbs. It is a challenging design goal because this kind of quantitative results requires a \emph{perfect matching} between the operational and the multi types semantics.

As we explain at the end of \refsect{tight}, such a goal does not seem to be cleanly achievable within the $\l$-calculus, where duplication is \emph{small-step}, that is, arguments are substituted on \emph{all} the occurrences of a variable at once. It is instead attainable in a \emph{micro-step} setting with explicit substitutions, where one copy at a time is performed, or, similarly, when studying an abstract machine. Moreover, in the $\l$-calculus duplication and erasure are hard to disentangle as they are both handled by $\beta$-reduction.

\myparagraph{Our Framework.} A good setting for our purpose is Accattoli and Kesner's \emph{linear substitution calculus} (LSC) \cite{DBLP:conf/rta/Accattoli12,DBLP:conf/popl/AccattoliBKL14}, a calculus with explicit substitutions which is a variation over a calculus by Milner \cite{DBLP:journals/entcs/Milner07,KesnerOConchuir} exploiting ideas from another work of theirs \cite{DBLP:conf/csl/AccattoliK10}. A key feature of the LSC is that it has \emph{separate} rules for duplication and erasure, and \emph{just one rule for each}. 

The LSC comes in two dialects, \cbn and \cbv, designed over linear logic proof nets \cite{DBLP:conf/ictac/Accattoli18,DBLP:journals/tcs/Accattoli15}. The further \cbneed LSC is obtained by taking the duplication rule of the \cbv LSC and the erasing rule of the \cbn LSC. It was first defined by Accattoli et al. \cite{DBLP:conf/icfp/AccattoliBM14} and then studied or extended by many recent works on \cbneed \cite{DBLP:conf/fossacs/Kesner16,DBLP:journals/pacmpl/BalabonskiBBK17,DBLP:conf/ppdp/AccattoliB17,DBLP:conf/aplas/AccattoliB17,DBLP:conf/fossacs/KesnerRV18,DBLP:conf/ppdp/BarenbaumBM18,DBLP:conf/esop/AccattoliGL19,DBLP:conf/fossacs/KesnerPV21,DBLP:conf/fscd/BalabonskiLM21,DBLP:conf/csl/AccattoliL22}.

Our starting point is the principled definition of the new \emph{silly (linear) substitution calculus} (SSC), the variant of the LSC obtained by mirroring the construction for \cbneed, namely by putting together the \cbn rule for duplication  and the \cbv one for erasure. In this way,  all the corners of the strategy diagram \refeq{strategies-diagram} fit into the same framework. Since \cbneed is usually studied with respect to weak evaluation (that is, not under abstraction), we only consider weak evaluation for the SSC and leave the strong, unrestricted case to future work.

\myparagraph{Contribution 1: Rewriting Properties.} After defining the SSC, we prove some of the rewriting properties that are expected from every well-behaved calculus. Namely, we provide a characterization of its normal forms, that reduction is confluent, and that---as it is the case for all dialects of the LSC---the \cbv erasing rule can be postponed. 

Via the multi types of the next contribution, we also prove \emph{uniform normalization} for the SSC, that is, the fact that a term is weakly normalizing (\ie it reduces to a  normal form) if and only if it is strongly normalizing (\ie it has no diverging reductions). This is trivially true in (weak) \cbv, where diverging sub-terms cannot be erased and terms with redexes cannot be duplicated, and false in \cbn and \cbneed, where diverging sub-terms can be erased. In the SSC, it is true but non-trivial, because terms with redexes can be duplicated. 

\myparagraph{Contribution 2: Operational Equivalence.} Next, we develop the mirrored image of the main qualitative result for \cbneed, that is, we prove the operational equivalence of \cbs and \cbv. We do so by mirroring Kesner's semantic proof via multi types \cite{DBLP:conf/fossacs/Kesner16}. 

To this purpose, we cannot use the existing multi type system for \cbv, due to Ehrhard \cite{DBLP:conf/csl/Ehrhard12}, as Kesner's proof is based on a system for the \emph{slower} of the two systems, in her case \cbn. Therefore, we introduce a silly multi type system, designed in a dual way to the one for \cbneed by Accattoli et al. \cite{DBLP:conf/esop/AccattoliGL19}. 

Our system is a minor variant over a system for \cbn strong normalization by Kesner and Ventura \cite{DBLP:conf/ifipTCS/KesnerV14}, and it also bears strong similarities with a system for \cbv by Manzonetto et al. \cite{DBLP:journals/fuin/ManzonettoPR19,DBLP:conf/fscd/KerinecMR21}. We  prove that it characterizes termination for closed terms in both \cbv and \cbs, from which it follows the coincidence of their respective contextual equivalences. This result shows that \cbv is a semantic-preserving optimization of \cbs, exactly as \cbneed is for \cbn.

Additionally, we show that the relational semantics induced by multi types, where the interpretation $\sem\tm$ of a term is the set of multi type judgements that can be derived for $\tm$, is a denotational model of the SSC, by proving that its kernel relation is an equational theory of the SSC (\refthm{silly-conversion-is-included-in-eqcsilly}).

\myparagraph{\cbv Contextual Equivalence is Blind to Efficiency.} 
Contextual equivalence is usually considered as \emph{the} notion of program equivalence. 
This paper points out a fundamental limitation of contextual equivalence for the effect-free untyped \cbv $\l$-calculus. In absence of effects, contextual equivalence does distinguish between \cbv and \cbn because of their different erasing policies (as they change the notion of termination) but it is unable to distinguish between their duplication policies (which only impact the efficiency of evaluation), because it is \emph{blind to efficiency}: it cannot 'count' how many times a sub-term is evaluated. 

It is well known that, in richer frameworks with effects, \cbv contextual equivalence can 'count'. In presence of state, indeed, more terms are discriminated, as doing once \emph{or} twice the same operation can now change the content of a state and contexts are expressive enough to generate and distinguish these state changes \cite{DBLP:conf/birthday/StovringL09,DBLP:conf/fossacs/BiernackiLP19}. But it is instead not well known, or not fully digested, that in the pure setting this is not possible---thus that \cbv contextual equivalence validates silly duplications---as we have repeatedly noted in discussions with surprised colleagues. 

To the best of our knowledge, such a blindness to efficiency has not been properly established before. Proving contextual equivalence is technically difficult, often leading to hand-waving arguments. Moreover, figuring out a general characterization of silly duplications is tricky. We provide here an easy way out to prove \cbv contextual equivalence for terms related by silly duplications: it is enough to reduce them to the same normal form in the SSC, a calculus qualitatively equivalent to \cbv but based on silly duplications.

Note also that the equivalence of \cbn and \cbneed is interpreted positively (despite implying that \cbneed contextual equivalence is blind to efficiency) because it gives a foundation to \cbneed. Curiously, the mirrored equivalence of \cbv and \cbs crystallizes instead a  negative fact about pure \cbv, since \cbv is  expected to be efficiency-sensitive.


\myparagraph{Intermission: Difference with the Conference Version.} After the results about contextual equivalence, the conference version of this paper introduced a \cbs evaluation strategy and proved some of its properties. In this paper, we shall first detour via a \cbs abstract machine, and only after that we shall introduce the \cbs strategy. This is done both to extend the content of the previous paper with new material and to factorize---and hopefully make more accessible---the concepts for the strategy. The definition of the strategy, indeed, is quite involved while the abstract machine is somewhat simpler, despite being a more fine-grained setting.

\myparagraph{Contribution 3: The \cbs Abstract Machine.} We extend the simplest abstract machine for \cbn, namely the Milner abstract machine or MAM, as to also evaluate shared sub-terms when they are \emph{no longer needed}, \ie, after all the needed copies have already been evaluated. The obtained Silly MAM is operationally very intuitive: when the MAM would stop, it just jumps inside the first entry on the environment and recursively start evaluating the content. This mirror nicely what \cbneed machines do: they jump inside the environment the first time that an entry is needed. The study of the machine follows the distillation technique by Accattoli et al. \cite{DBLP:conf/icfp/AccattoliBM14}. The machine is proved to terminate on a term $\tm$ exactly whenever the calculus terminates on $\tm$.

\myparagraph{Contribution 4: The \cbs Strategy.} After the machine, we specify a tricky notion of \cbs \emph{strategy}. The delicate point is the definition of silly evaluation context. The \cbs strategy is  an extension of the \cbn strategy, on non-erasing steps.

We then prove that the \cbs strategy is (essentially) deterministic (precisely, it is \emph{diamond}, a weakened form of determinism), and that the \cbv erasing rule of the strategy can be postponed. These usually simple results have in our case simple but lengthy proofs, because of the tricky grammar defining \cbs evaluation contexts. We also prove that the Silly MAM implements the \cbs strategy.

\myparagraph{Contribution 5: Tight Types and (Maximal) Evaluation Lengths.} Lastly, we show that the quantitative result for \cbneed by Accattoli et al. \cite{DBLP:conf/esop/AccattoliGL19} is also mirrored by our silly type system. We mimic \cite{DBLP:conf/esop/AccattoliGL19} and we extract the exact length of evaluations for the \cbs strategy from a notion of \emph{tight} type derivation. 

We then use this quantitative result to show that the \cbs strategy actually computes a \emph{maximal} evaluation sequence in the weak SSC. This is not so surprising, given that maximal evaluations and strong normalization are related and measured with similar systems of multi types, as in Bernadet and Lengrand \cite{DBLP:conf/fossacs/BernadetL11}, Kesner and Ventura \cite{DBLP:conf/ifipTCS/KesnerV14}, and Accattoli et al. \cite{DBLP:journals/jfp/AccattoliGK20}. 

At the same time, however, the maximality of \cbs is specific to weak evaluation and different from those results in the literature, which concern strong evaluation (that is, also under abstraction). In particular, \cbs would \emph{not} be maximal in a strong setting (not treated here), as it would erase a value before evaluating it under abstraction (as \cbv also does). This is precisely the difference between our type system and Kesner and Ventura's.


\myparagraph{Related Work.} The only work that might be vaguely reminiscent of ours is by Ariola et al. \cite{DBLP:conf/tlca/AriolaHS11,DBLP:conf/flops/AriolaDHNS12}, who study \cbneed with respect to the classical duality between \cbn and \cbv and control operators. Their work and ours however are orthogonal and incomparable: they do not obtain inefficient strategies and we do not deal with control operators.

\myparagraph{Journal Version and Proofs} This paper is the journal version of the FSCD 2024 conference paper with the same title. It adds some explanations about the variants of the LSC, some results about equational theories, the study of the abstract machine, and some proofs. A few proofs are particularly long and tedious so they are still omitted or partially omitted, but can be found on Arxiv in the technical report \cite{accattoli2024mirroring} associated to the conference paper, which is accessible as version two ("v2", see the bibliography entry for the link) of the present paper on Arxiv.
\section{The Weak Silly Substitution Calculus}
\label{sect:ssc}
\begin{figure}[t!]
\centering
		$\begin{array}{c}
			\!\!\!\!\!\!\!\!\!
			\arraycolsep=3pt
			\begin{array}{c}
				\textsc{GRAMMARS}	
				\\			
				\begin{array}{r@{\hspace{0.3cm}}rll}
					\textsc{Terms} &  \tm,\tmtwo,\tmthree & \grameq& \var \mid \la\var\tm \mid \tm\tmtwo \mid \tm\esub\var\tmtwo
					\\
					\textsc{Values} & \val,\valtwo & \grameq & \la\var\tm
					  \\[4pt]
					\textsc{Substitution contexts} & \sctx,\sctxtwo & \grameq & \ctxhole \mid \sctx\esub\var\tmtwo
					\\
					\textsc{Weak contexts} & \wctx & \grameq & \ctxhole \mid \wctx\tm \mid \tm\wctx \mid \tm\esub\var\wctx \mid \wctx\esub\var\tmtwo
				\end{array}
				\end{array}
			\\[25pt]
			\hline
			\\[-10pt]
			\begin{array}{c|c}
			\textsc{ROOT RULES} 
			\\[4pt]
			\textsc{Multiplicative / Explicit $\beta$}
			\\
			\begin{array}{rll}
			\sctxp{\la\var\tm}\tmtwo & \rtom & \sctxp{\tm\esub\var\tmtwo}
			\end{array}
			&
			\begin{array}{rll}
			\towm  &\defeq&  \wctxp\rtom
			\end{array}
			\\[4pt]
			\textsc{Exponential / Duplication}
			\\
			\begin{array}{rlll}
			\textsc{By name} & \wctxfp\var\esub\var{\tmtwo} & \rtoep\wctx & \wctxfp\tmtwo\esub\var{\tmtwo}
			\\
			\textsc{By value} & \wctxfp\var\esub\var{\sctxp\val} & \rtoevp\wctx & \sctxp{\wctxfp\val\esub\var{\val}}
			\end{array}
			&
			\begin{array}{rll}
			\towe  &\defeq&  \wctxp{\rtoep\wctx}
			\\
			\towev  &\defeq&  \wctxp{\rtoevp\wctx}
			\end{array}
			\\[12pt]
			\textsc{Garbage collection (GC) / erasure}
			\\
			\begin{array}{rllll}
			\textsc{By name} & \tm\esub\var{\tmtwo} & \rtogc & \tm & \mbox{if }\var\notin\fv\tm
			\\
			\textsc{By value} & \tm\esub\var{\sctxp\val} & \rtogcv & \sctxp{\tm} & \mbox{if }\var\notin\fv\tm
			\end{array}
				&
			\begin{array}{rll}
			\towgc  &\defeq&  \wctxp\rtogc
			\\
			\towgcv  &\defeq&  \wctxp\rtogcv
			\end{array}
			\end{array}
			\\[25pt]
			\hline
			\\[-10pt]


				\begin{array}{c}
					\textsc{VARIANTS OF THE LSC}
					\\[4pt]
					\begin{array}{c|c|c|c|c|c}
					\textsc{LSC variant} & \textsc{Mult.} & \textsc{Dup. name}  & \textsc{Dup. value} & \textsc{GC name}  & \textsc{GC  value}
					\\
					&\towm & \towe & \towev & \towgc & \towgcv
					\\\hline\hline
					\text{Call-by-Name} & \checkmark  & \checkmark&&\checkmark
					\\[3pt]
					\hline
					\text{Call-by-Value} & \checkmark  & &\checkmark&&\checkmark					
					\\[3pt]
					\hline
					\text{Call-by-Need} & \checkmark  & &\checkmark&\checkmark					
					\\[3pt]
					\hline
					\text{Call-by-Silly} & \checkmark  & \checkmark&&&\checkmark					
					\end{array}
				\end{array}
				\\
			\\
			\hline
			\\[-10pt]
			\textsc{Weak reduction of the Silly (Linear) Substitution Calculus (SSC)}
			\\
				\begin{array}{rll}
									\tow &\defeq &  \towm \cup \towe \cup \towgcv
				\end{array}
\end{array}$			
				\caption{Variants of the linear substitution calculus (LSC).}
				\label{fig:variants-lsc}
		\end{figure}
In this section, we introduce the silly substitution calculus as a new variant of the linear substitution calculus, of which we recall the \cbn, \cbv, and \cbneed variants.

\myparagraph{Terms} Accattoli and Kesner's linear substitution calculus (LSC) \cite{DBLP:conf/rta/Accattoli12,DBLP:conf/popl/AccattoliBKL14} is a micro-step $\l$-calculus with explicit substitutions. 
Terms of the LSC are given in \reffig{variants-lsc} and extend the $\l$-calculus with \emph{explicit substitutions}  $\tm\esub{\var}{\tmtwo}$ (shortened to ESs), that is 
a more compact notation for $\letin\var\tmtwo\tm$, but where the order of evaluation between $\tm$ and $\tmtwo$ is a priori not fixed. \emph{Micro-step} means that substitutions act on one variable occurrence at a time, rather than \emph{small-step}, that is, on all occurrences at the same time. 

The set $\fv{\tm}$ of \emph{free} variables of a term $\tm$  is defined as expected, in particular,
$\fv{\tm\esub{\var}{\tmtwo}} \defeq
(\fv{\tm} \setminus \set{\var}) \cup \fv{\tmtwo}$. Both $\la{\var}\tm$ and $\tm\esub{\var}{\tmtwo}$ bind $\var$ in $\tm$, and terms are considered up to $\alpha$-renaming. 
A term $\tm$ is \emph{closed} if $\fv{\tm} = \emptyset$, \emph{open} otherwise.
Meta-level capture-avoiding substitution is noted $\tm\isub\var\tmtwo$. Note that values are only abstraction; this choice shall be motivated in the next section. 

\myparagraph{Contexts} Contexts are terms with exactly one occurrence of the \emph{hole} $\ctxhole$, an additional constant, standing for a removed sub-term. We shall use many different contexts. The most general ones in this paper are \emph{weak contexts} $\wctx$, which simply allow the hole to be anywhere but under abstraction. To define the rewriting rules, \emph{substitution contexts} $\lctx$ (\ie lists of explicit substitutions) also play a role.
The main operation about contexts is \emph{plugging} $\wctxp{\tm}$ where the hole $\ctxhole$ in context 
$\wctx$ is replaced by $\tm$. Plugging, as usual with contexts, can
capture variables---for instance $((\ctxhole \tm)\esub\var\tmtwo)\ctxholep\var = (\var\tm)\esub\var\tmtwo$. 
We write $\wctxfp{\tm}$ when we want to stress that the context $\wctx$ does not capture the free variables of $\tm$.

\myparagraph{Variants of the LSC} The LSC exists in two main variants, \cbn and \cbv. The \cbv variant is usually presented with small-step rules, and called \emph{value substitution calculus}, the micro-step (or linear) variant of which appears for instance in Accattoli et al. \cite{DBLP:conf/icfp/AccattoliBM14,DBLP:conf/esop/AccattoliGL19}. The two calculi have similar and yet different duplication and erasing rewriting rules. They are given in \reffig{variants-lsc} and explained below.

A third combination is the \cbneed LSC, that is obtained by taking the duplication rule of \cbv and the erasing rule of \cbn. 
Uniform presentations of \cbn, \cbv, and \cbneed and of their evaluation strategies in the LSC are given by Accattoli et al. in relation to abstract machines in \cite{DBLP:conf/icfp/AccattoliBM14} and to multi types in \cite{DBLP:conf/esop/AccattoliGL19}.

The Silly (Linear) Substitution Calculus (shortened to SSC) studied in this paper is obtained by taking the opposite combination, namely the duplication rule of \cbn and the erasing rule of \cbv. Its evaluation is micro-step. We only define its weak rewriting system, which does not reduce into abstractions, as it is often done in comparative studies between \cbn, \cbv, and \cbneed such as \cite{DBLP:conf/icfp/AccattoliBM14,DBLP:conf/esop/AccattoliGL19}. The SSC deals with possibly open terms. The abstract machine and the \cbs strategy, to be defined in later sections, shall instead deal with closed terms only.

Note that for each calling mechanism we have given the \emph{weak calculus}, that allows one to apply the rewriting rules in every position but under abstraction. Each mechanism then comes with a more restricted \emph{evaluation strategy}; in the literature, calculi and strategies are often confused. We shall give the \cbn strategy in \refsect{MAM} and introduce the \cbs strategy in \refsect{cbs-strategy}. We shall not define \cbv and \cbneed strategies; they can be found for instance in \cite{DBLP:conf/icfp/AccattoliBM14,DBLP:conf/esop/AccattoliGL19}.

\myparagraph{Rewriting Rules.}
The reduction rules of the LSC, and thus of the SSC, are slightly unusual as they use \emph{contexts} both to allow one to reduce redexes located in sub-terms, which is standard, \emph{and} to define the redexes themselves, which is less standard. This approach is 
called \emph{at a distance} by Accattoli and Kesner \cite{DBLP:conf/csl/AccattoliK10} and it is related to cut-elimination in proof nets (from which the terminology \emph{multiplicative} and \emph{exponential} to be discussed next is also taken), see Accattoli \cite{DBLP:journals/tcs/Accattoli15,DBLP:conf/ictac/Accattoli18}. 

The \emph{multiplicative rule} $\rtom$ is essentially the $\beta$-rule, except that the argument goes into a new ES, rather then being immediately substituted, and that there can be a substitution context $\lctx$ in between the abstraction and the argument. Example: $(\la\var\vartwo)\esub\vartwo\tm\tmtwo \rtom \vartwo\esub\var\tmtwo\esub\vartwo\tm$. One with on-the-fly $\alpha$-renaming is $(\la\var\vartwo)\esub\vartwo\tm\vartwo \rtom \varthree\esub\var\vartwo\esub\varthree\tm$.

The \emph{exponential rule} $\rtoep\wctx$ replaces a single variable occurrence, the one appearing in the context $\wctx$. Example: $(\var\var)\esub\var{\vartwo\esub\vartwo\tm} \rtoep\wctx (\var(\vartwo\esub\vartwo\tm))\esub\var{\vartwo\esub\vartwo\tm}$. The notation $\rtoep\wctx$ stresses that the rule is parametric in a notion of context $\wctx$, that specifies where the variable replacements are allowed, and which shall be exploited to define the \cbs strategy in \refsect{cbs-strategy}. 

The \emph{garbage collection (GC) rule by value}  $\rtogcv$ eliminates a value $\val$, keeping the list of substitutions $\lctx$ previously surrounding the value, which might contain terms that are not values, and so cannot be erased with $\val$. This rule is the \cbv erasing rule. Example:  $(\la\varthree\vartwo\vartwo)\esub\var{\Id\esub\varfour\tm} \rtogcv (\la\varthree\vartwo\vartwo)\esub\varfour\tm$, where $\Id$ is the identity.

The three root rules $\rtom$, $\rtoep\wctx$, and $\rtogcv$ are then closed by weak contexts. We shorten $\wctxp\tm \towm \wctxp\tmtwo$ if $\tm\rtom\tmtwo$ with $\towm \defeq \wctxp\rtom$, and similarly for the other rules. The reduction $\tow$ encompasses all possible reductions in the weak SSC.

\myparagraph{Silliness Check and the Silly Extra Copy.} Let us evaluate with the SSC the examples used in the introduction to explain silly and wise behaviour. About silly erasures, $(\la\var\Id)\Omega$ diverges, as in \cbv. One can reduce it to $\Id\esub\var\Omega$, but then there is no way of erasing $\Omega$, which is not a value and does not reduce to a value. About silly duplications, $(\la\var\var\var) (\Id\Id)$ can reduce (in various steps) both to $\Id\Id (\Id\Id)$ and to $(\la\var\var\var) \Id$, as in \cbn. The silly aspect is the fact that one \emph{can} reduce to $\Id\Id (\Id\Id)$, since in \cbv this is \emph{not} possible. The \cbs abstract machine of \refsect{machine} and the \cbv strategy of \refsect{cbs-strategy} shall always select the silly option.

Note that $(\la\var\var\var) (\Id\Id)$ reduces to $(\var\var)\esub\var{\Id\Id}$ and then \emph{three} copies of $\Id\Id$ can be evaluated, as one can reduce to $\Id\Id(\Id\Id)\esub\var{\Id\Id}$ and the copy in the ES has to be evaluated before being erased, because it is not a value. Such a further copy would not be evaluated in \cbn, as terms are erased without being evaluated. We refer to this aspect as to the \emph{silly extra copy}.

\section{Basic Rewriting Notions}
\label{sect:app-rewriting-notions}
Given a rewriting relation $\Rew{\Rule}$, we write $\deriv \colon \tm \Rew{\Rule}^* \tmtwo$ for a $\Rew{\Rule}$-reduction sequence from $\tm$ to $\tmtwo$, the length of which is noted $\size{\deriv}$. Moreover, we use $\size{\deriv}_a$ for the number of $a$-\emph{steps} in $\deriv$, for a sub-relation $\Rew{a}$ of $\Rew{\Rule}$. 

A term $\tm$ is \emph{weakly $\Rule$-normalizing}, noted $\tm\in\wn\Rule$, if $\deriv \colon \tm \Rew{\Rule}^* \tmtwo$ with $\tmtwo$ $\Rule$-normal; and $\tm$ is \emph{strongly $\Rule$-normalizing}, noted $\tm\in\sn\Rule$, if there are no diverging $\Rule$-sequences from $\tm$, or, equivalently, if all its reducts are in $\sn\Rule$.

A relation $\Rew{\Rule}$ is \emph{diamond}\footnote{We use the diamond property phrased for \emph{non-reflexive reductions}, unlike Terese \cite{Terese}. This more permissive definition of the diamond property can already be found in previous works, \eg Dal Lago and Martini \cite{DBLP:journals/tcs/LagoM08}.} if $\tmtwo_1 \,{}_\Rule\!\!\lto \tm \Rew{\Rule} \tmtwo_2$  imply $\tmtwo_1 = \tmtwo_2$ or $\tmtwo_1 \Rew{\Rule} \tmthree \, {}_\Rule\!\!\lto \tmtwo_2$ for some $\tmthree$. 
If $\Rew{\Rule}$ is diamond then:
\begin{enumerate}
	\item \emph{Confluence}:
	$\Rew{\Rule}$ is confluent, that is, $\tmtwo_1 \,{}_\Rule^*\!\!\lto \tm \Rew{\Rule}^* \tmtwo_2$  implies $\tmtwo_1 \Rew{\Rule}^* \tmthree \, {}_\Rule^*\!\!\lto \tmtwo_2$ for some $\tmthree$; 
	\item \emph{Length invariance}: all $\Rule$-evaluations to normal form with the same start term have the same length (\ie if $\deriv \colon \tm \Rew{\Rule}^* \tmtwo$  and $\deriv' \colon \tm \Rew{\Rule}^* \tmtwo$ with $\tmtwo$ $\Rew{\Rule}$-normal then $\size{\deriv} = \size{\deriv'}$);
	\item \emph{Uniform normalization}: $\tm$ is weakly $\Rule$-normalizing if and only if it is strongly $\Rule$-normalizing.
\end{enumerate}
Basically, the diamond property captures a more liberal form of determinism.

\section{Rewriting Properties}
\label{sect:rewriting_properties}
In this section, we study some rewriting properties of the weak SSC. We give a characterization of weak normal forms and prove postponement of garbage collection by value and confluence. 

\myparagraph{Characterization of Normal Forms.} The characterization of (possibly open) weak normal forms requires the following concept. 
\begin{defi}[Shallow free variables]
The set $\ofv\tm$ of \emph{shallow free variables} of $\tm$ is the set of variables with free occurrences out of abstractions in $\tm$:$\label{def:ofv}$
\[\begin{array}{lll @{\hspace{.5cm}} lll}
\ofv\var & \defeq & \set\var 
&
\ofv{\la\var\tm} & \defeq & \emptyset
\\
\ofv{\tm\tmtwo} & \defeq & \ofv\tm \cup \ofv\tmtwo
&
\ofv{\tm\esub\var\tmtwo} & \defeq & (\ofv\tm\setminus\set\var) \cup \ofv\tmtwo
\end{array}\]
\end{defi}
Intuitively, the occurrences of shallow free variables are the only variable occurrences that are actually replaceable in a weak calculus with ESs, because occurrences under abstractions are not reachable in such a setting.
%

\begin{prop}
$\tm$ is $\tow$-normal if and only if $\tm$ is a weak normal term according to the following grammar:\label{prop:weak-nfs}
\[
\begin{array}{r@{\hspace{.5cm}}r@{\hspace{.2cm}}l@{\hspace{.3cm}}lll}
\textsc{Weak answers} &
\ans,\ans & \grameq & \val &\multicolumn{2}{l}{\mid \ans\esub\var\itm \mbox{ with }\var\notin \ofv\ans}
\\
&&&&\multicolumn{2}{l}{\mid \ans\esub\var\anstwo \mbox{ with }\var\in \fv\ans\setminus\ofv\ans}
\\[4pt]
\textsc{Inert terms}
&
\itm,\itmtwo & \grameq & \var &\mid \itm \ntm &\mid  \itm\esub\var\itmtwo \mbox{ with }\var\notin \ofv\itm 
\\ 
&&&&& \mid  \itm\esub\var\ans \mbox{ with }\var\in \fv\itm\setminus\ofv\itm
\\[0pt]
\textsc{Weak normal terms}
&
\ntm,\ntmtwo &  \grameq &  \ans &\mid \itm
\end{array}
\]
\end{prop}

\begin{proof}
\emph{Direction $\Rightarrow$}. By induction on $\tm$. Cases:
\begin{itemize}
\item \emph{Variable} or \emph{abstraction}. Trivial.

\item \emph{Application}, that is, $\tm=\tmtwo\tmthree$. By \ih, $\tmtwo$ belongs to the grammar. It cannot be a weak answer, otherwise $\tm$ would have a $\towm$ redex. Then $\tmtwo$ is a inert term. By \ih, $\tmthree$ is a $\ntm$ term. Then $\tm$ has shape $\itm\ntm$ and belongs to the grammar.

\item \emph{Substitution}, that is, $\tm=\tmtwo\esub\var\tmthree$. By \ih, $\tmtwo$ belongs to the grammar. We have that $\tmtwo$ cannot be written as $\wctxfp\var$ otherwise there would be a $\towe$ redex. Then $\var\notin\ofv\tmthree$. By \ih, $\tmthree$ belongs to the grammar. Note that if $\tmthree$ is a weak answer then it has shape $\lctxp\val$. Then $\var\in\fv\tmthree$, otherwise there would be a $\togcv$ redex.
\end{itemize}
\emph{Direction $\Leftarrow$}. By induction on $\tm$.
\begin{itemize}
\item \emph{Variable} or \emph{abstraction}. Trivial.

\item \emph{Application}, that is, $\tm=\itm\ntm$. By \ih, both $\itm$ and $\ntm$ are $\tow$-normal. Moreover, $\itm$ does not have shape $\lctxp\val$, so that the root application is not a $\towm$ redex. Thus $\tm$ is $\tow$-normal.

\item \emph{Substitution}, that is, $\tm=\tmtwo\esub\var\tmthree$. It follows from the \ih and the conditions on $\var$ in the grammar, which forbid $\towe$-redexes because $\var\notin\ofv\tmtwo$ and forbid $\towgcv$-redexes because then $\tmthree$ is a weak answer then $\var$ is required to occur in $\tmtwo$ (necessarily under abstraction).\qedhere
\end{itemize}

\end{proof}

\myparagraph{Postponement of GC by Value.} As it is usually the case in all the dialects of the LSC, the erasing rule of the weak SSC, that is $\towgcv$, can be postponed. 
Let $\townotgcv \defeq \towm\cup\towe$. We first prove a local form of postponement, and then extend it to arbitrary reduction sequences.

\begin{lem}[Local postponement of weak GC by value]
\label{l:local-posponement-weak}\hfill
\begin{enumerate}
\item \emph{Weak multiplicative}: if $\tm \towgcv \tm'\towm \tm''$ then $\tm \towm\tm'''\towgcv \tm''$ for some $\tm'''$.
\item \emph{Weak exponential}: if $\tm \towgcv\tm'\towe \tm''$ then $\tm \towe\tm'''\towgcv^+ \tm''$ for some $\tm'''$.
\end{enumerate}
\end{lem}

\begin{proof}
Both points are by induction on $\tm \towgcv\tm'$ and are straightforward case analyses, see \cite{accattoli2024mirroring} for the details. In the exponential case, the increase of garbage collection steps happens when $\tm = \tmthree \esub\vartwo{\tmtwo} \towgcv \tmthree \esub\vartwo{\tmtwo'}=\tm'$ and $\tmthree \esub\vartwo{\tmtwo'} \towe \tm''$ is a root step, that is, when $\tmthree=\ctxtwofp\vartwo$ for some $\ctxtwo$. Then:

\begin{equation}
\begin{tikzpicture}[baseline={([yshift=3pt]current bounding box.south)}]
		\node at (0,0)[align = center](source){\normalsize $\ctxtwofp\vartwo \esub\vartwo{\tmtwo}$};
		\node at (source.center)[right = 200pt](source-right){\normalsize $\ctxtwofp\vartwo \esub\vartwo{\tmtwo'}$};
		\node at (source.center)[below = 25pt](source-down){\normalsize $\ctxtwofp{\tmtwo} \esub\vartwo{\tmtwo} $};
		\node at (source-right|-source-down)(target){\normalsize $\ctxtwofp{\tmtwo'} \esub\vartwo{\tmtwo'}$};
		\node at \med{source-down.center}{target.center}(fifthnode){\normalsize $\ctxtwofp{\tmtwo'} \esub\vartwo{\tmtwo} $};
		
		\draw[->](source) to node[above] {\scriptsize $\wgcv$} (source-right);
		\draw[->](source-right) to node[right] {\scriptsize $\wesym$}(target);
		
		\draw[->, dotted](source) to node[left] {\scriptsize $\wesym$}(source-down);
		\draw[->, dotted](source-down) to node[above] {\scriptsize $\wgcv$} (fifthnode);
		\draw[->, dotted](fifthnode) to node[above] {\scriptsize $\wgcv$} (target);
\end{tikzpicture}
\tag*{\qedhere}
\end{equation}
\end{proof}

\begin{prop}[Postponement of garbage collection by value]
If $\deriv:\tm\tow^*\tmtwo$ then $\tm\townotgcv^k\towgcv^h\tmtwo$ with $k=\sizep\deriv{\wsym\neg\gcv}$ and $h\geq \sizep\deriv{\wsym\gcv}$.
\label{prop:gc-postponement}
\end{prop}

\begin{proof}
The statement follows from local postponement (\reflemma{local-posponement-weak}) via a well-known basic rewriting property. Namely, the local swaps in \reflemma{local-posponement-weak}.1-2 above can be put together as the following local postponement property: if $\tm \togcv \tm' \tonotgcv \tm''$ then $\tm \tonotgcv\tm'''\togcv^* \tm''$ for some $\tm'''$. This local property is an instance of the hypothesis of Hindley's strong postponement property, which implies the postponement property in the statement. Precisely, Hindley's property usually appears for \emph{commutation} (rather than for postponement) of rewriting relations---for instance, as Lemma 3.3.6 in Barendregt's book on the $\l$-calculus \cite{Barendregt84}---but it is standard that postponement of $\Rew{1}$ after $\Rew{2}$ can be seen as commutation of $_1\!\leftarrow$ (the mirror relation of $\Rew{1}$) and $\Rew{2}$.
\qedhere
\end{proof}

\myparagraph{Local Termination.} To prepare for confluence, we recall the following crucial property, which is essentially inherited from the \cbn LSC, given that termination for $\towgcv$ is trivial.

\begin{prop}[Local termination]
 Reductions $\towm$, $\towe$, and $\towgcv$ are strongly normalizing separately.$\label{prop:local-termination}$
\end{prop}

\begin{proof}
Strong normalization of $\towm$ and $\towgcv$ is trivial, since they decrease the number of constructors. Strong normalization of $\towe$ follows from the fact that $\towe$ is an instance of the exponential rule $\toe$ of the LSC, which is strongly normalizing, as proved in \cite{DBLP:conf/rta/Accattoli12}. \qedhere
\end{proof}

\myparagraph{Confluence.}
The proof of confluence for the weak SSC given here is a minor variant of what would be done for the LSC, except that there is no direct proof of confluence in the literature for the LSC\footnote{Confluence of the LSC holds, as it follows from stronger results about residuals in Accattoli et al. \cite{DBLP:conf/popl/AccattoliBKL14} but here we want to avoid the heaviness of residuals.}. Overviewing the proof allows us to explain a design choice of the SSC.

The proof is based on an elegant technique resting on local diagrams and local termination, namely the Hindley-Rosen method. In our case, it amounts to prove that the three rules $\towm$, $\towe$, and $\towgcv$ are:
\begin{itemize}
\item Confluent separately, proved by local termination and Newman's lemma, and 
\item Commute, proved by local termination and Hindley's strong commutation. 
\item Confluence then follows by Hindley-Rosen lemma, for which the union of confluent and commuting reductions is confluent. 
\end{itemize}
The Hindley-Rosen method is a modular technique often used for confluence of extensions of the $\l$-calculus, for instance in \cite{ArrighiD17,DBLP:conf/lics/FaggianR19,DBLP:conf/csl/Saurin08,DBLP:conf/fossacs/CarraroG14,Revesz92,DBLP:journals/lmcs/BucciarelliKR21,DBLP:conf/popl/AriolaFMOW95,DBLP:conf/flops/AccattoliP12,DBLP:conf/lics/Accattoli22}. 


\begin{lem}[Local confluence]
Reductions $\towm$ and $\towgcv$ are diamond, $\towe$ is locally confluent.\label{l:local-confluence}
\end{lem}

\begin{proof}
The three points are all proved by unsurprising inductions on one of the spanning steps and case analysis of the other spanning step. The case of $\towe$ requires a bit of technical work, as it needs:
\begin{itemize}
\item A notion of \emph{double context}, that is, a context with two holes, and some of its basic properties; 
\item A \emph{deformation lemma} analyzing how to retrieve the decomposition $\wctxfp\var$ in $\tmtwo$ when $\wctxfp\var \toe \tmtwo$.
\end{itemize}
This approach mimics the proof of local confluence in Accattoli \cite{DBLP:journals/lmcs/Accattoli23}. All the details can be found in \cite{accattoli2024mirroring}. The only case worth pointing out is the following one, where a duplication of steps happens at a distance, that is, without one of the steps being fully contained in the other one (it is a standard case for the LSC, but the unacquainted reader might find it interesting):
\begin{equation}
\scriptsize
\begin{tikzpicture}[baseline={([yshift=3pt]current bounding box.south)}]
		\node at (0,0)[align = center](source){\normalsize $\wctxtwop{\wctxfourfp\vartwo\esub\vartwo{\wctxthreefp\var}} \esub\var\tmtwo$};
		\node at (source.center)[right = 150pt](source-right){\normalsize $\wctxtwop{\wctxfourfp\vartwo\esub\vartwo{\wctxthreefp\tmtwo}} \esub\var\tmtwo$};

		\node at (source-right|-source-down)(target){\normalsize $\wctxtwop{\wctxfourfp{\wctxthreefp\tmtwo}\esub\vartwo{\wctxthreefp\tmtwo}} \esub\var\tmtwo$};
		\node at (source.center)[below = 25pt](source-down){\normalsize $\wctxtwop{\wctxfourp{\wctxthreefp\var}\esub\vartwo{\wctxthreefp\var}} \esub\var\tmtwo$};
		\node at \med{source-down.center}{target.center}[below=20pt](fifthnode){\normalsize $\wctxtwop{\wctxfourp{\wctxthreefp\tmtwo}\esub\vartwo{\wctxthreefp\var}} \esub\var\tmtwo$};
		
		\draw[|->](source) to node[above] {\scriptsize $\wesym$} (source-right);
		\draw[->](source) to node[left] {\scriptsize $\wesym$}(source-down);
		
		\draw[->, dotted](source-right) to node[right] {\scriptsize $\wesym$}(target);
		\draw[|->, dotted](source-down) to node[above] {\scriptsize $\wesym$} (fifthnode);
		\draw[|->, dotted](fifthnode) to node[above] {\scriptsize $\wesym$} (target);
\end{tikzpicture}
\tag*{\qedhere}
\end{equation}
\end{proof}

In contrast to confluence, commutation of two reductions $\Rew{1}$ and $\Rew{2}$ does \emph{not} follow from their \emph{local} commutation and strong normalization. In our case, however, the rules verify a non-erasing form of Hindley's strong (local) commutation \cite{HindleyPhD} of $\Rew{1}$ over $\Rew{2}$, here dubbed \emph{strict strong commutation}: if $\tmtwo_{1} \lRew{1} \tm \Rew{2} \tmtwo_{2}$ then $\exists\tmthree$ such that $\tmtwo_{1} \Rewp{2} \tmthree \lRew{1} \tmtwo_{2}$, that is, on one side of the commutation there are no duplications of steps, and on both sides there are no erasures, because $\towgcv$ can only erase values, which do not contain redexes since evaluation is weak. Strong commutation and strong normalization do imply commutation.

\begin{lem}[Strict strong local commutations]
Reduction $\towe$ (resp. $\towe$; resp. $\towgcv$) strictly strongly commutes over $\towm$ (resp. $\towgcv$; resp. $\towm$).\label{l:local-strong-commutation}
\end{lem}

\begin{proof}
The three points are all proved by unsurprising inductions on one of the spanning steps and case analysis of the other spanning step. All the details can be found in \cite{accattoli2024mirroring}. \qedhere
\end{proof}

\begin{thm}[Confluence]
\label{thm:confluence}
Reductions $\tow$ and $\townotgcv$ are confluent.\label{thm:confluence}
\end{thm}

\begin{proof}
By Newman lemma, local confluence (\reflemma{local-confluence}) and local termination (\refprop{local-termination}) imply that $\towm$, $\towe$, and $\towgcv$ are confluent separately. By a result of Hindley \cite{HindleyPhD}, strict strong local commutation  (\reflemma{local-strong-commutation}) and local termination imply that $\towm$, $\towe$, and $\towgcv$ are pairwise commuting. By Hindley-Rosen lemma, $\tow = \towm \cup \towe \cup \towgcv$ and $\townotgcv = \towm \cup \towe $ are confluent.\qedhere
\end{proof}

\begin{rem}
The design choice that values are only abstraction is motivated by the commutation of $\towe$ over $\towgcv$. Indeed, if one considers variables as values (thus allowing the erasure of variables by $\togcv$), such a commutation fails, and confluence does not hold, as the following non-commuting (and non-confluent) span shows: $\var  \esub\varthree{\varfour\varfour} 
\lRew{\wsym\gcv} 
\var \esub\vartwo\varthree\esub\varthree{\varfour\varfour} 
\Rew{\esym_\wctx}
 \var \esub\vartwo{\varfour\varfour} \esub\varthree{\varfour\varfour}$.
 
%
Interestingly, something similar happens in \cbneed, where values are only abstractions too. Indeed, if one considers variables as values (thus allowing the duplication of variables in \cbneed), then one has the following non-closable critical pair for the \cbneed strategy, non-closable because $\vartwo$ is not needed in $\Id (\la\varthree\varfour\var) \esub\var\vartwo \esub\vartwo{\Id}$ (the diagram closes in the \cbneed \emph{calculus} but the \cbneed \emph{strategy} is not confluent):
\begin{center}
\begin{tikzpicture}[ocenter]
		\node at (0,0)[align = center](source){\normalsize $\var (\la\varthree\varfour\var) \esub\var\vartwo \esub\vartwo{\Id}$};
		\node at (source.center)[right = 180pt](source-right){\normalsize $\var (\la\varthree\varfour\var) \esub\var\Id \esub\vartwo{\Id}$};
		\node at (source-right.center)[below = 20pt](target){\normalsize $\Id (\la\varthree\varfour\var) \esub\var\Id \esub\vartwo{\Id}$};
				
		\node at (source.center)[below = 20pt](source-down){\normalsize $\vartwo (\la\varthree\varfour\var) \esub\var\vartwo \esub\vartwo{\Id}$};
		\node at \med{source-down.center}{target.center}(fifthnode){\normalsize $\Id (\la\varthree\varfour\var) \esub\var\vartwo \esub\vartwo{\Id}$};
		
		\draw[->](source) to  (source-right);
		\draw[->, dotted](source-right) to (target);
		
		\draw[->](source) to (source-down);
		\draw[->, dotted](source-down) to  (fifthnode);
	\end{tikzpicture}
\end{center}
\end{rem}

\section{Silly Multi Types}
\label{sect:multi_types}
In this section, we introduce multi types and the silly multi type system. 

\myparagraph{Inception.} The design of the silly type system is specular to the one for \cbneed by Accattoli et al. \cite{DBLP:conf/esop/AccattoliGL19}. The \cbneed one tweaks the \cbv system in the literature (due to Ehrhard \cite{DBLP:conf/csl/Ehrhard12}) by changing its rules for applications and ESs, which are the rules using multi-sets, as to accommodate \cbn erasures. The silly one given here tweaks the \cbn system in the literature (due to de Carvalho \cite{DBLP:journals/mscs/Carvalho18}) by changing the same rules to accommodate \cbv erasures. In both cases, the underlying system is responsible for the duplication behavior. The desired erasure behavior is then enforced by changing the rules using multi-sets.

The silly system is in \reffig{silly-types}. It is the variant of the system for \cbn in \cite{DBLP:conf/esop/AccattoliGL19} (itself a reformulation of \cite{DBLP:journals/mscs/Carvalho18} tuned for weak evaluation) obtained by adding '$\mplus\mult\atype$' in the right premise of both rules $\ruleAp$ and $\ruleES$. Such an extra typing for these rules captures the \emph{silly extra copy} mentioned at the end of \refsect{ssc}. More details are given at the end of this section. 
 \begin{figure}[t]
 \centering
$\begin{array}{cccc}
\begin{array}{rrcll}
\textsc{Linear types} & \ltype, \ltypetwo  & \grameq   & \atype \mid  \arrowtype{\mtype}{\ltype}
\\
\textsc{Multi types}
    & \mtype, \mtypetwo  & \grameq  & \mult{\ltype_i}_{i \in I} \ \ \mbox{where $I$ is a finite set}
    \\
\textsc{Generic types}
    & \ttype, \ttypetwo  & \grameq  & \ltype \mid \mtype
\end{array}
\\[20pt]
 {\footnotesize
  \begin{tabular}{c@{\hspace{.3cm}}cc}
      \AxiomC{}
      \RightLabel{\footnotesize$\ruleAx$}
      \UnaryInfC{$\var \hastype \mset\ltype \vdashp 0 1 \var \hastype \ltype$}
      \DisplayProof
      &
      \AxiomC{$(\typctx_i \vdashp{\msteps_i}{\esteps_i} \tm \hastype \ltype_i)_{i\in I}$}
      \RightLabel{\footnotesize$\ruleMany$}
      \UnaryInfC{$\uplus_{i\in I} \typctx_i \vdashp{\sum_\iI\msteps_i}{\sum_\iI\esteps_i} \tm \hastype \mset{ \ltype_i }_{i\in I}$}
      \DisplayProof
      \\[16pt]
      \AxiomC{}
      \RightLabel{\footnotesize$\ruleAxLam$}
      \UnaryInfC{$\vdashp 0 0  \la\var\tm \hastype \atype$}
      \DisplayProof
      &
	      \AxiomC{$\typctx \vdashp\msteps\esteps \tm \hastype  \arrowtype\mtype\ltype $}
      \AxiomC{$\typctxtwo \vdashp\mstepstwo\estepstwo \tmtwo \hastype \mtype\mplus\mult\atype$}
      \RightLabel{\footnotesize$\ruleAp$}
      \BinaryInfC{$\typctx \uplus \typctxtwo \vdashp{\msteps+\mstepstwo+1}{\esteps+\estepstwo} \tm\tmtwo \hastype \ltype$}
      \DisplayProof
      \\[12pt]
      \AxiomC{$\typctx \vdashp\msteps\esteps  \tm \hastype \ltype$}      
      \RightLabel{\footnotesize$\ruleLam$}
      \UnaryInfC{$\typctx \sm \var \vdashp\msteps\esteps \la\var\tm\hastype \arrowtype {\typctx(\var)} \ltype$}
      \DisplayProof
      &
	      \AxiomC{$\typctx \vdashp\msteps\esteps \tm \hastype  \ltype $}
      \AxiomC{$\typctxtwo \vdashp\mstepstwo\estepstwo \tmtwo \hastype \typctx(\var)\mplus\mult\atype$}
      \RightLabel{\footnotesize$\ruleES$}
      \BinaryInfC{$(\typctx \sm\var) \uplus \typctxtwo \vdashp{\msteps+\mstepstwo}{\esteps+\estepstwo} \tm\esub\var\tmtwo \hastype \ltype$}
      \DisplayProof
    \end{tabular}
    }
\end{array}$
\caption{The silly multi type system.}
\label{fig:silly-types}
\end{figure}

\myparagraph{\cbs Types and Judgements.} 
 \emph{Linear types} and \emph{multi(-sets) types} are defined by mutual induction in \reffig{silly-types}. 
 Note the linear constant $\normal$ used to type abstractions, which are normal terms. For conciseness, sometimes we shorten it to $\ntype$. We shall show that every normalizing term is typable with $\atype$, hence its name. The constant $\atype$ shall also play a role in our quantitative study in \refsect{tight}. The empty multi set $\mult{\ }$ is also noted $\zero$. 
  
  A multi type $\mset{\ltype_1, \dots, \ltype_n}$ has to be intended as a conjunction $\ltype_1 \land \dots \land 
\ltype_n$ of linear types $\ltype_1, \dots, \ltype_n$, for a commutative, associative, non-idempotent conjunction 
$\land$ (morally a tensor $\otimes$), of neutral element $\emptymset$.
The intuition is that a linear type corresponds to a single use of a term $\tm$, and that $\tm$ is typed with a 
multiset 
$\mtype$ of $n$ linear types if it is going to be used (at most) $n$ times, that is, if $\tm$ is part of a larger term $\tmtwo$, then a copy 
of $\tm$ shall end up in evaluation position during the evaluation of $\tmtwo$.

The size $\size\mtype$ of a multi type $\mtype$ is the number of its elements, that is, $\size{\mset{\ltype_1, \dots, \ltype_n}}\defeq n$.

\emph{Judgments} have shape $\typctx \vdashp\msteps\esteps \tm \hastype \ttype$ where $\tm$ is a term, $\msteps$ and $\esteps$ are two natural numbers, $\ttype$ is either a multi type or a 
linear type, and $\typctx$ is a \emph{type context}, \ie, a total function from variables to multi types 
such that  $\dom{\typctx} \defeq \{\var \mid \typctx(\var) \neq \emptymset\}$ is finite, usually written as $\var_1 \hastype \mtype_1, \dots, \var_n \hastype \mtype_n$ (with $n \in 
\nat$) if $\dom{\typctx} \subseteq \{\var_1, \dots, \var_n\}$ and $\typctx(\var_i) = \mtype_{i}$ for $1 \leq i \leq 
n$. 

The indices $\msteps$ and $\esteps$ shall be used for measuring the length of evaluation sequences via type derivations, namely to measure the number of $\towm$ and $\towe$ steps. There is no index for $\towgcv$ steps in order to stay close to the type systems in Accattoli et al. \cite{DBLP:conf/esop/AccattoliGL19} for \cbn/\cbv/\cbneed, that do not have an index for GC either; the reason being that GC (by value) can be postponed. A quick look to the typing rules shows that $\msteps$ and $\esteps$ are not really 
needed, as $\msteps$ can be recovered as the number of $\app$ rules, and $\esteps$ as the number of $\ax$ rules in the typing derivation producing the judgement. It is however handy to note them~explicitly. 

We write $\typctx \vdash \tm \hastype \ttype$ when the information given by $\msteps$ and $\esteps$ is not relevant.

\myparagraph{Typing Rules.} The abstraction rule $\ruleLam$ uses the notation ${\typctx \sm \var}$ for the type context defined as $\typctx$ on  every variable but possibly $\var$, for which $(\typctx\sm\var)(\var)=\zero$. It is a compact way to express the rule in both the cases $\var\in\dom\typctx$ and $\var\notin\dom\typctx$. 

Rules $\ruleAp$ and $\ruleES$ require the argument to be typed with $\mtype \uplus\mult\atype$, which is necessarily introduced by rule $\ruleMany$, the hypotheses of which are a multi set of derivations, indexed by a possibly empty set $I$. When $I$ is empty, the rule has one premises, the one for $\atype$.


\myparagraph{Type Derivations.} We write $\tderiv \exder \typctx\vdash \tm \hastype \ltype$ if $\tderiv$ is a (\emph{type}) \emph{derivation} (\ie a tree constructed using the rules in \reffig{silly-types}) of final judgment $\typctx \vdash \tm \hastype \ltype$.

 The \emph{size} $\size\tderiv$ of a  derivation $\tderiv \exder \Deri[(\msteps, \esteps)] {\typctx}{\tm}{\type}$ is the number of non-$\ruleMany$ rules, which is always greater or equal to the sum of the indices, that is, $\size\tderiv \geq\msteps + \esteps$.
 
For listing $n$ identical premises of conclusion $\typctx\vdash \tm \hastype \type$ of a $\ruleMany$ rule, we sometimes use the abbreviation $(\typctx\vdash \tm \hastype \type)_{i=1,\ldots,n}$ on the conclusion of a single derivation. For instance, we denote the derivation having a $\ruleMany$ rule with $n+1$ premises, $n$ of which are all axioms of conclusion $\vartwo\hastype\mult\type \vdash \vartwo \hastype \type$, as follows:
\begin{center}
		\AxiomC{}  
	\RightLabel{$\ruleAx$}
	\UnaryInfC{$(\vartwo\hastype\mult\type \vdash \vartwo \hastype \type)_{i=1,\ldots,n}$}
	\AxiomC{}  
	\RightLabel{$\ruleAx$}
	\UnaryInfC{$\vartwo\hastype\mult\atype \vdash \vartwo \hastype \atype$}
	\RightLabel{$\ruleMany$}
	\BinaryInfC{$\vartwo\hastype\mult{\type^n,\atype} \vdash \vartwo \hastype \mset{\type^n,\atype}$}  
	\DisplayProof
\end{center}

\myparagraph{Further Technicalities about Types.} The type context $\typctx$ is \emph{empty} if $\dom{\typctx} = \emptyset$, and we write $\vdash \tm \hastype \ltype$ when $\typctx$ is empty.  
\emph{Multi-set sum} $\mplus$ is extended to type contexts point-wise,
\ie\  $(\typctx \mplus \typctxtwo)(\var) \defeq \typctx(\var) \mplus \typctxtwo(\var)$ for each variable $\var$.
This notion is extended to a finite family of type contexts as expected, 
in particular $\bigmplus_{i \in J\!} \typctx_i$ is the empty context  when $J = \emptyset$.
Given two type contexts $\typctx$ and $\typctxtwo$ such that $\dom{\typctx} \cap \dom{\typctxtwo} = \emptyset$, the 
type 
context $\typctx, \typctxtwo$ is defined by $(\typctx, \typctxtwo)(\var) \defeq \typctx(\var)$ if $\var \in 
\dom{\typctx}$, $(\typctx, \typctxtwo)(\var) \defeq \typctxtwo(\var)$ if $\var \in \dom{\typctxtwo}$, and $(\typctx, 
\typctxtwo)(\var) \defeq \emptymset$ otherwise.
Note that $\typctx, \var \hastype \emptymset = \typctx$, where we implicitly assume $\var \notin \dom{\typctx}$. 

\myparagraph{Relevance.}
    Note that no weakening is allowed in axioms. An easy induction then shows: 
\begin{lem}[Type contexts and variable occurrences]
Let $\tderiv \exder \Deri[(\msteps, \esteps)] \typctx \tm \ltype$ be a  derivation. Then $\ofv\tm \subseteq \dom\typctx \subseteq \fv\tm$.\label{l:typctx-varocc-tm}
\end{lem}

\begin{proof}
Straightforward induction on $\tderiv$. Note that, if one excludes rule $\ruleLam$, the management of type contexts by the rules of the type systems mimics exactly the definition of $\ofv\tm$, in particular in rules $\ruleAp$ and $\ruleES$ one can always apply the \ih to the right premise because of the $\atype$ requirement, which forbids rule $\ruleMany$ with zero premises. Rule $\ruleLam$ is what might allow $\ofv\tm \subsetneq\dom\typctx$, as one might also type under abstraction, thus extending $\dom\typctx$ beyond $\ofv\tm$. The part $\dom\typctx \subseteq \fv\tm$ is due to the absence of weakening in the type system.\qedhere
\end{proof}

\reflemma{typctx-varocc-tm} implies that  derivations of closed terms have empty type context.
Note that free variables of $\tm$ might not be in $\dom{\typctx}$, 
if they only occur in abstractions typed with $\ruleAxLam$.

\myparagraph{Typing the Silly Extra Copy.} The empty multi type $\zero$ is the type for variables that do not occur or whose occurrences are unreachable by weak evaluation. A typical example is $\la\var\vartwo$, that can be typed only with arrow types of the form $\tarrow\zero\ltype$ (plus of course with $\atype$), because of \reflemma{typctx-varocc-tm}. Note that in the silly system every term---even diverging ones---can be typed with $\zero$ by rule $\ruleMany$ (taking 0 premises). In \cbn, an argument for $\la\var\vartwo$ would only need to be typed with $\zero$, as typability with $\zero$ means that the term shall be erased. In the silly system, instead, the application rule $\ruleAp$ requires an argument of $\la\var\vartwo$ to additionally be typed with $\mult\atype$, because of the '$\uplus\mult\atype$' requirement for arguments (and ESs), forcing the argument to be $\tow$ normalizing and capturing silly erasures.

Now, note that the modification $\uplus\mult\atype$ at work in the silly system concerns \emph{all} arguments and ESs, not only those associated to $\zero$. This is what corresponds at the type level to the reduction of the \emph{silly extra copy} of \emph{every} argument/ES (out of abstractions).
  
\myparagraph{Relationship with  the Literature.} Our system is essentially the one used by Kesner and Ventura to measure \cbn strong normalization in the LSC with respect to strong evaluation, i.e. possibly under abstraction \cite{DBLP:conf/ifipTCS/KesnerV14}. There are two differences. First, they tweak the $\ruleAp$ and $\ruleES$ \cbn rules with $\uplus\mult\ltype$ rather than with $\uplus\mult\atype$, that is, they allow an arbitrary linear type for the additional copy to be evaluated. Second, they do not have rule $\ruleAxLam$, because their evaluation is strong. Thus,  $\Id\esub\var{\la\vartwo\Omega}$ is not typable in \cite{DBLP:conf/ifipTCS/KesnerV14}, while here it is.

The \cbv system by Manzonetto et al. \cite{DBLP:journals/fuin/ManzonettoPR19} is similar to ours in that it tweaks the \cbn system and requires multi-sets to be non-empty.

\section{The Weak Calculus, Types, and Strong Normalization}
\label{sect:weak-ssc}
Here, we show that silly multi types characterize strong normalization in the weak SSC. The proof technique is standard: we prove correctness (\ie $\tm$ typable implies $\tm\in\sn\wsym$) via quantitative subject reduction, and completeness (\ie $\tm\in\sn\wsym$ implies $\tm$ typable) via subject expansion and typability of normal forms.
Note that SN in our weak case is simpler than SN in strong calculi, since here erasing steps cannot erase divergence. Actually, they cannot erase \emph{any} step, as guaranteed by strict commutation in \refsect{rewriting_properties}. At the end of the section, we shall indeed obtain \emph{uniform normalization}, i.e. that weak and strong normalization for $\tow$ coincide.

\myparagraph{Correctness.} The quantitative aspect of the following subject reduction property is that the indices decrease with every non-erasing step, and erasing steps decrease the size of the typing derivation. The proof is standard. The point about GC by value is where the features of the new system play a role. The proof of the exponential case relies on the following splitting and linear substitution lemmas, that are typical of subject reduction proofs for multi types.

\begin{lem}[Splitting multi-sets with respect to derivations]
\label{l:types-splitting-multisets}
Let $\tm$ be a term, $\tderiv \exder  \tyjp{(\msteps,\esteps)}{\tm}{\typctx}{\mtype}$ a derivation, and $\mtype = \mtypetwo \mplus \mtypethree$  a splitting. Then there exist two derivations 
\begin{itemize}
  \item $\tderiv_{\mtypetwo} \exder  \tyjp{(\msteps_{\mtypetwo},\esteps_{\mtypetwo})}{\tm}{\typctx_{\mtypetwo}}{\mtypetwo}$, and
  \item $\tderiv_{\mtypethree} \exder  \tyjp{(\msteps_{\mtypethree},\esteps_{\mtypethree})}{\tm}{\typctx_{\mtypethree}}{\mtypethree}$ 
\end{itemize}
 such that:
	\begin{itemize}
	\item \emph{Type contexts}: $\typctx = \typctx_{\mtypetwo} \mplus \typctx_{\mtypethree}$,
	\item \emph{Indices}: $\msteps = \msteps_{\mtypetwo} + \msteps_{\mtypethree}$ and $\esteps = \esteps_{\mtypetwo} + \esteps_{\mtypethree}$.
	\end{itemize}
\end{lem}

\begin{proof}
The last rule of $\tderiv \exder  \tyjp{(\msteps,\esteps)}{\tm}{\typctx}{\mtype}$ can only be $\ruleMany$, thus it is enough to re-group its hypotheses according to $\mtypetwo$ and $\mtypethree$.\qedhere
\end{proof}

\begin{lem}[Linear substitution]
\label{l:open-linear-substitution}
If $\tderiv \exder \Deri[(\msteps, \esteps)]{\typctx, \var \hastype \mtype}{\octxfp{\var}}{\ttype}$ 
then there is a splitting $ \mtype =  \mtypetwo  \mplus \mtypethree$ with $\mtypetwo\neq\zero$ such that for every derivation $\tderivtwo \exder \Deri[(\mstepstwo, \estepstwo)]{\typctxtwo}{\tm}{\mtypetwo}$ with $\var\notin\dom\typctxtwo$ there is a derivation 
$\tderiv' \exder  \Deri[(\msteps+\mstepstwo, \esteps + \estepstwo - \size\mtypetwo)]{\typctx \mplus \typctxtwo, \var \hastype \mtypethree}{\octxfp{\tm}}{\ttype}$. \label{l:open-linear-substitution}
\end{lem}

\begin{proof}
See the technical report \cite{accattoli2024mirroring}.
\end{proof}

\begin{prop}[Quantitative subject reduction for Weak SSC]
  Let $\tderiv\exder  \Deri[(\msteps, \esteps)]\typctx{\tm}\ltype$ be a  derivation.  \label{prop:open-subject-reduction} 
  \begin{enumerate}
    \item \emph{Multiplicative}:  if $\tm\towm\tmtwo$ then $\msteps\geq 1$ and
  there exists
  $\tderivtwo\exder  \Deri[(\mstepstwo, \esteps)]\typctx{\tmtwo}\ltype$ with $\msteps >\mstepstwo$.
  
  \item \emph{Exponential}: 
 if $\tm\towe\tmtwo$ then $\esteps\geq 1$ and
  there exists
  $\tderivtwo\exder  \Deri[(\msteps, \estepstwo)]\typctx{\tmtwo}\ltype$ with $\esteps >\estepstwo$. 

  \item \emph{GC by value}: 
 if $\tm\towgcv\tmtwo$ then 
  there exists
  $\tderivtwo\exder  \Deri[(\msteps, \esteps)]\typctx{\tmtwo}\ltype$ with $\size\tderiv >\size\tderivtwo$. 
  \end{enumerate}
\end{prop}  

\begin{proof}
All three points are by induction on the context closing the root step of the point. The first two points are standard. For the first one, see the technical report \cite{accattoli2024mirroring}. The root case of the second one is given below, to illustrate the use of the substitution lemma above. The interesting case is the root case $\rtogcv$ of the third point, that shows both the use of the silly extra copy and the management of the substitution context $\lctx$; the contextual closure  of $\rtogcv$ is then standard.
\begin{enumerate}
\item[(2)] \emph{Exponential.} We give only the details of the root case because the contextual cases are standard and can be found in the technical report \cite{accattoli2024mirroring}. 
\begin{itemize}
\item \emph{Root step}, \ie  $\tm =  \octxfp{\var} \esub\var\tmthree \rtoep\octx \octxfp{\tmthree} \esub\var\tmthree$. Given a derivation $\tyjp{(\msteps',\esteps')}{\octxfp{\var}}{\typctx'}{\ltype}$, we have $\typctx'(\var)\neq\zero$ by \reflemma{typctx-varocc-tm}, which can be applied because $\var\in\ofv{\octxfp{\var}}$. Then $\tderiv$ has the following form:
\[
\AxiomC{ $\tderiv_{\octxfp{\var}} \exder \tyjp{(\msteps_{\octx},\esteps_{\octx})}{\octxfp{\var}}{\typctx_{\octx}, \var:\mtype}{\ltype}$ }
\AxiomC{ $\tyjp{(\msteps_\tmthree,\esteps_\tmthree)}{\tmthree}{\typctxtwo}{\mtype\mplus\mult\ntype}$ }
\RightLabel{$\ruleES$}
\BinaryInfC{ $\tyjp{(\msteps_{\octx} + \msteps_\tmthree, \esteps_{\octx} + \esteps_\tmthree)}{\octxfp{\var} \esub\var\tmthree}{\typctx_{\octx} \mplus \typctxtwo }{\ltype}$ }
\DisplayProof
\]
		with $\mtype\neq\zero$, $\typctx = \typctx_{\octx}  \mplus  \typctxtwo$,  $\msteps = \msteps_{\octx} + \msteps_\tmthree$,  and $\esteps = \esteps_{\octx} + \esteps_\tmthree$. 
		
	Let $\mtype = \mtypetwo \mplus  \mtypethree$ with $\mtypetwo\neq\zero$  be the splitting of $\mtype$ given by the linear substitution lemma (\reflemma{open-linear-substitution}) applied to $\tderiv_{\octxfp{\var}}$. By multi-sets splitting (\reflemma{types-splitting-multisets}) applied to $\mtype\mplus\mult\ntype$, there exist two derivations: 
\begin{enumerate}
\item  $\tderivtwo_\mtypetwo \exder \Deri[(\msteps_\tmthree^{\mtypetwo}, \esteps_\tmthree^{\mtypetwo})]{\typctxtwo^{\mtypetwo}}{\tmthree}{\mtypetwo}$ and

\item  $\tderivtwo_\mtypethree \exder \Deri[(\msteps_\tmthree^\mtypethree, \esteps_\tmthree^\mtypethree)]{\typctxtwo^\mtypethree}{\tmthree}{\mtypethree\mplus\mult\ntype}$.
\end{enumerate}
such that $\typctxtwo = \typctxtwo^\mtypetwo \mplus \typctxtwo^\mtypethree$, $\msteps_\tmthree = \msteps_\tmthree^\mtypetwo + \msteps_\tmthree^\mtypethree$,  and $\esteps_\tmthree = \esteps_\tmthree^\mtypetwo + \esteps_\tmthree^{\mtypethree}$.

Now, by applying again the linear substitution lemma to $\tderiv_{\octxfp{\var}}$ with respect to $\tderivtwo_\ltypetwo$, we obtain a derivation:
\[
\tderiv_{\octxfp{\tmthree}} \exder  \Deri[(\msteps_{\octx}+\msteps_\tmthree^\mtypetwo, \esteps_{\octx} + \esteps_\tmthree^\mtypetwo-\size\mtypetwo)]{\var:\mtypethree;
  \typctx_{\octx} \mplus \typctxtwo^\mtypetwo}{\octxfp{\tmthree}}{\ltype}
  \]

Then $\tderivtwo$ is built as follows:
\[
\AxiomC{ $\Deri[(\msteps_{\octx}+\msteps_\tmthree^\mtypetwo, \esteps_{\octx} + \esteps_\tmthree^\mtypetwo-\size\mtypetwo)]{\var:\mtypethree;
  \typctx_{\octx} \mplus \typctxtwo^\mtypetwo}{\octxfp{\tmthree}}{\ltype}$ }
\AxiomC{ $\Deri[(\msteps_\tmthree^\mtypethree, \esteps_\tmthree^\mtypethree)]{\typctxtwo^\mtypethree}{\tmthree}{\mtypethree\mplus\mult\ntype}$ }
\RightLabel{$\ruleES$}
\BinaryInfC{ $\tyjp{(\msteps_{\octx} + \msteps_\tmthree^\mtypetwo + \msteps_\tmthree^\mtypethree, \esteps_{\octx} + \esteps_\tmthree^\mtypetwo+ \esteps_\tmthree^\mtypethree - \size\mtypetwo)}{\octxfp{\tmthree} \esub\var\tmthree}{\typctx_{\octx} \mplus \typctxtwo^\mtypetwo \mplus \typctxtwo^\mtypethree}{\ltype}$ }
\DisplayProof
\]
	Now, note that the last judgement is in fact
	$$ \tyjp{(\msteps_{\octx} + \msteps_\tmthree, \esteps_{\octx} + \esteps_\tmthree - \size\mtypetwo)}{\octxfp{\tmthree} \esub\var\tmthree}{\typctx_{\octx} \mplus \typctxtwo}{\ltype} $$
	which in turn is 
	$$ \tyjp{(\msteps, \esteps - \size\mtypetwo)}{\octxfp{\tmthree} \esub\var\tmthree}{\typctx}{\ltype} $$
	as required, in particular $\estepstwo \defeq \esteps - \size\mtypetwo$ is strictly smaller than $\esteps$ because $\mtypetwo\neq\zero$. 
\end{itemize}

\item[(3)] \emph{Root case of GC by value}. We have  $\tm =  \tmthree \esub\var{\lctxp\val} \rtogcv \lctxp\tmthree=\tmtwo$. By induction on $\lctx$. Cases:
\begin{enumerate}
\item $\lctx=\ctxhole$. Then $\tm =  \tmthree \esub\var\val \rtogcv \tmthree=\tmtwo$ with $\var\notin\fv\tmthree$. By \reflemma{typctx-varocc-tm}, given a derivation $\tyjp{(\msteps,\esteps)}{\tmthree}{\typctx_\tmthree}{\ltype}$ one has $\typctx_\tmthree(\var)=\zero$. Then, the derivation $\tderiv$ for $\tm$ has the following shape:
\[
\AxiomC{ $\tderiv_\tmthree \exder \tyjp{(\msteps,\esteps)}{\tmthree}{\typctx}{\ltype}$ }
    \AxiomC{}
    \RightLabel{\footnotesize$\ruleAxLam$}
    \UnaryInfC{ $\tyjp{(0,0)}{\val}{}{\ntype}$ }
	\RightLabel{$\ruleMany$}
	\UnaryInfC{ $\tyjp{(0,0)}{\val}{}{\mult\ntype}$ }
\RightLabel{$\ruleES$}
\BinaryInfC{ $\tyjp{(\msteps, \esteps)}{\tmthree \esub\var\val}{\typctx }{\ltype}$ }
\DisplayProof
\]
Then $\tderivtwo \defeq \tderiv_\tmthree$ verifies the statement, since $\size\tderivtwo = \size\tderiv-2$ (the $\ruleMany$ rule does not count for the size).

\item $\lctx=\lctxtwo\esub\vartwo\tmfour$. Then $\tm =  \tmthree \esub\var{\lctxtwop\val\esub\vartwo\tmfour} \rtogcv \lctxtwop\tmthree\esub\vartwo\tmfour=\tmtwo$ with $\var\notin\fv\tmthree$. By \reflemma{typctx-varocc-tm}, given a derivation $\tyjp{(\msteps,\esteps)}{\tmthree}{\typctx_\tmthree}{\ltype}$ one has $\typctx_\tmthree(\var)=\zero$. Then, the derivation $\tderiv$ for $\tm$ has the following shape:
\[
\AxiomC{ $\tyjp{(\msteps_\tmthree,\esteps_\tmthree)}{\tmthree}{\typctx_\tmthree}{\ltype}$ }
    \AxiomC{ $\tyjp{(\mstepstwo,\estepstwo)}{\lctxtwop\val}{\typctxtwo}{\ntype}$ }
    \AxiomC{ $\tyjp{(\msteps_\tmfour,\esteps_\tmfour)}{\tmfour}{\typctxtwo_\tmfour}{\typctxtwo(\vartwo)\uplus\mult\ntype}$ }
    \RightLabel{$\ruleES$}
    	\BinaryInfC{ $\tyjp{(\mstepstwo+\msteps_\tmfour,\estepstwo+\esteps_\tmfour)}{\lctxtwop\val\esub\vartwo\tmfour}{\typctxtwo\uplus\typctxtwo_\tmfour}{\ntype}$ }
	\RightLabel{$\ruleMany$}
	\UnaryInfC{ $\tyjp{(\mstepstwo+\msteps_\tmfour,\estepstwo+\esteps_\tmfour)}{\lctxtwop\val\esub\vartwo\tmfour}{\typctxtwo\uplus\typctxtwo_\tmfour}{\mult\ntype}$ }
\RightLabel{$\ruleES$}
\BinaryInfC{ $\tyjp{(\msteps_\tmthree + \mstepstwo+\msteps_\tmfour, \esteps_\tmthree +  \estepstwo+\esteps_\tmfour)}{\tmthree \esub\var{\lctxtwop\val\esub\vartwo\tmfour}}{\typctx_\tmthree \uplus \typctxtwo\uplus\typctxtwo_\tmfour }{\ltype}$ }
\DisplayProof
\]
with $\typctx = \typctx_\tmthree \uplus \typctxtwo\uplus\typctxtwo_\tmfour$, $\msteps = \msteps_\tmthree + \mstepstwo+\msteps_\tmfour$, and $\esteps = \esteps_\tmthree +  \estepstwo+\esteps_\tmfour$.
Then consider the following derivation $\tderiv'$, where the substitution $\esub\vartwo\tmfour$ has been removed and of size $\size{\tderiv'}=\size\tderiv-1$:
\[\tderiv'\defeq
\AxiomC{ $ \tyjp{(\msteps_\tmthree,\esteps_\tmthree)}{\tmthree}{\typctx_\tmthree}{\ltype}$ }
    \AxiomC{ $\tyjp{(\mstepstwo,\estepstwo)}{\lctxtwop\val}{\typctxtwo}{\ntype}$ }
	\RightLabel{$\ruleMany$}
	\UnaryInfC{ $\tyjp{(\mstepstwo,\estepstwo)}{\lctxtwop\val}{\typctxtwo}{\mult\ntype}$ }
\RightLabel{$\ruleES$}
\BinaryInfC{ $\tyjp{(\msteps_\tmthree + \mstepstwo, \esteps_\tmthree +  \estepstwo)}{\tmthree \esub\var{\lctxtwop\val}}{\typctx_\tmthree \uplus \typctxtwo }{\ltype}$ }
\DisplayProof
\]
By \ih, we obtain  a derivation $\tderivtwo' \exder \tyjp{(\msteps_\tmthree + \mstepstwo, \esteps_\tmthree +  \estepstwo)}{\lctxtwop\tmthree }{\typctx_\tmthree \uplus \typctxtwo }{\ltype}$ such that $\size{\tderivtwo'}<\size{\tderiv'}$. Finally, by re-introducing $\esub\vartwo\tmfour$ we obtain the derivation $\tderivtwo$ satisfying the statement and such that $\size\tderivtwo = \size{\tderivtwo'}+1<_{\ih}\size{\tderiv'}+1=\size\tderiv$:
\begin{equation*}
\raisebox{\depth}{$
\AxiomC{ $\tderivthree \exder \tyjp{(\msteps_\tmthree + \mstepstwo, \esteps_\tmthree +  \estepstwo)}{\lctxtwop\tmthree }{\typctx_\tmthree \uplus \typctxtwo }{\ltype}$ }
    \AxiomC{ $\tyjp{(\msteps_\tmfour,\esteps_\tmfour)}{\tmfour}{\typctxtwo_\tmfour}{\typctxtwo(\vartwo)\uplus\mult\ntype}$ }
    \RightLabel{$\ruleES$}
\BinaryInfC{ $\tyjp{(\msteps_\tmthree + \mstepstwo+\msteps_\tmfour, \esteps_\tmthree +  \estepstwo+\esteps_\tmfour)}{\lctxtwop\tmthree \esub\vartwo\tmfour}{\typctx_\tmthree \uplus \typctxtwo\uplus\typctxtwo_\tmfour }{\ltype}$ }
\DisplayProof
$}
\tag*{\qedhere}
\end{equation*}
\end{enumerate}
\end{enumerate}
\end{proof}

Note that $\towgcv$ steps do not change the $\msteps$ and $\esteps$ indices. This is a consequence of rules $\ruleAp$ and $\ruleES$ having been modified for \cbs with $\uplus\mult\atype$ (and not with $\uplus\mult\ltype$ as in \cite{DBLP:conf/ifipTCS/KesnerV14}).

\begin{thm}[Weak SSC correctness]
  Let $\tm$ be a term.
  If $\tderiv \exder   \Deri[(\msteps, \esteps)] {\typctx}{\tm}{\ltype}$ then $\tm\in\sn\wsym$. Moreover, if $\deriv:\tm \tow^*\ntm$ is a normalizing sequence then $\sizep\deriv{\wsym\msym}\leq \msteps$ and $\sizep\deriv{\wesym}\leq \esteps$. \label{thm:weak-correctness}
\end{thm}
\begin{proof} 
By lexicographic induction on $(\msteps+\esteps, \size\tderiv)$ and case analysis on whether $\tm$ reduces or not. If $\tm$ is $\tow$-normal then the statement trivially holds. If $\tm$ is not $\tow$-normal we show that all its reducts are SN for $\tow$, that is, $\tm$ is SN. If $\tm \towm \tmthree$ then by quantitative subject reduction (\refprop{silly-subject-reduction}) there is a derivation $\tderivtwo \exder  \Deri[(\mstepstwo, \esteps)] {\typctx}{\tmthree}{\ltype}$ with $\mstepstwo<\msteps$. By \ih, $\tmthree$ is SN.   If $\tm \towe \tmthree$ or $\tm \towgcv \tmthree$ we reason similarly, looking at the $\esteps$ index for $\towe$ and to the size $\size\tderivtwo$ of the derivation for $\towgcv$.
  The \emph{moreover} part follows from the fact that $\towm$ (resp. $\towe$) steps strictly decrease $\msteps$ (resp. $\esteps$).
\end{proof}

\myparagraph{Completeness.} We first deal with typability of weak normal forms. The two points of the next proposition are proved by mutual induction. The second point is stronger, as it has a universal quantification about linear types, crucial for the induction to go through. 

\begin{prop}[Weak normal forms are typable]
\label{prop:sillyopen-nfs-exists}
\hfill
\begin{enumerate}
\item Let $\ans$ be a weak answer. Then there exists $\tderiv \exder   \Deri[(0, 0)] {\typctx}{\ans}{\atype}$ with $\dom\typctx=\ofv\ans$.
\item Let $\itm$ be an inert term. Then for any linear type $\ltype$ there exist a type context $\typctx$ such that $\dom{\typctx}=\ofv\itm$ and a derivation $\tderiv \exder   \Deri[(0, 0)] {\typctx}{\itm}{\ltype}$.
\end{enumerate}
\end{prop}

\begin{proof}
The proof is by mutual induction on weak answers $\ans$ and inert terms $\itm$ (see \refprop{weak-nfs} for their grammars).
\begin{enumerate}
\item \emph{Weak answers}. Cases:
\begin{itemize}
\item \emph{Abstraction}, \ie $\ans = \la\var\tm$. Then $\tderiv$ is simply given by:
\[      \AxiomC{}
      \RightLabel{\footnotesize$\ruleAxLam$}
      \UnaryInfC{$\Deri[(0, 0)] {}{ \la\var\tm}{\atype}$}
      \DisplayProof
\]

\item \emph{Substitution of inert terms on answers}, \ie $\ans = \anstwo \esub\var\itm$ with $\var\notin\ofv\anstwo$. By \ih on $\anstwo$ (Point 1), there is a derivation $\tderiv_{\anstwo} \exder   \Deri[(0, 0)] {\typctx_{\anstwo}}{\anstwo}{\atype}$ with $\dom\typctx=\ofv\anstwo$. Therefore, $\var\notin\dom\typctx$. By \ih on $\itm$ (Point 2) with respect to $\ltype\defeq\atype$, there is derivation $\tderiv_{\itm} \exder   \Deri[(0, 0)] {\typctx_\itm}{\itm}{\atype}$ with $\dom{\typctx_\itm}=\ofv\itm$. Then $\tderiv$ is given by:
\[
	      \AxiomC{$\tderiv_{\anstwo} \exder   \Deri[(0, 0)] {\typctx_{\anstwo}}{\anstwo}{\atype}$}
      \AxiomC{$\tderiv_{\ansthree} \exder   \Deri[(0, 0)] {\typctx_\itm}{\ntm}{\atype}$}
      \RightLabel{\footnotesize$\ruleES$}
      \BinaryInfC{$\Deri[(0, 0)] {\typctx_{\anstwo} \mplus \typctx_\itm}{\anstwo\esub\var\ntm}{\atype}$}
      \DisplayProof
\]
Note that $\dom{\typctx_{\anstwo} \mplus \typctx_\itm} = \dom{\typctx_{\anstwo}} \cup \dom{\typctx_\itm} =_{\ih} \ofv\anstwo \cup \ofv\itm = \ofv{\anstwo\esub\var\itm}$ (because $\var\notin\ofv\anstwo$ by hypothesis).

\item \emph{Substitution of answers on answers}, \ie $\ans = \anstwo \esub\var\ansthree$ with $\var\in\fv\anstwo\setminus\ofv\anstwo$. By \ih on $\anstwo$ (Point 1), there is a derivation $\tderiv_{\anstwo} \exder   \Deri[(0, 0)] {\typctx_{\anstwo}}{\anstwo}{\atype}$ with $\dom\typctx=\ofv\anstwo$. Therefore, $\var\notin\dom\typctx$. By \ih on $\ansthree$ (Point 1), there is derivation $\tderiv_{\ansthree} \exder   \Deri[(0, 0)] {\typctx_{\ansthree}}{\ansthree}{\atype}$ with $\dom{\typctx_{\ansthree}}=\ofv\ansthree$. Then $\tderiv$ is given by:
\[
	      \AxiomC{$\tderiv_{\anstwo} \exder   \Deri[(0, 0)] {\typctx_{\anstwo}}{\anstwo}{\atype}$}
      \AxiomC{$\tderiv_{\ansthree} \exder   \Deri[(0, 0)] {\typctx_{\ansthree}}{\ntm}{\atype}$}
      \RightLabel{\footnotesize$\ruleES$}
      \BinaryInfC{$\Deri[(0, 0)] {\typctx_{\anstwo} \mplus \typctx_{\ansthree}}{\anstwo\esub\var\ansthree}{\atype}$}
      \DisplayProof
\]
Note that $\dom{\typctx_{\anstwo} \mplus \typctx_{\ansthree}} = \dom{\typctx_{\anstwo}} \cup \dom{\typctx_{\ansthree}} =_{\ih} \ofv\anstwo \cup \ofv\ansthree = \ofv{\anstwo\esub\var\ansthree}$ (because $\var\notin\ofv\anstwo$ by hypothesis).

\end{itemize}

\item \emph{Inert terms}. Cases:
\begin{itemize}
\item \emph{Variable}, \ie $\itm = \var$. Then, for any linear type $\ltype$, $\tderiv$ is simply given by:
\[      \AxiomC{}
      \RightLabel{\footnotesize$\ruleAx$}
      \UnaryInfC{$\Deri[(0, 0)] {\var \hastype \mset\ltype}{ \var}{\ltype}$}
      \DisplayProof
\]

\item \emph{Application}, \ie $\itm = \itmtwo \ntm$. By \ih on $\ntm$ (either Point 1 or Point 2 with respect to $\ltype\defeq \atype$), there is a derivation $\tderiv_{\ntm} \exder   \Deri[(0, 0)] {\typctx_\ntm}{\ntm}{\atype}$ with $\dom{\typctx_\ntm}=\ofv\ntm$. By \ih (Point 2) on $\itmtwo$ with respect to the linear type $\tarrow\zero\ltype$, there is a derivation $\tderiv_{\itmtwo} \exder   \Deri[(0, 0)] {\typctx_{\itmtwo}}{\itmtwo}{\tarrow\zero\ltype}$ with $\dom{\typctx_{\itmtwo}}=\ofv\itmtwo$. Then $\tderiv$ is given by:
\[
	      \AxiomC{$\tderiv_{\itmtwo} \exder   \Deri[(0, 0)] {\typctx_{\itmtwo}}{\itmtwo}{\tarrow\zero\ltype}$}
      \AxiomC{$\tderiv_{\ntm} \exder   \Deri[(0, 0)] {\typctx_\ntm}{\ntm}{\atype}$}
      \RightLabel{\footnotesize$\ruleAp$}
      \BinaryInfC{$\Deri[(0, 0)] {\typctx_{\itmtwo} \mplus \typctx_\ntm}{\itmtwo \ntm}{\ltype}$}
      \DisplayProof
\]
Note that $\dom{\typctx_{\itmtwo} \mplus \typctx_\ntm} = \dom{\typctx_{\itmtwo}} \cup \dom{\typctx_\ntm} =_{\ih} \ofv\itmtwo \cup \ofv\ntm = \ofv{\itmtwo\ntm}$.

\item \emph{Substitution of inert terms on inert terms}, \ie $\itm = \itmtwo \esub\var\itmthree$ with $\var\notin\ofv\itmtwo$. By \ih (Point 1) with respect to $\ltype$, there is a derivation $\tderiv_{\itmtwo} \exder   \Deri[(0, 0)] {\typctx_{\itmtwo}}{\itmtwo}{\ltype}$ with $\dom{\typctx_{\itmtwo}}=\ofv\itmtwo$. Therefore, $\var\notin\dom{\typctx_{\itmtwo}}$. By \ih on $\itmthree$ (Point 2)  with respect to $\ltype\defeq\atype$, there is derivation $\tderiv_{\itmthree} \exder   \Deri[(0, 0)] {\typctx_{\itmthree}}{\itmthree}{\atype}$ with $\dom{\typctx_{\itmthree}}=\ofv{\itmthree}$. Then $\tderiv$ is given by:
\[
	      \AxiomC{$\tderiv_{\itmtwo} \exder   \Deri[(0, 0)] {\typctx_{\itmtwo}}{\itmtwo}{\ltype}$}
      \AxiomC{$\tderiv_{\itmthree} \exder   \Deri[(0, 0)] {\typctx_\ntm}{\itmthree}{\atype}$}
      \RightLabel{\footnotesize$\ruleES$}
      \BinaryInfC{$\Deri[(0, 0)] {\typctx_{\itmtwo} \mplus \typctx_{\itmthree}}{\itmtwo\esub\var{\itmthree}}{\ltype}$}
      \DisplayProof
\]
Note that $\dom{\typctx_{\itmtwo} \mplus \typctx_{\itmthree}} = \dom{\typctx_{\itmtwo}} \cup \dom{\typctx_{\itmthree}} =_{\ih} \ofv\itmtwo \cup \ofv\itmthree = \ofv{\itmtwo\esub\var\itmthree}$ (because $\var\notin\ofv\itmtwo$ by hypothesis).

\item \emph{Substitution of answers on inert terms}, \ie $\itm = \itmtwo \esub\var\ans$ with $\var\notin\ofv\itmtwo$. By \ih (Point 1) with respect to $\ltype$, there is a derivation $\tderiv_{\itmtwo} \exder   \Deri[(0, 0)] {\typctx_{\itmtwo}}{\itmtwo}{\ltype}$ with $\dom{\typctx_{\itmtwo}}=\ofv\itmtwo$. Therefore, $\var\notin\dom{\typctx_{\itmtwo}}$. By \ih on $\ans$ (Point 1), there is derivation $\tderiv_{\ans} \exder   \Deri[(0, 0)] {\typctx_\ans}{\ans}{\atype}$ with $\dom{\typctx_\ans}=\ofv\ans$. Then $\tderiv$ is given by:
\[
	      \AxiomC{$\tderiv_{\itmtwo} \exder   \Deri[(0, 0)] {\typctx_{\itmtwo}}{\itmtwo}{\ltype}$}
      \AxiomC{$\tderiv_{\ans} \exder   \Deri[(0, 0)] {\typctx_\ans}{\ans}{\atype}$}
      \RightLabel{\footnotesize$\ruleES$}
      \BinaryInfC{$\Deri[(0, 0)] {\typctx_\itmtwo \mplus \typctx_\ans}{\itmtwo\esub\var\ans}{\ltype}$}
      \DisplayProof
\]
Note that $\dom{\typctx_{\itmtwo} \mplus \typctx_\ans} = \dom{\typctx_{\itmtwo}} \cup \dom{\typctx_\ans} =_{\ih} \ofv\itmtwo \cup \ofv\ans = \ofv{\itmtwo\esub\var\ans}$ (because $\var\notin\ofv\itmtwo$ by hypothesis).
\qedhere
\end{itemize}
\end{enumerate}
\end{proof}

The rest of the proof of completeness is dual to the one for correctness, with an anti-substitution lemma (\techrep{in the Appendix}\camerar{in the technical report \cite{accattoli2024mirroring}}) needed for subject expansion for $\towe$, the proof of which is specular to the one for subject expansion. We omit the indices because for completeness they are irrelevant.

	\begin{prop}[Subject expansion for Weak SSC]
		Let $\tderiv\exder  \Deri\typctx{\tmtwo}\ltype$ be a  derivation.  
If $\tm\tow\tmtwo$ then
			there exists a  derivation 
			$\tderivtwo\exder  \Deri\typctx{\tm}\ltype$.\label{prop:open-subject-expansion} 
	\end{prop}  


\begin{thm}[Weak SSC completeness]
	Let $\tm$ be a term.
	If $\tm \tow^*\ntm$ and $\ntm$ is a weak normal form then there exists $\tderiv \exder   \Deri {\typctx}{\tm}{\atype}$.  \label{thm:open-completeness}
\end{thm}
\begin{proof} 
By induction on $k = \size\deriv$. If $k=0$: then $\tm=\ntm$ and is typable with $\atype$ by \refprop{sillyopen-nfs-exists}. If $k>0$  then $\tm \tow \tmthree \tow^{k-1} \ntm$ for some $\tmthree$ and by \ih there is a derivation $\tderiv'
		\exder \Deri[] {\typctx}{\tmthree}{\ltype}$.
		By subject expansion
		(\refprop{open-subject-expansion}), $\tderiv
		\exder \Deri[] {\typctx}{\tm}{\ltype}$. 
\end{proof}


\begin{cor}[Weak SSC Uniform normalization]
\label{coro:unif-norm}
$\tm$ is weakly $\tow$-normalizing if and only if $\tm$ is strongly $\tow$-normalizing.\label{coro:unif-norm}
\end{cor}
\begin{proof}
The non-obvious direction is $\Rightarrow$. By completeness (\refthm{open-completeness}), $\tm\in\wn\wsym$ implies typability, which in turn, by correctness (\refthm{weak-correctness}), implies $\tm\in\sn\wsym$.
\end{proof}

	\myparagraph{Type Equivalence.} We define the notion of type equivalence induced by the silly multi type system in order to formally show that the type system induces an equational theory in \refthm{silly-conversion-is-included-in-eqcsilly} below.

	\begin{defi}[(Silly) Type Equivalence]
		Two terms are type equivalent $\tm\eqtype\tmtwo$ if they are typable by the same set of typing judgments, that is if
			$\forall \typctx,\ltype, ~\Deri {\typctx}{\tm}{\ltype} \iff \Deri {\typctx}{\tmtwo}{\ltype}$.
	\end{defi}

\begin{thm}[Type equivalence is an equational theory for the SSC]
	\label{thm:silly-conversion-is-included-in-eqcsilly}
	\hfill
	\begin{enumerate}
	\item \label{p:silly-conversion-is-included-in-eqcsilly-invariance} \emph{Invariance}: if $\tm\tow\tmtwo$ then $\tm\eqtype\tmtwo$;
	\item \label{p:silly-conversion-is-included-in-eqcsilly-compatibility} \emph{Compatibility}: if $\tm\eqtype\tmtwo$ then $\ctxp\tm\eqtype\ctxp\tmtwo$ for all contexts $\ctx$.
	\end{enumerate}
\end{thm}

\begin{proof}
	\hfill
	\begin{enumerate}
	\item By subject reduction and expansion of the silly multi types for the weak case (\refprop{open-subject-reduction} and \refprop{open-subject-expansion}) we have that for all $\typctx,\ltype,$ $\typctx\vdash\tm \hastype \ltype$ iff $\typctx\vdash\tmtwo\hastype\ltype$, that is, $\tm\eqtype\tmtwo$.
	
	\item By induction on $\ctx$:
	\begin{itemize}
		\item $\ctx=\ctxhole$: trivial by the hypothesis $\tm\eqtype\tmtwo$.
		
		\item $\ctx=\la\var\ctxtwo$: all derivations of $\ctxp\tm$ and $\ctxp\tmtwo$ must start with the typing rule $\ruleAx_\l$ or the typing rule $\ruleLam$. For $\ruleAx_\l$, it is trivial to conclude. For $\ruleLam$, it follows from the \ih on the premise of the rule.
		
		\item $\ctx=\ctxtwo\tmthree$: all derivations of $\ctxp\tm$ and $\ctxp\tmtwo$ must start with the typing rule $\ruleAp$. It follows from the  \ih on the left premise of the rule $\ruleAp$.
		
		\item Other cases follow the same argument.\qedhere
	\end{itemize}
	\end{enumerate}
\end{proof}

Since the defining property of denotational models is that they induce an equational theory for the underlying calculus,  \refthm{silly-conversion-is-included-in-eqcsilly} can be understood as the fact that silly multi type judgements provide a denotational model for the SSC.
\section{Call-by-Value and Operational Equivalence}
\label{sect:cbv}
Here, we show that the silly multi types characterize \cbv termination as well, 
and infer the operational equivalence of the SSC and \cbv.

%
%

\begin{figure}[t]
\centering
	\begin{tabular}{cc|cc}
		$\arraycolsep=3pt
		\begin{array}{r@{\hspace{.25cm}}rcl}
		\textsc{Terms} & 
		\tm,\tmtwo,\tmthree  &\grameq& \var \mid \la\var\tm \mid  \tm\tmtwo
		\\ 
		\textsc{Values} & \val,\valtwo  &\grameq&  \la\var\tm
		\\[6pt]
		\textsc{\cbv contexts} &  \vctx&   \grameq&  \ctxhole\mid \tm\vctx \mid \vctx\tm
		\end{array}$ 
		
		&&  
		\begin{tabular}{c}
			{\textsc{Rewriting rule}}
			\\
			$\arraycolsep=3pt
			\begin{array}{rll}
			(\la\var\tm)\val & \rtobv & \tm\isub\var{\val}
			\\[6pt]
			\tobv  &\defeq&  \vctxp\rtobv
			\end{array}$
		\end{tabular}
	\end{tabular}
\caption{\label{fig:cbv} The call-by-value $\l$-calculus.}
\end{figure}
\myparagraph{\ccbv without ESs.} We define the \cbv $\l$-calculus in \reffig{cbv}, mostly following the presentation of Dal Lago and Martini \cite{DBLP:journals/tcs/LagoM08}, for which the $\betav$-rule is non-deterministic but diamond (thus trivially uniformly normalizing, see the rewriting preliminaries). The only change with respect to \cite{DBLP:journals/tcs/LagoM08} is that here values are only abstractions, for uniformity with the Weak SSC. We shall consider the weak evaluation of closed terms only, for which there is no difference whether variables are values or not, since free variables cannot be arguments (out of abstractions) anyway. Closed normal forms are exactly the abstractions.

We discuss only the closed case for \cbv because the silly type system is \emph{not} correct for \cbv with open terms. This point is properly explained after the operational equivalence theorem. It is also the reason why we do not present \cbv via the LSC, as also explained after the theorem. The problem with open terms does not hinder the operational equivalence of the SSC and \cbv, because contextual equivalence is based on closed terms only.

\myparagraph{Judgements for \cbv} In this section, the index $\esteps$ of the silly type system does not play any role, because $\tobv$ steps are bound only by the $\msteps$ index. Therefore, we omit $\esteps$ and write $\tderiv \exder  \typctx \Deri[\msteps] {}{\tm}{\ltype}$ instead of $\tderiv \exder  \typctx \Deri[(\msteps,\esteps)] {}{\tm}{\ltype}$.

\myparagraph{Correctness} The proof technique is the standard one. As usual, subject reduction is proved via a substitution lemma \techrep{in the Appendix}\camerar{specified in the technical report \cite{accattoli2024mirroring}}. The index $\msteps$ is used as decreasing measure to prove correctness.

\begin{prop}[Quantitative subject reduction for \ccbv]
  Let $\tm$ be a closed term. If  $\tderiv\exder \Deri[\msteps]{}{\tm}\ltype$ 
  and $\tm \tobv \tmtwo$
  then $\msteps\geq 1$ and
  there exists
  $\tderivtwo\exder \Deri[\mstepstwo] {}{\tmtwo}\ltype$
  with $ \msteps >   \mstepstwo$.$\label{prop:value-subject-reduction}$
\end{prop}

	\begin{thm}[\ccbv correctness]
  Let $\tm$ be a closed term.
  If $\tderiv \exder  \Deri[\msteps] {}{\tm}{\ltype}$ then there are an abstraction $\val$ and a reduction sequence $\deriv \colon 
\tm
  \tobv^* \val$ with $\size\deriv \leq \msteps$.\label{thm:value-correctness}
\end{thm}

\myparagraph{Completeness.} Completeness is also proved in a standard way, omitting the index $\msteps$ because it is irrelevant.
\label{s:cbv-completeness}

\begin{prop}[Subject expansion for \ccbv]
	If 
	$\tderiv\exder \Deri {}{\tmtwo}\ltype$
	and $\tm \tobv \tmtwo$
	then 
	there exists a typing $\tderivtwo$ such that $\tderivtwo\exder \Deri{}{\tm}\ltype$.\label{prop:value-subject-expansion}
\end{prop}

\begin{thm}[\ccbv completeness]
	Let $\tm$ be a closed $\l$-term.
	If there exists a value $\val$ and a reduction sequence $\deriv \colon 
	\tm
	\tobv^* \val$ then $\tderiv \exder  \Deri {}{\tm}{\atype}$.$\label{thm:value-completeness}$
\end{thm}

\myparagraph{Operational Equivalence.} We define abstractly contextual equivalence for a language of terms and arbitrary contexts which are terms with an additional hole construct.
\begin{defi}[Contextual Equivalence] Given a rewriting relation $\to$, we define the associated \emph{contextual equivalence} $\ctxeq$ as follows:
		$\tm \ctxeq \tmp$ if, for all contexts $\ctx$ such that $\ctxp{\tm}$ and $\ctxp\tmp$ are closed terms, 
				$\ctxp\tm$  is weakly $\to$-normalizing if and only if $\ctxp\tmp$ is weakly $\to$-normalizing. 
\end{defi}
Let $\eqcsilly$ and $\eqcvalue$ be the contextual equivalences for $\tow$ and $\tobv$. The next section shall show that $\eqcsilly$ can equivalently be defined using the \cbs strategy.

\begin{thm}[Operational equivalence of \cbs and \cbv]
	On $\l$-terms, $\tm\eqcsilly\tmtwo$ if and only if $\tm\eqcvalue \tmtwo$.\label{thm:op-equiv-cbs-cbv}
\end{thm}
\begin{proof}
On closed $\l$-terms, both $\tow$-termination and $\tobv$-termination are equivalent to typability in the silly system (Theorems \ref{thm:weak-correctness} and \ref{thm:open-completeness} for $\tosi$ and Theorems \ref{thm:value-correctness} and \ref{thm:value-completeness} for $\tobv$). Thus the contextual equivalences coincide.
\end{proof}

\myparagraph{Call-by-Silly Helps to Prove Contextual Equivalence.} As the last theorem says, contextual equivalences induced by \cbv and \cbs coincide. Reductions do differ, and \cbs reduction sometimes provides a way to prove \cbv contextual equivalence in cases where \cbv reduction does not. Consider the following four different terms, where $\itm$ could be any normal form that is not of the shape $\sctxp\val$, for example $\itm=\vartwo\Id$:

\begin{center}
	\begin{tabular}{c@{\hspace{1cm}}c@{\hspace{1cm}}c@{\hspace{1cm}}c}
	$(\la\var\var\var)\,\itm$ & $(\la\var\var\itm)\,\itm$ & $(\la\var\itm\itm)\,\itm$ & $\itm\itm$
\end{tabular}
\end{center}

These four terms can intuitively be seen as \cbv contextually equivalent, as we now outline. When one of these terms is plugged in a closing context $\ctx$, reduction shall provide substitutions on $\itm$ making it either converge to a value $\val$ or diverge. If it diverges, so will the four terms. If it converges to $\val$, then all four terms will reduce to $\val\val$. This reasoning however cannot easily be made formal.

The easiest way to prove that two terms are contextually equivalent for a reduction $\rightarrow_r$ is to prove that they are related by $=_r$, the smallest equivalence relation including $\rightarrow_r$. We shall now see how the silly calculus helps in equating more terms (than \cbv) with its reduction.

\emph{The four terms above are not $=_{\betav}$-related.} First, note that the four terms are all $\tobv$-normal (assuming that $\itm$ is $\tobv$-normal). If they were $=_{\betav}$-related, then by the Church-Rosser property they should have a common reduct. As the four terms are syntactically different normal forms, they cannot be equated by $=_{\betav}$.

\emph{The first three terms are $=_\wsym$-related.} The first three terms rewrite in the SSC to $\itm\itm\esub\var\itm$, which is why 	$(\la\var\var\var)\,\itm=_\wsym(\la\var\var\itm)\,\itm=_\wsym(\la\var\itm\itm)\,\itm=_\wsym\itm\itm\esub\var\itm$. 

Unfortunately, the fourth term $\itm\itm$ is a $\tow$-normal form and hence $\itm\itm\neq_\wsym\itm\itm\esub\var\itm$. Thus, not all \cbv contextually equivalent terms are equated by the silly calculus, not even when the difference amounts to the duplication of a non-value term. The issue has to do with the \emph{silly extra copy} that cannot be erased easily. We believe that a small-step (that is, where substitution replaces all the occurrences of a variable at the same time) silly calculus could help equating more terms, but we have not managed to work out multi types for such a calculus.

The fact that $=_\wsym$-related terms are contextually equivalent follows from the invariance property of $\eqcsilly$ below and the fact that $\eqcsilly$ is an equivalence relation. Sligthly more generally, we show that silly contextual equivalence $\eqcsilly$ is an equational theory for the SSC.

\begin{lem}[Type equivalence implies contextual equivalence]
\label{l:type-ctx-equiv}
$\tm\eqtype\tmtwo$ implies $\tm\eqcsilly\tmtwo$.
\end{lem}

\begin{proof}
If $\ctxp\tm$ is $\tow$-normalizing then, by completeness of the silly multi types for the weak case (\refthm{open-completeness}), $\ctxp\tm$ is typable. By \refthmp{silly-conversion-is-included-in-eqcsilly}{compatibility} and the hypothesis $\tm \eqtype \tmtwo$, we obtain $\ctxp\tm \eqtype \ctxp\tmtwo$. Therefore, $\ctxp\tmtwo$ is typable. By correctness of silly multi types (\refthm{weak-correctness}), $\ctxp\tmtwo$ is $\tow$-normalizing.\qedhere
\end{proof}

\begin{prop}[Contextual equivalence is an equational theory for the SSC]
\hfill
\begin{enumerate}
	\item \emph{Invariance}: $\tm\tow\tmtwo$ implies $\tm\eqcsilly\tmtwo$;	
	\item \emph{Compatibility}: $\tm\eqcsilly\tmtwo$ implies $\ctxp\tm\eqcsilly\ctxp\tmtwo$ for all contexts $\ctx$.	
\end{enumerate}
\end{prop}

\begin{proof}
\hfill
	\begin{enumerate}
	\item It follows from $\tow\subseteq\eqtype$ (\refthmp{silly-conversion-is-included-in-eqcsilly}{invariance}) and $\eqtype\subseteq\eqcsilly$ (\reflemma{type-ctx-equiv}).
	\item Straightforward consequence of the definition of $\eqcsilly$.\qedhere
	\end{enumerate}
\end{proof}

\subsection*{The Issue with Open CbV and the Silly Type System.}
There is an issue if one considers the silly type system relatively to \cbv with open terms, namely subject reduction breaks. This point is delicate. In fact, there are no issues if one considers only Plotkin's $\betav$ rule, except that Plotkin's rule is \emph{not} an adequate operational semantics for \cbv with open terms, as it is well-known and discussed at length by Accattoli and Guerrieri \cite{DBLP:conf/aplas/AccattoliG16,DBLP:journals/pacmpl/AccattoliG22}. Adequate operational semantics for \ocbv do extend Plotkin's. One such semantics is Carraro and Guerrieri's \emph{shuffling calculus} \cite{DBLP:conf/fossacs/CarraroG14}, that extends $\betav$ by adding some $\sigma$-rules. We briefly discuss it here; the last paragraph of this section shall explain because we prefer it to the \cbv LSC for the explanation of the issue with open terms.

It turns out that one of the $\sigma$ rules breaks subject reduction for the silly type system (while there are no problems if one considers instead Ehrhard's \cbv multi types \cite{DBLP:conf/csl/Ehrhard12}), as we now show. Rule $\rtosigthree$ is defined as follows:
    \begin{center}
    $\begin{array}{ccc}
    \varthree((\la\var\vartwo)\tmtwo) &\rtosigthree& (\la\var\varthree\vartwo)\tmtwo
    \end{array}$
   \end{center}
For $n\geq1$, we have the following derivation for the source term $\varthree((\la\var\vartwo)\tmtwo)$ in the silly type system:
  	\begin{center}
        \begin{adjustbox}{max width=\textwidth}$
  		\AxiomC{$\tderiv_{\varthree} \exder\ldots$}
  		\AxiomC{}  
  		\RightLabel{$\ruleAx$}
  		\UnaryInfC{$\vartwo\hastype\mult\type \vdash \vartwo \hastype \type$}  
  		\RightLabel{$\ruleLam$}
  		\UnaryInfC{$ \vartwo\hastype\mult\type \vdash \la\var\vartwo \hastype \arrowtype\zero\type$}  
  		\AxiomC{$\tderiv_{\tmtwo} \exder \red\typctx\vdash \tmtwo \hastype \atype$}  
  		\RightLabel{$\ruleMany$}
  		\UnaryInfC{$\red\typctx\vdash \tmtwo \hastype \mult\atype$}  
  		\RightLabel{$\ruleAp$}
  		\BinaryInfC{$ (\vartwo\hastype\mult\type,\red\typctx\vdash (\la\var\vartwo)\tmtwo \hastype \type)_{i=1,\ldots,n}$}  
  		\AxiomC{}  
  		\RightLabel{$\ruleAx$}
  		\UnaryInfC{$\vartwo\hastype\mult\atype \vdash \vartwo \hastype \atype$}  
  		\RightLabel{$\ruleLam$}
  		\UnaryInfC{$ \vartwo\hastype\mult\atype \vdash \la\var\vartwo \hastype \arrowtype\zero\atype$}  
  		\AxiomC{$\tderiv_{\tmtwo} \exder \red\typctx\vdash \tmtwo \hastype \atype$}  
  		\RightLabel{$\ruleMany$}
  		\UnaryInfC{$\red\typctx\vdash \tmtwo \hastype \mult\atype$}  
  		\RightLabel{$\ruleAp$}
  		\BinaryInfC{$ \vartwo\hastype\mult\atype,\red\typctx\vdash (\la\var\vartwo)\tmtwo \hastype \atype$}  
  		\RightLabel{$\ruleMany$}
  		\BinaryInfC{$ \vartwo\hastype\mult{\type^n,\atype},\red\typctx^{n+1}\vdash (\la\var\vartwo)\tmtwo \hastype \mult{\type^n,\atype}$}
  		\RightLabel{\footnotesize$\ruleAp$}
  		\BinaryInfC{$ \varthree\hastype\mult{\arrowtype{\mult{\type^n}}\typetwo }, \vartwo\hastype\mult{\type^n,\atype},\red\typctx^{n+1}\vdash \varthree((\la\var\vartwo)\tmtwo) \hastype \typetwo$}
  		\DisplayProof
        $\end{adjustbox}
    \end{center}
  	  Where $(\vartwo\hastype\mult\type,\red\typctx\vdash (\la\var\vartwo)\tmtwo \hastype \type)_{i=1,\ldots,n}$ is an abbreviation standing for $n$ copies of the derivation ending in that sequent, and where $\tderiv_{\varthree} \exder\ldots$ stands for:
	\begin{center}
  		\small$
  		\AxiomC{}  
  		\RightLabel{$\ruleAx$}
  		\UnaryInfC{$ \tderiv_\varthree\exder\varthree\hastype\mult{\arrowtype{\mult{\type^n}}\typetwo } \vdash\varthree\hastype \arrowtype{\mult{\type^n}}\typetwo$ }
  		\DisplayProof$
  	\end{center}

       The target term $(\la\var\varthree\vartwo)\tmtwo$ of rule $\rtosigthree$, instead, can only be typed as follows, the key point being that $\red\typctx^{n+1}$ is replaced by $\red\typctx$:
\begin{center}
	\tiny$
	\AxiomC{}
	\RightLabel{$\ruleAx$}
	\UnaryInfC{$\varthree\hastype\mult{\arrowtype{\mult{\type^n}}\typetwo } \vdash\varthree\hastype \arrowtype{\mult{\type^n}}\typetwo $}
	\AxiomC{}  
	\RightLabel{$\ruleAx$}
	\UnaryInfC{$(\vartwo\hastype\mult\type \vdash \vartwo \hastype \type)_{i=1,\ldots,n}$}
	\AxiomC{}  
	\RightLabel{$\ruleAx$}
	\UnaryInfC{$\vartwo\hastype\mult\atype \vdash \vartwo \hastype \atype$}
	\RightLabel{$\ruleMany$}
	\BinaryInfC{$\vartwo\hastype\mult{\type^n,\atype} \vdash \vartwo \hastype \mset{\type^n,\atype}$}  
	\RightLabel{\footnotesize$\ruleAp$}
	\BinaryInfC{$\varthree\hastype\mult{\arrowtype{\mult{\type^n}}\typetwo }, \vartwo\hastype\mult{\type^n,\atype} \vdash\varthree\vartwo\hastype \typetwo $}
	\RightLabel{$\ruleLam$}
	\UnaryInfC{$\varthree\hastype\mult{\arrowtype{\mult{\type^n}}\typetwo }, \vartwo\hastype\mult{\type^n,\atype} \vdash \la\var\varthree\vartwo \hastype \arrowtype\zero\typetwo$}  
	\AxiomC{$\tderiv_{\tmtwo} \exder \red\typctx\vdash \tmtwo \hastype \atype$}  
	\RightLabel{$\ruleMany$}
	\UnaryInfC{$\red\typctx\vdash \tmtwo \hastype \mult\atype$}  
	\RightLabel{$\ruleAp$}
	\BinaryInfC{$\varthree\hastype\mult{\arrowtype{\mult{\type^n}}\typetwo }, \vartwo\hastype\mult{\type^n,\atype},\red\typctx \vdash (\la\var\varthree\vartwo)\tmtwo \hastype \typetwo$}  
	\DisplayProof
	$\end{center}

This counter-example adapts the counter-example given by Delia Kesner to subject reduction for the multi type system by Manzonetto et al. \cite{DBLP:journals/fuin/ManzonettoPR19,DBLP:conf/fscd/KerinecMR21}, as reported in the long version on Arxiv of \cite{DBLP:journals/pacmpl/AccattoliG22}, which appeared after the publication of \cite{DBLP:journals/fuin/ManzonettoPR19,DBLP:conf/fscd/KerinecMR21}, where there is no mention of this issue. The work in the present paper can be actually seen as a clarification of the failure of subject reduction for the system in \cite{DBLP:journals/fuin/ManzonettoPR19,DBLP:conf/fscd/KerinecMR21}. Essentially, the system in \cite{DBLP:journals/fuin/ManzonettoPR19,DBLP:conf/fscd/KerinecMR21} is a system for \cbs, not for \cbv, but is therein used to study \cbv strong evaluation with possibly open terms, unaware that the system models a different evaluation mechanism.

We conjecture that the silly type system is adequate for \ocbv (that is, a term is silly typable if and only if it is \cbv terminating) even if it is not invariant for \ocbv (that is, subject reduction does not hold).

\subsection*{Naturality of the Issue.} A first reaction to the shown issue is to suspect that something is wrong or ad-hoc in our approach, especially given that the operational equivalence of \cbn and \cbneed does not suffer of this issue, that is, the \cbn multi type system is invariant for \cbneed evaluation of open terms. At high-level, however, the issue is natural and to be expected, as we now explain. The two systems of each pair \cbn/\cbneed and \cbs/\cbv have different duplicating policies and the same erasing policy. The pair \cbn/\cbneed has no restrictions on erasure, thus open normal forms have no garbage. Therefore, the different ways in which they duplicate garbage are not observable. For the pair \cbs/\cbv, instead, erasure is restricted to \emph{values}, with the consequence that garbage has to be evaluated before possibly being erased, and that with open terms some garbage might never be erased. Thus, the different duplicating policies of \cbs/\cbv leave different amounts of non-erasable garbage in open normal forms, which is observable and changes the denotational semantics.

\subsection*{The Open \cbv LSC} Another adequate formalism for \ocbv is Accattoli and Paolini's \emph{value substitution calculus} (shortened to VSC) \cite{DBLP:conf/flops/AccattoliP12}, or its micro-step variant, the \cbv LSC. We now discuss the \cbv LSC, but everything we say applies also to the VSC. 

The difference between the shuffling calculus and the \cbv LSC is that the latter uses ESs and modifies the rewriting rules at a distance. In particular, duplication is done by the following exponential rule:
\begin{center}
$\begin{array}{ccc}
\wctxp\var\esub\var{\sctxp\val} &\Rew{\vsym\esym}& \sctxp{\wctxp\val\esub\var\val}
\end{array}$
\end{center}
Subject reduction breaks also for the \cbv LSC, as we can show by adapting Kesner's counter-example. Consider the step $(\varthree\vartwo')\esub{\vartwo'}{\Id\esub\vartwo\tmtwo} \Rew{\vsym\esym} (\varthree\Id)\esub{\vartwo'}{\Id}\esub\vartwo\tmtwo$:

\begin{center}
\begin{adjustbox}{max width=\textwidth}
	\tiny$
	\AxiomC{$\tderiv_1 \exder \varthree\hastype\mult{\arrowtype{\mult{\type^n}}\typetwo },\vartwop\hastype\mult{\type^n} \vdash\varthree\vartwop\hastype 
		\typetwo $}
	\AxiomC{$\tderiv_\Id \vdash \Id \hastype \type$}  
	\AxiomC{$\tderiv_{\tmtwo} \exder \red\typctx\vdash \tmtwo \hastype \atype$}  
	\RightLabel{$\ruleMany$}
	\UnaryInfC{$\red\typctx\vdash \tmtwo \hastype \mult\atype$}  
	\RightLabel{$\ruleES$}
	\BinaryInfC{$ (\red\typctx\vdash \Id\esub\vartwo\tmtwo \hastype \type)_{i=1,\ldots,n}$}  
	\AxiomC{}  
	\RightLabel{$\ruleAx_\l$}
	\UnaryInfC{$\vdash \Id \hastype \atype$}  
	\AxiomC{$\tderiv_{\tmtwo} \exder \red\typctx\vdash \tmtwo \hastype \atype$}  
	\RightLabel{$\ruleMany$}
	\UnaryInfC{$\red\typctx\vdash \tmtwo \hastype \mult\atype$}  
	\RightLabel{$\ruleES$}
	\BinaryInfC{$ \red\typctx\vdash \Id\esub\vartwo\tmtwo \hastype \atype$}  
	\RightLabel{$\ruleMany$}
	\BinaryInfC{$ \red\typctx^{n+1}\vdash \Id\esub\vartwo\tmtwo \hastype \mult{\type^n,\atype}$}
	\RightLabel{\footnotesize$\ruleES$}
	\BinaryInfC{$ \varthree\hastype\mult{\arrowtype{\mult{\type^n}}\typetwo },\red\typctx^{n+1}\vdash (\varthree\vartwo')\esub{\vartwo'}{\Id\esub\vartwo\tmtwo} \hastype \typetwo$}
	\DisplayProof
	$\end{adjustbox}\end{center}

As for $\rtosigthree$, the key point is that that $\red\typctx^{n+1}$ gets replaced by $\red\typctx$ in the typing of the reduct:

\begin{center}
\begin{adjustbox}{max width=\textwidth}
	$
	\AxiomC{}
	\RightLabel{$\ruleAx$}
	\UnaryInfC{$ \varthree\hastype\mult{\arrowtype{\mult{\type^n}}\typetwo } \vdash\varthree\hastype \arrowtype{\mult{\type^n}}\typetwo $}
	\AxiomC{$ (\tderiv_\Id \vdash \Id \hastype \type)_{i=1,\ldots,n}$}  
	\AxiomC{}  
	\RightLabel{$\ruleAx_\l$}
	\UnaryInfC{$\vdash \Id \hastype \atype$}
	\RightLabel{$\ruleMany$}
	\BinaryInfC{$ \vdash \Id \hastype \mult{\type^n,\atype}$}
	\RightLabel{$\ruleES$}
	\BinaryInfC{$ \varthree\hastype\mult{\arrowtype{\mult{\type^n}}\typetwo }\vdash\varthree\Id\hastype \typetwo $}
	\AxiomC{}  
	\RightLabel{$\ruleAx_\l$}
	\UnaryInfC{$\vdash \Id \hastype \atype$}  
	\RightLabel{$\ruleMany$}
	\UnaryInfC{$ \vdash \Id \hastype \mult{\atype}$}
	\RightLabel{\footnotesize$\ruleES$}
	\BinaryInfC{$ \varthree\hastype\mult{\arrowtype{\mult{\type^n}}\typetwo }\vdash (\varthree\Id)\esub{\vartwo'}{\Id} \hastype \typetwo$}
	\AxiomC{$\tderiv_{\tmtwo} \exder \red\typctx\vdash \tmtwo \hastype \atype$}  
	\RightLabel{$\ruleMany$}
	\UnaryInfC{$\red\typctx\vdash \tmtwo \hastype \mult\atype$}  
	\RightLabel{\footnotesize$\ruleES$}
	\BinaryInfC{$ \varthree\hastype\mult{\arrowtype{\mult{\type^n}}\typetwo },\red\typctx\vdash (\varthree\Id)\esub{\vartwo'}{\Id}\esub\vartwo\tmtwo \hastype \typetwo$}
	\DisplayProof
	$\end{adjustbox}\end{center}

What breaks it is the use of the substitution context $\lctx$ in $\Rew{\vsym\esym}$, which can be seen as the analogous of rule $\rtosigthree$ of the shuffling calculus. The difference is that while $\rtosigthree$ is needed only for open terms in the shuffling calculus, rule $\Rew{\vsym\esym}$ is used also for the evaluation of closed terms in the \cbv LSC. This is why we preferred to avoid using the \cbv LSC to study \ccbv. In fact, one can prove that in the closed case the modification of rule $\Rew{\vsym\esym}$ without $\lctx$ is enough to reach normal forms. But this fact needs a technical study and a theorem, which we preferred to avoid.

\section{Preliminaries About Abstract Machine}
\label{sect:prel-machines}

In this section, we introduce terminology and basic concepts about abstract machines, that shall be the topic of the following sections.

\myparagraph{Abstract Machines Glossary.}  Abstract machines manipulate \emph{pre-terms}, that is, terms without implicit $\alpha$-renaming. In this paper, an \emph{abstract 
machine} is a quadruple $\mach = (\States, \tomach, \compilrel\cdot\cdot, \decode\cdot)$ the components of which are as follows.
\begin{itemize}

\item \emph{States.} A state $\state\in\States$ is composed by the \emph{active term} $\tm$ plus some data structures depending on the actual machine. Terms in states are actually pre-terms.

\item  \emph{Transitions.} The pair $(\States, \tomach)$ is a transition system with transitions $\tomach$ partitioned into \emph{principal transitions}, whose union is noted $\tomachpr$ and that are meant to correspond to steps on the calculus, and \emph{search transitions}, whose union is noted $\tomachsea$, that take care of searching for (principal) redexes.

\item \emph{Initialization.} The component $\compilrel{}{}\subseteq\terms\times\States$ is the \emph{initialization relation} associating closed terms without ESs to 
initial states. It is a \emph{relation} and not a function because $\compilrel\tm\state$ maps a \emph{closed} $\l$-term $\tm$ (considered modulo $\alpha$) to a state $\state$ having a \emph{pre-term representant} of $\tm$ (which is not modulo $\alpha$) as active term. Intuitively, any two states $\state$ and $\statetwo$ such that $\compilrel\tm\state$ and $\compilrel\tm\statetwo$ are $\alpha$-equivalent. 
A state $\state$ is \emph{reachable} if it can be reached starting from an initial state, that is, if $\statetwo \tomach^*\state$ where $\compilrel\tm\statetwo$ for some $\tm$ and $\statetwo$, shortened as $\compilrel\tm\statetwo \tomach^*\state$.

\item \emph{Read-back.} The read-back function $\decode\cdot:\States\to\esterms$ turns reachable states into 
terms possibly with ESs and satisfies the \emph{initialization constraint}: if $\compilrel\tm\state$ then $\decode{\state}=_\alpha\tm$.
\end{itemize}
A state is \emph{final} if no transitions apply.
 A \emph{run} $\run: \state \tomach^*\statetwo$ is a possibly empty finite sequence of transitions, the length of which is noted 
$\size\run$; note that the first and the last states of a run are not necessarily initial and final. 
If $a$ and $b$ are transitions labels (that is, $\tomachhole{a}\subseteq \tomach$ and 
$\tomachhole{b}\subseteq \tomach$) then $\tomachhole{a,b} \defeq \tomachhole{a}\cup \tomachhole{b}$ and $\sizep\run a$ 
is the number of $a$ transitions in $\run$, $\sizep\run {a,b} \defeq \sizep\run a + \sizep\run b $, and similarly for more than two labels. 

For the machines at work in this paper, the pre-terms in initial states shall be \emph{well-bound}, that is, they have pairwise distinct bound names; for instance $(\la\var\var)\la\vartwo\vartwo$ is well-bound while $(\la\var\var)\la\var\var$ is not. 
We shall also write $\renamenop{\tm}$ in a state $\state$ for a \emph{fresh well-bound renaming} of $\tm$,
\ie $\renamenop{\tm}$ is $\alpha$-equivalent to $\tm$, well-bound, and its bound variables
are fresh with respect to those in $\tm$ and in the other components of $\state$.

\section{The Milner Abstract Machine}
\label{sect:MAM}
In this section, we briefly recall from the literature the \cbn strategy and the Milner Abstract Machine---shortened to \emph{MAM}---that implements it. The content of this section is essentially taken from Accattoli et al. \cite{DBLP:conf/icfp/AccattoliBM14}, with the very minor change that we use a more minimal notion of structural reduction $\stredapp$ rather than their structural equivalence $\streq$.


\myparagraph{Call-by-Name Strategy.} In \reffig{cbn-strategy}, we define the \cbn strategy $\toname$ of the weak LSC. The \cbn strategy uses the root rules of the \cbn LSC given in \reffig{variants-lsc} (page \pageref{fig:variants-lsc}), in particular the GC rule $\rtogc$ that is not part of the SSC, and closes them via the notion of name evaluation contexts $\nctx$, that never enter into arguments or ESs. 

The \cbn strategy is almost deterministic: rules $\tonm$ and $\tone$ are deterministic and its erasing rule $\tongc$ is non-deterministic but diamond; for instance:
\begin{center}
\begin{tikzpicture}[ocenter]
		\node at (0,0)[align = center](source){\normalsize $\Id \esub\vartwo\Id \esub\varthree{\Id\Id}$};
		\node at (source.center)[below = 30pt](source-down){\normalsize $\Id \esub\varthree{\Id\Id}$};
		\node at (source.center)[right = 80pt](source-right){\normalsize $\Id \esub\vartwo\Id$};
		\node at (source-right|-source-down)(target){\normalsize $\Id $};
				
		\draw[->](source) to node[above] {\scriptsize $\nsym\gcsym$} (source-right);
		\draw[->](source) to node[left] {\scriptsize $\nsym\gcsym$} (source-down);
		
		\draw[->, dotted](source-down) to node[above] {\scriptsize $\nsym\gcsym$} (target);
		\draw[->, dotted](source-right) to node[right] {\scriptsize $\nsym\gcsym$} (target);

	\end{tikzpicture}
\end{center}
Moreover, $\tongc$ is postponable, \ie if $\tm \toname^*\tmtwo$ then $\tm\tonnotgc^*\tongc^*\tmtwo$ where $\tonnotgc\defeq\tonm\cup\tone$, which is why it is often omitted from the micro-step presentation of \cbn.

\begin{figure}[t!]
	\begin{center}
				\arraycolsep=3pt
\fbox{				$\begin{array}{c}
						\begin{array}{r@{\hspace{.5cm}}r@{\hspace{.1cm}}l@{\hspace{.1cm}}ll}
					  	\textsc{Name contexts} &     \nctx,\nctxtwo & \grameq & \ctxhole \mid \nctx\tm \mid \nctx\esub\var\tm
					\end{array}
\\[16pt]\hline
\textsc{Rewriting relations}\\[4pt]
				\begin{array}{rcl@{\hspace{1cm}}rcl}
					\tonm  &\defeq&  \nctxp\rtom 
					&
					\toname &\defeq &  \tonm \cup \tone \cup \tongc
					\\
					\tone  &\defeq&  \nctxp{\rtoep\nctx}
					&
					\tonnotgc & \defeq & \tonm\cup\tone
					\\
					\tongc  &\defeq&  \nctxp{\rtogc}
					\end{array}
				\end{array}$
				} 
				\caption{The call-by-name strategy.}
				\label{fig:cbn-strategy}	
				
			\end{center}
			
		\end{figure}


\begin{figure}[t!]
\centering
\fbox{
$\begin{array}{c}
\begin{array}{r@{\hspace{.25cm}} rcl  @{\hspace{.5cm}}  @{\hspace{.5cm}} r @{\hspace{.25cm}} rcl}
	\textsc{Stacks} & \stack 	& \grameq & \emptylist \mid \tm \cons \stack
	&
	\textsc{Environments} & \env	& \grameq & \emptylist \mid \esub\var\tm \cons \env 
	\\
	\textsc{States} &	\state	& \defeq & \mamst\tm\stack\env
\end{array}
\\
\hline
\textsc{Transitions}
\\
  	{\setlength{\arraycolsep}{0.35em}
  	\begin{array}{c|c|c||l||c|c|cccc}
		 \mbox{Code} & \mbox{Stack} &  \mbox{Env} 
		&&
		 \mbox{Code} & \mbox{Stack} &  \mbox{Env} \\
\hhline{=|=|=|=|=|=|=}
		 \tm\tmtwo & \stack &  \env
	  	&\tomachsea&
	  	 \tm & \tmtwo\cons\stack & \env
	  	\\
		 \la\var\tm & \tmtwo \cons\stack  &\env
		& \tomachm &
		\tm & \stack  &\esub\var\tmtwo\cons\env
		\\
		\var & \stack & \env\cons\esub\var\tm\cons\envtwo
	  	&\tomache &
		\renamenop\tm & \stack & \env\cons\esub\var\tm\cons\envtwo
	\end{array}}
\\[4pt]
\mbox{Where $\renamenop{\tm}$ is any well-bound code $\alpha$-equivalent to $\tm$ such that its }
\\ \mbox{bound names  are fresh with respect to those in the other data structures.}
\\[4pt]
\hline
\begin{array}{r@{\hspace{.25cm}} rcl  @{\hspace{.25cm}}@{\hspace{.35cm}}|  r @{\hspace{.25cm}} rcl}
	\multicolumn{8}{c}{\textsc{Read-back}}
	\\
	\textsc{Empty} & \decode\emptylist & \defeq & \ctxhole
	&
	\textsc{Envs} & \decode{\esub\var\tm \cons \env } & \defeq & \decodep\env{\ctxhole\esub\var\tm}
	\\
	\textsc{Stacks} & \decode{\tm \cons \stack} & \defeq & \decodep\stack{\ctxhole\tm}
	&
	\textsc{States} &	\decode{\mamst\tm\stack\env} & \defeq & \decodep\env{\decodep\stack\tm}
\end{array}
\end{array}$
}
\caption{The Milner Abstract Machine (MAM).}
\label{fig:mam}
\end{figure}
\myparagraph{The MAM} The MAM is an abstract machine for the \cbn strategy, defined in \reffig{mam}; more precisely, for the non-erasing part $\tonnotgc$ of the \cbn strategy. It has two data structures, the (applicative) stack $\stack$ and the environment $\env$. The stack collects the argument encountered so far by the machine, via the search transition $\tomachsea$. The environment collects the explicit substitutions created by the $\msym$ transition $\tomachm$. When the active code is a variable $\var$, transition $\tomache$ looks up in the environment and replaces the variable with a freshly $\alpha$-renamed copy $\renamenop\tm$ of the term $\tm$ associated by $\env$ to $\var$. The renaming is essential for the correctness of the machine, and it amounts---when concretely implementing the MAM---to make a copy of $\tm$; see Accattoli and Barras for OCaml implementations of the MAM \cite{DBLP:conf/ppdp/AccattoliB17} together with their complexity analyses. 

Transitions $\msym$ and $\esym$ are the principal transitions of the MAM.

The MAM never garbage collects ESs, exploiting the fact that GC can be postponed. The read-back of MAM states to terms (defined in \reffig{mam}) can be compactly described by turning stacks and environment into contexts made out of arguments and ESs, respectively.

\begin{exa}
We show the first few transitions of the diverging execution of the MAM on $\Omega$. We use the following notation: $\delta_\var \defeq \la\var\var\var$.
\begin{center}$\begin{array}{c|c|c|lllll}
		 \mbox{Code} & \mbox{Stack} &  \mbox{Env} 
		&
		\\
\hhline{=|=|=|}
\delta_\var\delta_\vartwo & \emptylist & \emptylist & \tomachsea
\\
\delta_\var & \delta_\vartwo & \emptylist & \tomachm
\\
\var\var & \emptylist & \esub\var{\delta_\vartwo} & \tomachsea
\\
\var & \var & \esub\var{\delta_\vartwo} & \tomache
\\
\delta_\varthree & \var & \esub\var{\delta_\vartwo} & \tomachm
\\
\varthree\varthree & \emptylist & \esub\varthree\var\cons\esub\var{\delta_\vartwo} & \tomachsea\ldots
\end{array}$\end{center}
\end{exa}
\myparagraph{Matching the Strategy and the Machine} The bisimulation between the MAM and the non-erasing \cbn strategy $\tonnotgc$ is given by the read-back function. The read-back allows one to project transitions of the MAM on the LSC so that the $\msym$ and $\esym$ transitions map to $\tonm$ and $\tone$ steps respectively, while $\seasym$ transitions map to equalities on terms. The converse simulation is obtained by proving that if the read-back term can reduce then the machine can do a transition (that actually maps on that step).

In fact, this simple schema has a glitch, concerning $\msym$ transitions. Consider one such transition $\state = \mamst {\la\var\tm} {\tmtwo \cons\stack}  \env \tomachm \mamst \tm \stack {\esub\var\tmtwo\cons\env} = \statetwo$ and let us read it back:
\begin{equation}\begin{array}{rllll}
 \decode{\mamst {\la\var\tm} {\tmtwo \cons\stack}  \env}
 & = &
 \decodep\env{\decodep{\tmtwo \cons\stack}{\la\var\tm}}
\\ & = &
 \decodep\env{\decodep\stack{(\la\var\tm)\tmtwo}}
\\ & \tonm &
 \decodep\env{\decodep\stack{\tm\esub\var\tmtwo}}
\\ & \red{\neq} &
 \decodep{\env}{\decodep\stack{\tm}\esub\var\tmtwo}
\\ & = &
 \decodep{\esub\var\tmtwo\cons\env}{\decodep\stack{\tm}}
 & = &		
	\decode{\mamst \tm \stack {\esub\var\tmtwo\cons\env}}.
\end{array}
\label{eq:mam-need-for-stred}
\end{equation}
Note that the $\tonm$ steps can be applied because $\decodep\env{\decode\stack}$ is a \cbn evaluation context. Note, however, that the simulation does not work because of $\decodep\env{\decodep\stack{\tm\esub\var\tmtwo}}
\neq  \decodep{\env}{\decodep\stack{\tm}\esub\var\tmtwo}$. In order to validate such a step, the following rule commuting ESs out of the applications is needed: 
\begin{equation}
\begin{array}{lll\colspace \colspace ll}
(\tm\esub\var\tmtwo)\,\tmthree & \Rew{@l} & (\tm\,\tmthree)\esub\var\tmtwo               & \text{if $\var \not\in \fv{\tmthree}$} 
\end{array}
\label{eq:stred}
\end{equation}
Unfortunately, such a rule is not part of the LSC, so we cannot simulate $\msym$ transitions using $\tonm$ alone.

At this point, the key observation is that the rule in \refeq{stred}, while not part of the LSC, has a very nice property: it is itself a strong bisimulation with respect to $\tonnotgc$. More precisely, let $\stredapp$ be defined as in \reffig{stred-name}. Then we have the following property.
\begin{figure}[t!]
\centering
\fbox{
 \begin{tabular}{c}
 \textsc{Root case of structural reduction $\stredapp$}
\\[4pt]

$\begin{array}{rll@{\hspace{1em}}l}
(\tm\esub\var\tmtwo)\,\tmthree & \stredapp & (\tm\,\tmthree)\esub\var\tmtwo               & \text{if $\var \not\in \fv{\tmthree}$} \\
\end{array}$
\\[5pt]\hline
\begin{tabular}{c}
\textsc{Closure rules}
\\[4pt]
\begin{tabular}{c\colspace c\colspace c\colspace c}
\RightLabel{$\symfont{N\text{-}ctx}$}
\AxiomC{$\tm \streq_{\asym} \tm'$}
\UnaryInfC{$\nctxp\tm \streq \nctxp{\tm'}$}
\DisplayProof
&
\AxiomC{}
\RightLabel{$\symfont{refl}$}
\UnaryInfC{$\tm \stredapp \tm$}
\DisplayProof
&
\AxiomC{$\tm \stredapp \tmtwo$}
\AxiomC{$\tmtwo \stredapp \tmthree$}
\RightLabel{$\symfont{tran}$}
\BinaryInfC{$\tm \stredapp \tmthree$}
\DisplayProof
\end{tabular}
\end{tabular}
\end{tabular}
}
\caption{Structural reduction for simulating the MAM.}
\label{fig:stred-name}
\end{figure}

\begin{prop}[$\stredapp$ is a strong bisimulation]
\label{prop:name-strong-bisim}
Let $\tm \stredapp \tmtwo$  and $\asym\in\set{\nsym\msym,\nsym\esym}$. Then:
\begin{enumerate}
\item If $\tmtwo \Rew{\asym} \tmtwo'$ then there exists $\tm'$ such that $\tm \Rew{\asym} \tm'$ and $\tm' \stredapp \tmtwo'$.
\item If $\tm \Rew{\asym} \tm'$ then there exists $\tmtwo'$ such that $\tmtwo \Rew{\asym} \tmtwo'$ and $\tm' \stredapp \tmtwo'$.
\end{enumerate}
\end{prop}

\begin{proof}
Both points are unsurprising inductions on $\tm \stredapp \tmtwo$ and case analyses of the reduction step. We formalized them in the Abella proof assistant, the sources can be found on GitHub \cite{AbellaSources}.\qedhere
\end{proof}

Note that an immediate consequence of Point 1 is that $\stredapp$ can always be postponed.

\begin{cor}[$\stredapp$ postponement]
If $\tm (\stredapp \tonnotgc\stredapp )^k \tmtwo$ then $\tm \tonnotgc^k\stredapp  \tmtwo$.
\end{cor}

The idea then is that the sequence of equalities in \refeq{mam-need-for-stred} can be completed with $\stredapp$, obtaining Point 1.1 of the next lemma; the other two are proved by simply unfolding the definitions (and using some easy omitted invariants, the next section shall give all the details for a more general case).
\begin{lem}
Let $\state$ be a reachable MAM state.
\begin{enumerate}
\item \emph{Principal projection}:
\begin{enumerate}
\item If  $\state \tomachm\statetwo$ then $\decode\state\tonm\stredapp\decode\statetwo$;
\item If  $\state \tomache\statetwo$ then $\decode\state\tone\decode\statetwo$;
\end{enumerate}

\item \emph{Search transparency}: if  $\state \tomachsea\statetwo$ then $\decode\state=\decode\statetwo$.
\end{enumerate}
\end{lem}

More generally, $\tonnotgc$ simulates the MAM up to $\stredapp$, as summed up by the following theorem.

\begin{thm}[MAM implementation]
\label{thm:mam-implem}
\hfill
\begin{enumerate}
\item \emph{Runs to evaluations}: for any MAM run $\run: \compilrel\tm\state \tomach^* \statetwo$ there exists a \cbn evaluation $\deriv: \tm \tonnotgc^* \stredapp\decode\statetwo$;

\item \emph{Evaluations to runs}: for every \cbn evaluation $\deriv: \tm \tonnotgc^* \tmtwo$ there exists a 
MAM run $\run: \compilrel\tm\state \tomach^* \statetwo$ such that $\decode\statetwo \stredapp \tmtwo$;

\end{enumerate}
\end{thm}

The theorem is a minor variation over the same one appearing in Accattoli et al. \cite{DBLP:conf/icfp/AccattoliBM14} proved via the distillation technique, that is used also in \cite{DBLP:conf/aplas/AccattoliBM15,DBLP:conf/lics/AccattoliC15,DBLP:conf/fscd/AccattoliMPC25}; the word \emph{distillation} refers to the fact that the LSC distills abstract machines by working up to the search mechanism, as captured by the \emph{search transparency} property above. 

The difference is that usually the distillation is given with respect to a structural equivalence $\streq$ rather than to a structural reduction such as our $\stredapp$. Structural equivalence $\streq$ is a standard concept of the theory of the LSC, that extends $\stredapp$ with further commutations (while still being a strong bisimulation with respect to $\tonnotgc$) and useful beyond the study of abstract machines. Here we decided to give a slightly different, more economical presentation using the structural reduction $\stredapp$, that contains only what is strictly needed to simulate the MAM. The change does not affect the proof of the theorem, that only relies on the strong bisimulation property of $\streq$, given for $\stredapp$ by \refprop{name-strong-bisim}. 

In the next section, we shall develop a slight variant of the distillation technique to study the silly extension of the MAM, thus the reader shall see it at work.

\myparagraph{Final States} Final states of the MAM have shape $\mamst{\la\var\tm}\emptylist\env$, which read back to terms of the form $\sctxp{\la\var\tm}$, which are normal forms for $\tonnotgc$. To prove that final states have that shape one uses the \emph{closure invariant} of the machine: whenever the active code is a variable $\var$ there is an entry for $\var$ in the environment $\env$, so that the machine cannot be stuck on variables. The invariant holds because the machine is always run on a closed term. In the next section, we shall precisely state the invariant for the silly machine.
\section{The Silly Abstract Machine}
\label{sect:machine}
In this section, we extend the MAM as to compute non-erasing \cbs normal forms for closed terms, obtaining the Silly MAM. Usually, abstract machines are shown to implement a strategy. For the moment being, we do not define a \cbs evaluation strategy; it shall be the topic of the following section. We then only show a slightly weaker implementation theorem for the Silly MAM, with respect to the SSC rather than with respect to a strategy.

\myparagraph{Closed Normal Forms} First of all, we characterize the normal forms for non-erasing SSC reduction on closed terms, that have a simple shape, simply referred to as \emph{answers}. They shall be the outcome of the Silly MAM runs.

\begin{defi}[Answers]
Answers are defined as follows:
\begin{center}
	$\begin{array}{r@{\hspace{.25cm}}r@{\hspace{.1cm}}l@{\hspace{.1cm}}ll}
		\textsc{Answers} &     \ans,\anstwo & \grameq & \val \mid \ans\esub\var\anstwo 
	\end{array}$
\end{center}
\end{defi}

\begin{lem}
Let $\ans$ be an answer. Then $\ofv\ans=\emptyset$.\label{l:ans-empty-ofv}
\end{lem}

\begin{proof}
Immediate induction on $\ans$.
\end{proof}

\begin{prop}
\label{l:compact-answers-are-nfs-for-non-erasing-cbs-strategy}
Let $\tm$ be such that $\ofv\tm=\emptyset$. Then $\tm$ is $\townotgcv$-normal if and only if $\tm$ is an answer.
\end{prop}

\begin{proof}
\begin{itemize}
\item  $\Leftarrow$) By induction on $\tm$. Cases:
\begin{itemize}
\item \emph{Value}, that is, $\tm = \val$. Then $\tm$ is a $\townotgcv$-normal.
\item \emph{Substitution}, that is, $\tm = \ans\esub\var\anstwo$. By \ih, $\ans$ and $\anstwo$  are $\townotgcv$-normal. By \reflemma{ans-empty-ofv}, $\var\notin\ofv\ans$ so there are no $\rtoep\wctx$ redexes involving $\esub\var\anstwo$ in $\tm$. Then $\tm$ is $\townotgcv$-normal.
\end{itemize}

\item $\Rightarrow$) By induction on $\tm$. Cases:
\begin{itemize}
\item \emph{Variable}, that is, $\tm = \var$. Impossible because $\ofv\tm=\emptyset$.
\item \emph{Abstraction}, that is, $\tm = \la\var\tmtwo$. Then $\tm$ has the required shape.
\item \emph{Application}, that is, $\tm=\tmtwo\tmthree$. By \ih, $\tmtwo$ has shape $\sctxp{\la\var\tmfour}$ and then $\tm$ has a $\tosim$-redex---absurd. Then this case is impossible.
\item \emph{Substitution}, that is, $\tm = \tmtwo\esub\var\tmthree$.  By \ih, both $\tmtwo$ and $\tmthree$ are answers, that is, $\tm$ has the required shape.\qedhere
\end{itemize}
\end{itemize}
\end{proof}
\begin{figure}[t!]
\centering
\fbox{
$\begin{array}{c}
\begin{array}{r@{\hspace{.25cm}} rcl  @{\hspace{.5cm}}  @{\hspace{.5cm}} r @{\hspace{.25cm}} rcl}
	\textsc{Stacks} & \stack 	& \grameq & \emptylist \mid \tm \cons \stack
	&
	\textsc{Partial Answers} & \pans & \grameq & \emptylist \mid \pans\cons(\val,\var)
	\\
	\textsc{States} &	\state	& \defeq & \glamst\pans\tm\stack\env
	&
	\textsc{Environments} & \env	& \grameq & \emptylist \mid \esub\var\tm \cons \env 
\end{array}
\\
\hline
\textsc{Transitions}
\\
  	{\setlength{\arraycolsep}{0.35em}
  	\begin{array}{c|c|c|c||l||c|c|c|cccc}
		 \mbox{PAnsw} &\mbox{Code} & \mbox{Stack} &  \mbox{Env} 
		&&
		 \mbox{PAnsw} &\mbox{Code} & \mbox{Stack} &  \mbox{Env} \\
\hhline{=|=|=|=|=|=|=|=|=}
		 \pans & \tm\tmtwo & \stack &  \env
	  	&\tomachseaone&
	  	 \pans & \tm & \tmtwo\cons\stack & \env
	  	\\
%
		 \pans & \la\var\tm & \tmtwo \cons\stack  &\env
		& \tomachm &
		\pans &  \tm & \stack  &\esub\var\tmtwo\cons\env
		\\
		 \pans & \var & \stack & \env\cons\esub\var\tm\cons\envtwo
	  	&\tomache &
		 \pans & \renamenop\tm & \stack & \env\cons\esub\var\tm\cons\envtwo
	  	\\		
		 \pans & \val & \emptylist & \esub\var\tm\cons\env
	  	&\tomachseatwo &
		 \pans\cons(\val,\var) & \tm & \emptylist & \env
	  	\\
	\end{array}}
\\[4pt]
\mbox{Where $\renamenop{\tm}$ is any well-bound code $\alpha$-equivalent to $\tm$ such that its }
\\ \mbox{bound names  are fresh with respect to those in the other data structures.}
\\[4pt]
\hline
\begin{array}{r@{\hspace{.25cm}} rcl  @{\hspace{.25cm}}| @{\hspace{.25cm}}  r @{\hspace{.25cm}} rcl}
	\multicolumn{8}{c}{\textsc{Read-back}}
	\\
	\textsc{Empty} & \decode\emptylist & \defeq & \ctxhole
	\\
	\textsc{Stacks} & \decode{\tm \cons \stack} & \defeq & \decodep\stack{\ctxhole\tm}
	&
	\textsc{Envs} & \decode{\esub\var\tm \cons \env } & \defeq & \decodep\env{\ctxhole\esub\var\tm}
	\\
	\textsc{PAnswers} &	\decode{\pans\cons(\val,\var)}	& \defeq & \decodep\pans{\val}\esub\var\ctxhole
	&
	\textsc{States} &	\decode{\fourstate\pans\tm\stack\env} & \defeq & \decodep\env{\decodep\pans{\decodep\stack\tm}}
\end{array}
\end{array}$
}
\caption{The Silly Milner Abstract Machine (Silly MAM).}
\label{fig:silly-machine}
\end{figure}

\myparagraph{The Silly MAM} The Silly Milner Abstract Machine (Silly MAM) is in \reffig{silly-machine}. At the dynamic level, the Silly MAM  is the extension of the MAM that, once the MAM can no longer reduce, starts evaluating the ESs (out of abstractions) left hanging in the environment. At the level of the specification, it extends the MAM by adding a data structure $\pans$ dubbed \emph{partial answer}, and a further search transition $\tomachseatwo$ that move the focus of the machine inside the first ES on the environment only when it is sure that it is no longer needed. When the MAM stops, which happens when the code is a value and the stack is empty, the Silly MAM continues: it stores the obtained value on the partial answer and pops the first entry of the environment, starting the evaluation of its term and saving its variable on the partial answer. The partial answer encodes the portion of the answer computed so far, via a mechanism that is made explicit by the read-back in \reffig{silly-machine}. As we shall prove below, the silly MAM stops when both the stack and the environment are empty.

An important detail, is that, when an ES $\esub\var\tm$ at the beginning of the environment is \emph{no longer needed}, it does not necessarily mean that $\var$ has no free occurrences in the rest of the state, as occurrences might be blocked by abstractions, as for instance in $\fourstate\emptylist{\la\vartwo\var}\emptylist{\esub\var\tmtwo}$ where $\var$ occurs but only under abstraction, and the Silly MAM has to evaluate inside $\tmtwo$. Therefore, \emph{no longer needed} rather means \emph{without shallow free occurrences} (\refdef{ofv}, page \pageref{def:ofv}), which is captured by the notion of answer given above.

\begin{exa}
We show the execution of the Silly MAM on $(\la\var\Id) (\Id \Id)$. We use the following notation: $\Id_\var \defeq \la\var\var$. At first, The Silly MAM runs as the MAM:
\begin{center}$\begin{array}{c|l|c|c|lllll}
		 \mbox{PAnsw} &\mbox{Code} & \mbox{Stack} &  \mbox{Env} 
		&
		\\
\hhline{=|=|=|=}
\emptylist &(\la\var\Id_\vartwo) (\Id_\varthree \Id_\varfour) & \emptylist & \emptylist & \tomachsea
\\
\emptylist &\la\var\Id_\vartwo  & \Id_\varthree \Id_\varfour & \emptylist & \tomachm
\\
\emptylist &\Id_\vartwo  & \emptylist & \esub\var{\Id_\varthree \Id_\varfour}
\end{array}$\end{center}
Now, this is a MAM final state, but the Silly MAM keeps running:
\begin{center}$\begin{array}{r|l|c|clllll}
		 \mbox{PAnsw} &\mbox{Code} & \mbox{Stack} &  \mbox{Env} 
		&
		\\
\hhline{=|=|=|=|}
\emptylist &\Id_\vartwo  & \emptylist & \esub\var{\Id_\varthree \Id_\varfour} & \tomachseatwo
\\
(\Id_\vartwo, \var) & \Id_\varthree \Id_\varfour & \emptylist & \emptylist & \tomachseaone
\\
(\Id_\vartwo, \var) & \Id_\varthree  & \Id_\varfour & \emptylist & \tomachm
\\
(\Id_\vartwo, \var) & \varthree  & \emptylist & \esub\varthree{\Id_\varfour} & \tomache
\\
(\Id_\vartwo, \var) & \Id_{\varfour'}  & \emptylist & \esub\varthree{\Id_\varfour} & \tomachseatwo
\\
(\Id_\vartwo, \var)\cons (\Id_{\varfour'},\varthree) & \Id_\varfour  & \emptylist & \emptylist 
\end{array}$\end{center}
Which is a final state that reads back to $(\la\vartwo\vartwo) \esub\var{\Id_{\varfour'}}\esub\varthree{\Id_\varfour}$ that is an answer.
\end{exa}

\myparagraph{Invariants.} To prove properties of the Silly MAM we need to isolate some of its invariants in \reflemma{skeletal-mam-qual-invariants} below. The closure one ensures that the closure of the initial term extends, in an appropriate sense, to all reachable states. The well-bound invariant, similarly, ensures that binders in reachable states are pairwise distinct, as in the initial term. To compactly formulate the closure invariant we need the notion of terms of a state. 

\begin{defi}[Terms of a state]
Let $\state=\fourstate\pans{\tmtwo}{\stack}{\env}$ be a Silly MAM state where $\pans$ is a possibly empty chain $ \emptylist\cons(\val_k,\var_k)\cons\ldots\cons(\val_1,\var_1)$ for some $k\geq0$. 
The \emph{terms of $\state$} are $\tmtwo$, every term in $\stack$ and every term in an entry of $\env$. 
\end{defi}

\begin{lem}[Qualitative invariants]
\label{l:skeletal-mam-qual-invariants}
Let $\compilrel\tm\state \tomach^*\statetwo=\fourstate\pans{\tmtwo}{\stack}{\env}$ be a Silly MAM run.
\begin{enumerate}
\item \emph{Closure}: for every term $\tmthree$ of $\statetwo$  and for any variable $\var\in\fv\tmthree$ there is an environment entry $\esub\var\tm$ or a partial answer entry $(\valtwo,\var)$ on the right of $\tmthree$ in $\statetwo$.
\item \emph{Well-bound}: if $\la\var\tmthree$ occurs in $\statetwo$ and $\var$ has any other
occurrence in $\statetwo$ then it is as a free variable of $\tmthree$, and for any environment entry $\esub\vartwo\tm$ or partial answer entry $(\val,\vartwo)$ in $\statetwo$ the name $\vartwo$ can occur (in any form) only on the left of that entry in $\statetwo$.

\end{enumerate}
\end{lem}

\begin{proof}
For both points, the proof is by induction on the length of the run. The base case trivially holds, the inductive case is by analysis of the last transition, which is always a straightforward inspection of the transitions using the \ih
\end{proof}

\myparagraph{Structural Reductions.} As for the MAM, the simulation of transition $\msym$ requires the structural reduction $\stredapp$. The presence of the partial answer data structure $\pans$, however, forces us to introduce a further structural reduction $\stredans$. It is defined in \reffig{streds}.

\begin{figure}[t!]
\centering
\begin{tabular}{c}
\fbox{
 \begin{tabular}{c}
 \textsc{Root case of structural reduction $\stredapp$}
\\[4pt]

$\begin{array}{rll@{\hspace{1em}}l}
(\tm\esub\var\tmtwo)\,\tmthree & \stredapp & (\tm\,\tmthree)\esub\var\tmtwo               & \text{if $\var \not\in \fv{\tmthree}$}
\end{array}$
\\[5pt]\hline
\begin{tabular}{c}
\textsc{Closure rules}
\\[4pt]
\begin{tabular}{c\colspace c\colspace c\colspace c}
\RightLabel{$\symfont{w\text{-}ctx}$}
\AxiomC{$\tm \streq_{\asym} \tm'$}
\UnaryInfC{$\wctxp\tm \streq \wctxp{\tm'}$}
\DisplayProof
&
\AxiomC{}
\RightLabel{$\symfont{refl}$}
\UnaryInfC{$\tm \streq \tm$}
\DisplayProof
&
\AxiomC{$\tm \streq \tmtwo$}
\AxiomC{$\tmtwo \streq \tmthree$}
\RightLabel{$\symfont{tran}$}
\BinaryInfC{$\tm \streq \tmthree$}
\DisplayProof
\end{tabular}
\end{tabular}
\end{tabular}
}
\\\\
\fbox{
 \begin{tabular}{c}
 \textsc{Root case of structural reduction $\stredans$}
\\[4pt]

$\begin{array}{rll@{\hspace{1em}}l}
\ans\esub\var{\tm\esub\vartwo\tmtwo} & \stredans & \ans\esub\var\tm\esub\vartwo\tmtwo   & \text{if $\vartwo \not\in \fv{\ans}$} 
\end{array}$
\\[5pt]\hline
\begin{tabular}{c}
\textsc{Closure rules}
\\[4pt]
\begin{tabular}{c\colspace c\colspace c\colspace c}
\RightLabel{$\symfont{sub\text{-}ctx}$}
\AxiomC{$\tm \streq_{\asym} \tm'$}
\UnaryInfC{$\sctxp\tm \streq \sctxp{\tm'}$}
\DisplayProof
&
\AxiomC{}
\RightLabel{$\symfont{refl}$}
\UnaryInfC{$\tm \streq \tm$}
\DisplayProof
\end{tabular}
\end{tabular}
\end{tabular}
}
\end{tabular}
\caption{Structural reductions.}
\label{fig:streds}
\end{figure}

Note that $\stredans$ has a different contextual closure than $\stredapp$, as it is restricted to substitution contexts. This is to ensure its strong bisimulation property: answers are not stable by application and closing by applicative contexts---thus obtaining an extended reduction $\stredp{\symfont{ans}'}$---would break the strong bisimulation property, as the next diagram shows:
\begin{center}
\begin{tikzpicture}[ocenter]
		\node at (0,0)[align = center](source){\normalsize $(\la\var\var\vartwo)\esub\vartwo{\tm\esub\varthree\tmtwo} \tmthree$};
		\node at (source.center)[right = 100pt](source-right){\normalsize $(\var\vartwo)\esub\var\tmthree\esub\vartwo{\tm\esub\varthree\tmtwo} $};
		\node at (source-right.center)[below = 20pt](target){\normalsize $(\var\vartwo)\esub\var\tmthree\esub\vartwo{\tm}\esub\varthree\tmtwo $};
				
		\node at (source.center)[below = 20pt](source-down){\normalsize $(\la\var\var\vartwo)\esub\vartwo\tm\esub\varthree\tmtwo \tmthree$};
		
		\node at \med{source-down.center}{source.center}(fifthnode){\normalsize $\stredp{\symfont{ans}'}$};
		\node at \med{source-right.center}{target.center}(fifthnode){\normalsize $\not\stredp{\symfont{ans}'}$};

		\draw[->](source) to node[above] {\scriptsize $\wsym\msym$} (source-right);
		\draw[->, dotted](source-down) to node[above] {\scriptsize $\wsym\msym$} (target);
	\end{tikzpicture}
\end{center}
This is the reason why we define two separate structural reductions rather than one with two initial cases. 

The complex definition of the silly strategy of the next section is due to a similar interaction between applications and ESs in the specification of evaluation contexts. The slight advantage of the abstract machine with respect to the strategy is that these technicalities show up only when relating the machine to the calculus, but not on the definition of the machine itself; the definition of the strategy instead needs to take them into account.

Note also that $\stredans$ does not need a transitive closure: the structural reduction required for handling partial answers, indeed, has to commute out of at most one ES (while $\stredapp$ might have to commute with many arguments, so it needs the transitive closure).

Lastly, note that the definition of $\stredapp$ in \reffig{streds} is slightly different from the one in the previous section (despite using the same symbol): it is now closed by all weak contexts $\wctx$.

\begin{prop}
\label{prop:silly-strong-bisim}
Both $\stredapp$ and $\stredans$ are strong bisimulations with respect to both $\towm$ and $\towe$. Namely, let $\asym\in\set{@l,\symfont{ans}}$ and $\bsym\in\set{\wsym\msym,\wsym\esym}$, then:
\begin{enumerate}
\item If $\tm \stredp{\asym} \tmtwo$ and $\tmtwo \Rew{\bsym} \tmtwo'$ then there exists $\tm'$ such that $\tm \Rew{\bsym} \tm'$ and $\tm' \stredp{\asym} \tmtwo'$.
\item If $\tm \stredp{\asym} \tmtwo$ and $\tm \Rew{\bsym} \tm'$ then there exists $\tmtwo'$ such that $\tmtwo \Rew{\bsym} \tmtwo'$ and $\tm' \stredp{\asym} \tmtwo'$.
\end{enumerate}
\end{prop}

\begin{proof}
Similar to the \cbn case (\refprop{name-strong-bisim}), both points are unsurprising inductions on $\tm \stredapp \tmtwo$ and case analyses of the reduction step. We formalized them in the Abella proof assistant, the sources can be found on GitHub \cite{AbellaSources}.\qedhere
\end{proof}

The concatenation $\stredapp\stredans$ of the two structural reductions is also a strong bisimulation, simply because strong bisimulations are stable by concatenation. 

\textbf{Notation}: we use $\tm \stred \tmtwo$ when $\tm\stredapp\tmthree\stredans \tmtwo$ for some $\tmthree$.

\myparagraph{Key Properties.} The invariants allow us to prove some key properties of the Silly MAM with respect to the SSC. These properties are the bricks of the distillation technique: the bisimulation of the machine and the calculus shall only use these properties, and never directly refer to the invariants. The closure invariant is used for proving the halt property, the well-bound one for proving principal projection.
\begin{prop}[Key properties]
\label{prop:SiMAM-properties}
Let $\state$ be a reachable Silly MAM state.
\begin{enumerate}
\item \label{p:SiMAM-properties-projection}\emph{Principal projection}: 
\begin{enumerate}
\item If $\state \tomachm \statetwo$ then $\decode\state \towm\stred \decode\statetwo$;
\item If $\state \tomache \statetwo$ then $\decode\state \towe \decode\statetwo$.
\end{enumerate}

\item \label{p:SiMAM-properties-overhead-transparency}\emph{Search transparency}: if $\state \tomachhole{\seasym_1,\seasym_2} \statetwo$ then $\decode\state = \decode\statetwo$.

\item \label{p:SiMAM-properties-overhead-terminates}\emph{Search terminates}: $\tomachhole{\seasym_1,\seasym_2}$ is terminating.

\item \label{p:SiMAM-properties-halt}\emph{Halt}: if $\state$ is final then it has shape $\fourstate\pans\val\emptylist\emptylist$ and $\decode\state$ is an answer.

\end{enumerate}
\end{prop}

\begin{proof}
\hfill
\begin{enumerate}
\item 
\begin{enumerate}
\item If $\state = \fourstate \pans {\la\var\tm} {\tmtwo \cons\stack}  \env \tomachm \fourstate \pans \tm \stack {\esub\var\tmtwo\cons\env} = \statetwo$ then:
\[\begin{array}{rllll}
 \decode{\fourstate \pans {\la\var\tm} {\tmtwo \cons\stack}  \env}
 & = &
 \decodep\env{\decodep\pans{\decodep{\tmtwo \cons\stack}{\la\var\tm}}}
\\ & = &
 \decodep\env{\decodep\pans{\decodep\stack{(\la\var\tm)\tmtwo}}}
\\ & \towm &
 \decodep\env{\decodep\pans{\decodep\stack{\tm\esub\var\tmtwo}}}
\\ (\reflemmaeq{skeletal-mam-qual-invariants}.2) & \stredapp &
 \decodep{\env}{\decodep\pans{\decodep\stack{\tm}\esub\var\tmtwo}}
\\ (\reflemmaeq{skeletal-mam-qual-invariants}.2) & \stredans &
 \decodep{\env}{\decodep\pans{\decodep\stack{\tm}}\esub\var\tmtwo}
\\ & = &
 \decodep{\esub\var\tmtwo\cons\env}{\decodep\pans{\decodep\stack{\tm}}}
 & = &		
	\decode{\fourstate \pans \tm \stack {\esub\var\tmtwo\cons\env}}.
\end{array}\]
Note that the $\towm$ steps applies because $ \decodep\env{\decodep\pans{\decode\stack}}$ is a weak context. The structural reduction steps apply because the well-bound invariant (\reflemma{skeletal-mam-qual-invariants}.2) ensures that $\var$ does not occur in $\decode\stack$ and $\decode\pans$.

\item If $\state = \fourstate \pans \var \stack {\env\cons\esub\var\tm\cons\envtwo}
	  	\tomache
		\fourstate \pans {\renamenop\tm} \stack {\env\cons\esub\var\tm\cons\envtwo}
 = \statetwo$ then:
\[\begin{array}{rllll}
 \decode{\fourstate \pans \var \stack {\env\cons\esub\var\tm\cons\envtwo}}
 & = &
 \decodep{\env\cons\esub\var\tm\cons\envtwo}{\decodep\pans{\decodep\stack{\var}}}
\\ & = &
 \decodep{\envtwo}{\decodep\env{\decodep\pans{\decodep\stack{\var}}}\esub\var\tm}
\\ &\towe &
 \decodep{\envtwo}{\decodep\env{\decodep\pans{\decodep\stack{\renamenop\tm}}}\esub\var\tm}
\\ & = &
 \decodep{\env\cons\esub\var\tm\cons\envtwo}{\decodep\pans{\decodep\stack{\renamenop\tm}}}
\\ & = &		
	\decode{\fourstate \pans {\renamenop\tm} \stack {\env\cons\esub\var\tm\cons\envtwo}}.
\end{array}\]
The $\towe$ steps applies because:
\begin{enumerate}
\item As in the previous case, both $ \decode{\envtwo}$ and $ \decodep{\envtwo}{\decodep\env{\decodep\pans{\decode\stack}}\esub\var\tm}$ are weak contexts;
\item The well-bound invariant (\reflemma{skeletal-mam-qual-invariants}.2) ensures that there is only one environment entry on $\var$ in $\state$, so that there are no entries bounding $\var$ in $\env$ (if there were another entry $\esub\var\tmtwo$ in $\env$ then the occurrence of $\var$ in $\esub\var\tm$ would violate the invariant).
\end{enumerate}
\end{enumerate}

\item Cases of the search transition:
\begin{enumerate}
\item If $\state = \fourstate \pans {\tm\tmtwo} \stack \env
	  	\tomachseaone
	  	\fourstate \pans \tm {\tmtwo\cons\stack} \env = \statetwo$ then:
\[\begin{array}{lllllllll}
 \decode{\fourstate \pans {\tm\tmtwo} \stack \env}
 & = &
 \decodep{\env}{\decodep\pans{\decodep\stack{\tm\tmtwo}}}
& = &
 \decodep{\env}{\decodep\pans{\decodep{\tmtwo\cons\stack}\tm}}
 & = &		
	\decode{\fourstate \pans \tm {\tmtwo\cons\stack} \env}.
\end{array}\]

\item If $\state = \fourstate \pans \val \emptylist {\esub\var\tm\cons\env}
	  	\tomachseatwo 
	  	\fourstate {\pans\cons(\val,\var)} \tm \emptylist \env  = \statetwo$ then:
\[\begin{array}{lllllllll}
 \decode{\fourstate \pans \val \emptylist {\esub\var\tm\cons\env}}
 & = &
 \decodep{\esub\var\tm\cons\env}{\decodep\pans{\val}}
\\ & = &
 \decodep\env{\decodep\pans{\val}\esub\var\tm}
\\ & = &
 \decodep\env{\decodep{\pans\cons(\val,\var)}\tm}
 & = &		
	\decode{\fourstate {\pans\cons(\val,\var)} \tm \emptylist \env}.
\end{array}\]
\end{enumerate}

\item Note that:
\begin{enumerate}
\item $\sizep\run{\seasym_2} \leq \sizep\run\msym$, because every $\seasym_2$ transition consumes one  entry from the environment, which are created only by $\msym$ transitions.
\item $\seasym_1$ transitions decrease the size of the code;
\item The size of the code is increased by $\seasym_2$ transitions, but after a $\seasym_1$ transition there cannot be any $\seasym_2$ transition because the stack becomes non-empty.
\item Therefore, the maximum number of consecutive $\seasym_1$ transitions is bound by the size of the code.
\end{enumerate}

\item Let us first show that final states $\state_f$ have shape $\fourstate\pans\val\emptylist\emptylist$. We have the following points:
\begin{itemize}
\item \emph{The code of a final state is a value}: it cannot be an application, because transition $\seasym_1$ would apply, and cannot be a variable $\var$, because by the closure invariant (\reflemma{skeletal-mam-qual-invariants}.1) there has to be an environment entry on $\var$, and so  transition $\esym$ would apply;
\item \emph{The stack is empty}: because, given that the code is an abstraction, if the stack is non-empty then a $\msym$ transition would apply;
\item \emph{The environment is empty}: because, given that the code is a value and the stack is empty, then transition $\seasym_2$ would apply.
\end{itemize}
A straightfoward induction on partial answers shows that final states read back to answers.
\qedhere
\end{enumerate}
\end{proof}

The next theorem states that the Silly MAM implements the SSC. It is slightly weaker than other implementation theorems in two respects. First, because the simulation of the calculus by the machine is formulated in a \emph{big-step} manner (that is, looking at evaluations to normal form), rather than step-by-step (or small-step). This is inevitable given the fact that we are relating the machine to the whole calculus and not to a specific strategy. Second, the structural reduction part requires a transitive closure, as explained by the following remark.

\begin{rem}
The structural reduction $\stredapp$ that we used for the MAM is closed transitively, so that $\stredapp\stredapp\subset \stredapp$, that is, we can merge many structural reduction steps into a single one. This is used in the proof of the MAM implementation theorem to have a single structural reduction step at the end of the reduction sequence at the calculus level. 

Note that, instead, it is not possible to merge $\stred$ steps (which are defined as the concatenation $\stredapp\stredans$), for two reasons. First, $\stredans$ is not closed transitively. Second, even if we would close $\stredans$ transitively, $\stredapp$ and $\stredans$ have different contextual closures, and this fact forbids (the transitive closure of) $\stredans$ to postpone after $\stredapp$, which is required to merge two $\stred$ steps.
\end{rem}

\begin{thm}[Silly MAM implementation]
Let $\tm$ be a closed term.
  \begin{enumerate}
   \item \label{p:exec-to-deriv} \emph{Runs to small-step evaluations}: for any Silly MAM run $\run: \compilrel\tm\state \tomach^* \statetwo$ there exists a 
$\tow$-evaluation $\deriv: \tm \townotgcv^* \stred^*\decode\statetwo$;

\item \label{p:deriv-to-exec} \emph{Big-step evaluations to runs}: for every $\tostrat$-evaluation $\deriv: \tm \townotgcv^* \ans$ to an answer $\ans$ there exists a 
Silly MAM run $\run: \compilrel\tm\state \tomach^* \statetwo$ such that $\ans\stred^* \decode\statetwo$;
  \end{enumerate}
\end{thm}

\begin{proof}
\hfill
\begin{enumerate}
\item By induction on $\sizepr\run \in \nat$. Cases:
  \begin{itemize}
  \item $\sizepr\run = 0$. Then $\run \colon \compilrel\tm\state \tomachsea^* \statetwo$ and hence $\decode{\state} = \decode\statetwo$ by search transparency (\refpropp{SiMAM-properties}{overhead-transparency}).
  Moreover, $\tm = \decode{\state}$ since decoding is the inverse of initialization, therefore the statement holds with respect to the empty evaluation $\deriv$ with starting and end term $\tm$.
  
\item $\sizepr\run > 0$. Then, $\run \colon \compilrel\tm\state \tomach^* \statetwo$ is the concatenation of a run $\runtwo \colon \compilrel\tm\state \tomach^* \statethree$ followed by a run $\runthree \colon \statethree \tomachhole{\asym} \statefour \tomachsea^* \statetwo$ with $\asym\in\set{\msym,\esym}$.
  By \ih applied to $\runtwo$, there exists an evaluation $\derivtwo \colon \tm \townotgcv^* \stred^*\decode\statetwo$.
  By principal projection (\refpropp{SiMAM-properties}{projection}) and search transparency (\refpropp{SiMAM-properties}{overhead-transparency}) applied to $\runthree$, one obtains a one-step evaluation $\derivthree \colon \decode\statetwo \Rew{\asym} \stred\decode\statethree = \decode\state$. 
Concatenating $\derivtwo$ and $\derivthree$, we obtain an evaluation $\derivfour \colon \tm  \townotgcv^* \stred^*\decode\statetwo\townotgcv \stred\decode\state$.
Postponing structural reduction (\refprop{silly-strong-bisim}), we obtain the required evaluation $\deriv \colon \tm  \townotgcv^* \townotgcv\stred^*\stred\decode\state$.
 \end{itemize}
 
\item Cases of $\size\deriv\in \nat$:
  \begin{itemize}
  \item $\size\deriv = 0$. Then $\tm = \ans$.
  Since decoding is the inverse of initialization, one has $\decode\state = \ans$.
  Then the statement holds with respect to the empty (\ie without transitions) run $\run$ with initial (and final) state $\state$.
  
 \item  $\size\deriv > 0$. Then, $\deriv\colon \tm \townotgcv^* \ans$ starts with a step $\derivtwo \colon \tm \townotgcv \tmtwo$  for some $\tmtwo$. Let $\compilrel\tm\state$. By the initialization constraint, $\decode\state=\tm$. Consider the normal form $\nfo{\state}$ of $\state$ with respect to $\tomachsea$, that exists by search termination (\refpropp{SiMAM-properties}{overhead-terminates}), and that verifies $\decode{\nfo{\state}}=\decode\state = \tm$ by search transparency (\refpropp{SiMAM-properties}{overhead-transparency}). Note that:
 \begin{itemize}
 \item Then run on Silly MAM is non-empty: since $\tm$ can reduce, by the halt property (\refpropp{SiMAM-properties}{halt}), the Silly MAM can do either a $\msym$ or a $\esym$ transition on $\nfo{\state}$. 
 \item The maximal run $\run$ of the Silly MAM starting on $\nfo{\state}$ cannot be divergent: by search termination (\refpropp{SiMAM-properties}{overhead-terminates}) a divergent run has an infinite number of $\msym/\esym$ transitions, which by principal projection (\refpropp{SiMAM-properties}{projection})  give a diverging evaluation $\decode\run$ from $\tm$, against the fact that by hypothesis there is a normalizing evaluation from $\tm$ and that thus, by uniform normalization (\refcoro{unif-norm}), $\tm$ is strongly normalizing.
 \end{itemize}
 Then consider the maximal run $\run:\nfo{\state} \tomach^* \statetwo$ with $\statetwo$ final. By point 1 and the halt property, we obtain a reduction $\decode{\nfo{\state}}=\tm \townotgcv^* \anstwo \stred^* \decode\statetwo$. By confluence of $\townotgcv$ (\refthm{confluence}), $\ans = \anstwo$. \qedhere
  \end{itemize}
  \end{enumerate}
\end{proof}
\section{The Call-by-Silly Strategy}
\label{sect:cbs-strategy}
In this section, we define the \cbs evaluation strategy as a sort of extension of the \cbn one. The aim is to show that the evaluation strategy underlying the Silly MAM can be expressed inside the SSC, that is, that the machine is only a handy device. 

The definition of the strategy is somewhat surprisingly more technical than the definition of the Silly MAM. This is due to the fact that the machine keeps the ESs all grouped together in the environment, that helps with the specification of what to evaluate next. On terms, instead, ESs are spread around, and this induces a tricky management of evaluation contexts.

\begin{figure}[t!]
	\centering
				\arraycolsep=3pt
\fbox{				$\begin{array}{cc|cc}
					\begin{array}{r@{\hspace{.25cm}}r@{\hspace{.1cm}}l@{\hspace{.1cm}}ll}
					\textsc{Name ctxs} &     \nctx,\nctxtwo & \grameq & \ctxhole \mid \nctx\tm \mid \nctx\esub\var\tm
					\\
						\textsc{Aux. ctxs} & \actx,\actxtwo & \grameq & \ctxhole \mid \ans\esub\var\actx \mid \actx\esub\var\tm
						\\[3pt]
						\textsc{Silly ctxs} & \sictx,\sictxtwo & \grameq & \actxp\nctx
					\end{array}
				& &&
				\begin{array}{rll}
					\tosim & \defeq & \sictxp\rtom					
					\\
					\tosigcv & \defeq & \sictxp{\rtogcv}
					\\
					\tosieaa & \defeq & \actxp{\rtoep\sictx}
					\\
					\tosieyn & \defeq & \sictxp{\rtoep\nctx}
					\\
                    \tosi & \defeq & \tosim \cup \tosieaa \\ && \cup \tosieyn \cup \tosigcv
				\end{array}
				\end{array}$
} 
				\caption{The call-by-silly strategy $\tosi$.}
				\label{fig:cbs-strategy}	
		\end{figure}

%

\myparagraph{The Call-by-Silly Strategy.} The \cbs strategy $\tosi$ is defined in \reffig{cbs-strategy} via silly evaluation contexts $\sictx$---explained next---(we use $\sictx$ for \emph{sillY} because $\sctx$ is already used for substitution contexts) and the \cbv erasing rule $\rtogcv$ of the SSC. As it is the case for the Silly MAM, the \cbs strategy is an extension of the \cbn strategy. Such an extension is specified via silly evaluation contexts $\sictx$, that are defined by extending \cbn contexts $\nctx$ via \emph{auxiliary contexts} $\actx$, whose key production is $\ans\esub\var\actx$ for evaluation contexts where $\ans$ is an answer, that is, a \cbs normal form with no shallow occurrences of $\var$. There are, however, a few subtleties.

\begin{enumerate}
\item \emph{Answers, distance, and applications}: in order to be sure that no occurrences of $\var$ shall become shallow because of other steps, in $\ans\esub\var\sictx$ the sub-term $\ans$ is a normal form, otherwise in a case such as $((\la\vartwo\var)\Id)\esub\var\tmtwo$ the reduction of the multiplicative step to $\var\esub\vartwo\Id\esub\var\tmtwo$ turns a blocked occurrence of $\var$ into a shallow occurrence. Because of rules at a distance, however, the term $(\la\vartwo\var)\esub\var\tmtwo \Id$ suffers of the same problem but the argument comes \emph{on the right} of the ES. For this reason, we forbid silly evaluation contexts to be applied once they enter an ES, that is, silly contexts $\sictx$ are defined as $\actxp\nctx$, that is, rigidly separating the construction that applies (at work in $\nctx$, the same contexts of the \cbn strategy) from the one that enters ES (at work in $\actx$).

\item \emph{Two exponential rules}: such a rigid separation in the definition of silly contexts forces us to have \emph{two} exponential rules. Exponential rules use contexts twice in their definition: to select the occurrence to replace in the root rule \emph{and} to extend the applicability of the root rule. The subtlety is that if both these contexts are silly, then one can nest an $\actx$ context (selecting an occurrence) under a $\nctx$ context (extending the rule), which, as explained, has to be avoided. Therefore there are two silly exponential rules $\tosieaa$ and $\tosieyn$ carefully designed as to avoid the dangerous combination (that would be given by $\nctxp{\rtoep\actx}$).

\item \emph{Non-erasing normal forms and postponement of GC by value}: in order to mimic the property of \cbn that GC is postponable, we shall ask that in $\tm\esub\var\sictx$ the sub-term $\tm$ is a normal form only for the \emph{non-erasing} \cbs, otherwise in a case such as $\var\esub\vartwo\Id\esub\var\tmtwo$ the \cbs strategy would be forced to erase $\esub\vartwo\Id$ before evaluating $\tmtwo$. Normal forms for non-erasing \cbs are the answers of the previous section, that are shown  below to be normal forms for the non-erasing \cbs strategy (\reflemma{compact-answers-are-nfs-for-cbs-strategy}). We shall also prove that $\tosigcv$ steps are postponable (\refprop{gc-postponement}).
\end{enumerate}

The \cbs strategy is supposed to be applied to closed terms only, while there is no closure hypothesis on the weak calculus.

\begin{exa}
\label{ex:evaluation}
A good example to observe the differences between \cbn and \cbs is given by the term $\tm \defeq (\la\vartwo\la\var\var(\la\varfour\var))  \Omega (\Id\Id)$ where $\Id \defeq \la \varthree \varthree$, $\Omega \defeq \delta \delta$, and $\delta \defeq \la\varthree\varthree\varthree$. In \cbn, it evaluates with 4 multiplicative steps and 3 exponential steps, as follows:
\begin{center}
	$\begin{array}{llll}
  \tm 
  & \tonm  (\la\var\var(\la\varfour\var)) \esub\vartwo\Omega (\Id\Id) &
  & \tonm  (\var(\la\varfour\var)) \esub\var{\Id\Id}\esub\vartwo\Omega \\
  & \tone  ((\Id\Id)(\la\varfour\var)) \esub\var{\Id\Id}\esub\vartwo\Omega & 
  & \tonm  (\varthree\esub\varthree \Id (\la\varfour\var)) \esub\var{\Id\Id}\esub\vartwo\Omega \\
  & \tone  (\Id\esub\varthree \Id (\la\varfour\var)) \esub\var{\Id\Id}\esub\vartwo\Omega &
  & \tonm  \varthree' \esub{\varthree'}{\la\varfour\var} \esub\varthree \Id  \esub\var{\Id\Id}\esub\vartwo\Omega  \\
  & \tone  (\la\varfour\var) \esub{\varthree'}{\la\varfour\var}  \esub\varthree \Id  \esub\var{\Id\Id}\esub\vartwo\Omega  & &=: \tmtwo
  \end{array}$
\end{center}
The obtained term $\tmtwo$ is normal for the multiplicative and exponential rules. Since in \cbn one can erase every term, $\tmtwo$ then reduces (in \cbn) as follows:
\begin{center}
	$\begin{array}{lllllllll}
  \tmtwo &  \tongc & (\la\varfour\var)   \esub\varthree \Id  \esub\var{\Id\Id}\esub\vartwo\Omega  &\tongc& 
  (\la\varfour\var)    \esub\var{\Id\Id}\esub\vartwo\Omega & \tongc & (\la\varfour\var) \esub\var{\Id\Id}
  \end{array}$
\end{center}
The obtained term is now normal for the \cbn strategy, since evaluation does not enter into ESs.
	Note that $\tm$ diverges for \cbv evaluation (defined in \refsect{cbv}). In \cbs, the first part of the evaluation of $\tm$, up to $\tmtwo$, is the same as for \cbn (where $\tone$ steps become $\tosieyn$ steps). Then, $\tmtwo$ evaluates differently than in \cbn, diverging:
\begin{center}{
	$\begin{array}{lllllllll}
  \tmtwo &  \tosigcv & (\la\varfour\var)   \esub\varthree \Id  \esub\var{\Id\Id}\esub\vartwo\Omega  &\tosigcv& 
  (\la\varfour\var)    \esub\var{\Id\Id}\esub\vartwo\Omega 
  \\
  & \tosim & (\la\varfour\var) \esub\var{\varthree\esub\varthree\Id}\esub\vartwo\Omega
  &\tosieaa & (\la\varfour\var) \esub\var{\Id\esub\varthree\Id}\esub\vartwo\Omega   
  \\
  &\tosigcv & (\la\varfour\var) \esub\var{\Id}\esub\vartwo\Omega    &\tosim & (\la\varfour\var) \esub\var{\Id}\esub\vartwo{\varthree\varthree\esub\varthree\delta} &...
  \end{array}$
}
\end{center}
\end{exa}

\myparagraph{Normal Forms for the \cbs Strategy} We characterize normal forms for both the \cbs strategy and its non-erasing variant $\tosinotgcv \defeq \tosim \cup \tosieaa \cup \tosieyn$.

\begin{lem}[Normal forms for the non-erasing \cbs strategy]
\label{l:compact-answers-are-nfs-for-cbs-strategy}
Let $\tm$  be such that $\ofv\tm=\emptyset$. Then $\tm$ is $\tosinotgcv$-normal if and only if it is an answer.
\end{lem}

\begin{proof}
Direction $\Rightarrow$ is trivial, because $\tosinotgcv$ is a sub-relation of $\townotgcv$, that has answers as normal forms (\reflemma{compact-answers-are-nfs-for-non-erasing-cbs-strategy}). Direction $\Leftarrow$ is by induction on $\tm$. Cases:
\begin{itemize}
\item \emph{Variable}, that is, $\tm = \var$. Impossible because $\ofv\tm=\emptyset$.
\item \emph{Abstraction}, that is, $\tm = \la\var\tmtwo$. Then $\tm$ is an answer.
\item \emph{Application}, that is, $\tm=\tmtwo\tmthree$.  By \ih, $\tmtwo$ has shape $\sctxp{\la\var\tmfour}$, and then $\tm$ has a $\tosim$-redex---absurd. Then this case is impossible.
\item \emph{Substitution}, that is, $\tm = \tmtwo\esub\var\tmthree$.  Since $\ctxhole\esub\var\tmthree$ is a silly evaluation context, $\tmtwo$ is $\tosinotgcv$-normal. By \ih, $\tmtwo$ is an answer. Then $\tosinotgcv$ does evaluate $\tmthree$, which implies that $\tmthree$ is $\tosinotgcv$-normal. By \ih, $\tmthree$ is an answer. Then so is $\tm$.\qedhere
\end{itemize}
\end{proof}

\begin{lem}[Normal forms for the \cbs strategy]
Let $\tm$ be such that $\ofv\tm=\emptyset$.\label{l:closed-nfs}
The following are equivalent:
\begin{enumerate}
\item $\tm$ is $\tow$-normal;
\item $\tm$ is $\tosi$-normal;
\item $\tm$ is a strict answer, defined as follows:
\begin{center}
	$\begin{array}{r@{\hspace{.5cm}} rll}
		\textsc{Strict answers} &   \sans & \grameq & \val \mid \sans\esub\var{\sanstwo} \ \mbox{ with }\var\in\fv{\sans}
		\end{array}$ 
\end{center}
\end{enumerate}
\end{lem}

\begin{proof}
The proposition refines \reflemma{compact-answers-are-nfs-for-cbs-strategy} and \reflemma{compact-answers-are-nfs-for-non-erasing-cbs-strategy} by adding garbage collection to both $\tosinotgcv$ and $\townotgcv$ which clearly has the effect of removing garbage ESs out of abstraction from answers, obtaining exactly strict answers. \qedhere
\end{proof}

\myparagraph{Non-Erasing Determinism and Erasing Diamond} Next, we prove that the \cbs strategy is almost deterministic. Its non-erasing rules are deterministic, while $\tosigcv$ is diamond, exactly as it is the case for the \cbn strategy. For instance:
\begin{center}
\begin{tikzpicture}[ocenter]
		\node at (0,0)[align = center](source){\normalsize $\Id \esub\vartwo\Id \esub\varthree\delta$};
		\node at (source.center)[below = 30pt](source-down){\normalsize $\Id \esub\varthree\delta$};
		\node at (source.center)[right = 80pt](source-right){\normalsize $\Id \esub\vartwo\Id$};
		\node at (source-right|-source-down)(target){\normalsize $\Id $};
				
		\draw[->](source) to node[above] {\scriptsize $\sisym\gcvsym$} (source-right);
		\draw[->](source) to node[left] {\scriptsize $\sisym\gcvsym$} (source-down);
		
		\draw[->, dotted](source-down) to node[above] {\scriptsize $\sisym\gcvsym$} (target);
		\draw[->, dotted](source-right) to node[right] {\scriptsize $\sisym\gcvsym$} (target);

	\end{tikzpicture}
\end{center}
 
\begin{prop}
Reductions $\tosim$, $\tosieaa$, and $\tosieyn$ are deterministic, and $\tosigcv$ is diamond.\label{prop:silly-determinism-diamond}
\end{prop}
 
\begin{proof}
The proofs of these properties are tedious and somewhat long but otherwise straightforward, the details can be found in the technical report \cite{accattoli2024mirroring}. Since silly evaluation contexts are defined by composing auxiliary and name contexts, one first needs to prove the properties for name contexts, that is, for the \cbn strategy, for which they are straightforward. For both name and silly contexts the proofs are by induction on the term $\tm$ giving rise to the steps. For the exponential steps, one additionally needs lemmas proving that if $\tm = \sictxfp\var$ then such a decomposition is unique and $\tm$ is normal, building on analogous statements for name and auxiliary contexts. The proof of the diamond property is an unsurprising case analysis.\qedhere
\end{proof}

\myparagraph{Linear Postponement} Lastly, we prove the \emph{linear} postponement of $\tosigcv$, that is, the fact that it can be postponed while preserving the length of the evaluation. The length preservation shall be used to prove the maximality of $\tosi$ in the next section. We first give the local property that is then iterated to prove the global postponement.

\begin{lem}[Local postponement of silly GC by value]
\label{l:local-posponement-silly}
\hfill
\begin{enumerate}
\item \emph{Silly multiplicative}: if $\tm \tosigcv\tm'\tosim \tm''$ then $\tm \tosim\tm'''\tosigcv \tm''$ for some $\tm'''$.
\item \emph{Silly exponential 1}: if $\tm \tosigcv\tm'\tosieaa \tm''$ then $\tm \tosieaa\tm'''\tosigcv \tm''$ for some $\tm'''$.
\item \emph{Silly exponential 2}: if $\tm \tosigcv\tm'\tosieyn \tm''$ then $\tm \tosieyn\tm'''\tosigcv \tm''$ for some $\tm'''$.
\end{enumerate}
\end{lem}

\begin{proof}
All points are by induction on the first step and case analysis of the second. As for the previous proof, the details are tedious and somewhat long but otherwise straightforward, they can be found in the technical report \cite{accattoli2024mirroring}. One needs to first show the properties for the closure by name contexts. The exponential steps need a \emph{backward deformation lemma} stating that if $\tm \tosigcv\sictxfp\var$ then $\tm=\sictxtwofp\var$ for some \cbs context $\sictxtwo$ such that $\sictxtwofp\tmtwo \tosigcv\sictxfp\tmtwo$ for every $\tmtwo$ (plus a similar statement for the case of name contexts), mirroring the (forward) deformation lemma used in the proof of local confluence (\reflemma{local-confluence}). \qedhere
\end{proof}

\begin{prop}[Postponement of $\tosigcv$]
If $\deriv:\tm\tosi^*\tmtwo$ then $\tm\tosinotgcv^k\tosigcv^h\tmtwo$ with $k=\sizep\deriv{\sisym\neg\gcv}$ and $h= \sizep\deriv{\sisym\gcv}$.
\label{prop:silly-gc-postponement}
\end{prop}
 
 \begin{proof}
As it was the case for the postponement of GC by value in \refprop{gc-postponement}, the statement  is an instance of a well-known rewriting property: the local swaps in \reflemma{local-posponement-silly} above can be put together as the following \emph{linear} local postponement property (note the \emph{single} step postponed after $\tosinotgcv$):
\begin{center}
$\tm \tosigcv \tm' \tosinotgcv \tm''$ then $\tm \tosinotgcv\tm'''\tosigcv \tm''$ for some $\tm'''$
\end{center}
giving the global postponement of the statement when iterated on a reduction sequence. \qedhere
\end{proof}

\myparagraph{Back to the Silly MAM} We can now prove that the principal transitions of the Silly MAM read back to steps of the strategy. The next lemma is the key point.

\begin{lem}[Contextual read-back invariant]
\label{l:ctx-read-back}
Let $\state=\fourstate\pans\tm\stack\env$ a reachable Silly MAM state. Then $\decode\pans$ and $\decodep\env{\decode\pans}$ are auxiliary contexts, and $\decodep\pans{\decode\stack}$  and $\decodep\env{\decodep\pans{\decode\stack}}$ are silly evaluation contexts.
\end{lem}

\begin{proof}
The proof is by induction on the length of the run. The base case trivially holds, the inductive case is by analysis of the last transition, which is always a straightforward inspection of the transitions using the \ih
\end{proof}

\begin{prop}[Refined principal projection]
Let $\state$ a reachable Silly MAM state. 
\begin{enumerate}
\item If $\state \tomachm \statetwo$ then $\decode\state \tosim\stred \decode\statetwo$;
\item If $\state \tomache \statetwo$ then $\decode\state \tosieaa \decode\statetwo$.
\end{enumerate}
\end{prop}

\begin{proof}
By the contextual read-back invariant (\reflemma{ctx-read-back}), in the proof of principal projection (\refpropp{SiMAM-properties}{projection}) the context around the $\msym$ step is a silly context, while for the exponential step the context around the occurrence is a silly context and the  context surrounding the root step is a substitution context, thus an auxiliary context.\qedhere
\end{proof}

A possibly interesting point is that exponential transitions never read back to $\tosieyn$. This is not because $\tosieyn$ is useless: it happens because the structural rearrangement at work in the machine, that moves every created ES together with the other ones in the environment, so that $\tosieyn$ steps are out of the image of the read-back.

\section{Tight Derivations, Exact Lengths, and Maximality}
\label{sect:tight}
Here, we focus on the \cbs strategy on closed terms and isolate a class of \emph{tight} type derivations whose indices measure exactly the number of non-erasing \cbs steps, mimicking faithfully what was done in call-by-name/value/need by Accattoli et al. in \cite{DBLP:conf/esop/AccattoliGL19}. An outcome shall be that the \cbs strategy actually performs the \emph{longest} weak evaluation on closed terms.

\myparagraph{Tight Derivations.} The basic tool for our quantitative analysis are tight type derivations, which shall be used to refine the correctness theorem of the weak case (\refthm{weak-correctness}) in the case of closed terms. Tight derivations are defined via a predicate on their last judgement, as in \cite{DBLP:conf/esop/AccattoliGL19}. 

\begin{defi}[Tight types and derivations]
A type $\ttype$ is \emph{tight} if $\ttype=\atype$ or $\ttype=\mset\atype$. 
A derivation $\tderiv \exder  \Deri[(\msteps, \esteps)]{\typctx}{\tm\!}{\ttype}$  is \emph{tight} if $\ttype$ is tight.\label{def:tightderiv}
  \end{defi}

\myparagraph{Tight Correctness.} We first show that all tight derivations for answers have null indices.

\begin{prop}[Tight typing of normal forms for non-erasing \cbs]
Let $\ans$ be an answer and 
$\tderiv\exder  \Deri[(\msteps, \esteps)]{\typctx}{\ans}\normal$ be a derivation. 
Then $\typctx$ is empty and $\msteps = \esteps = 0$.	\label{prop:silly-normal-forms-forall}
\end{prop}

\begin{proof}
By induction on $\ans$. Cases:
\begin{itemize}
\item \emph{Base}, \ie $\ans = \la{\var}\tm$. Then $\tderiv$ can only be:
$$      \infer[\normal]{
        \Deri[(0,0)] {} { \la\var\tm } \normal
      }{}
$$
which satisfies the statement.

\item \emph{Inductive}, \ie $\ans= \anstwo \esub\var\ansthree$. Then $\tderiv$ has the following shape:
\[
\AxiomC{ $\tderiv_{\anstwo} \exder \Deri[(\msteps_{\anstwo}, \esteps_{\anstwo})] {\typctx_{\anstwo}} \anstwo {\atype}$ }
\AxiomC{ $ \Deri[(\msteps_{\ansthree}, \esteps_{\ansthree})] {\typctx_{\ansthree}} {\ansthree} {\typctx_{\anstwo}(\var) \mplus \mult\atype}$ }
\RightLabel{$\ruleES$}
\BinaryInfC{ $\Deri[(\msteps_{\anstwo} + \msteps_{\ansthree}, \esteps_{\anstwo} + \esteps_{\ansthree})] {\typctx_{\anstwo} \mplus \typctx_{\ansthree}}
             {\anstwo \esub\var\ansthree} \atype$ }
\DisplayProof
\]
with $\typctx = \typctx_{\anstwo} \mplus \typctx_{\ansthree}$, $\msteps = \msteps_{\anstwo} + \msteps_{\ansthree}$, and $\esteps = \esteps_{\anstwo} + \esteps_{\ansthree}$. By \ih, we have that $\tderiv_{\anstwo}$ actually has shape $\tderiv_{\anstwo} \exder \Deri[(0,0)] {} \anstwo {\atype}$. Thus, more precisely, $\tderiv$ has the following shape:
\[
\AxiomC{ $\tderiv_{\anstwo} \exder \Deri[(0,0)] {} \anstwo {\atype}$ }
\AxiomC{ $\tderiv_{\ansthree} \exder \Deri[(\msteps_{\ansthree}, \esteps_{\ansthree})] {\typctx_{\ansthree}} \ansthree {\atype}$ }
\RightLabel{$\ruleMany$}
\UnaryInfC{ $ \Deri[(\msteps_{\ansthree}, \esteps_{\ansthree})] {\typctx_{\ansthree}} {\ansthree} {\mult\atype}$ }
\RightLabel{$\ruleES$}
\BinaryInfC{ $\Deri[(\msteps_{\ansthree}, \esteps_{\ansthree})] {\typctx_{\ansthree}}
             {\anstwo \esub\var\ansthree} \atype$ }
\DisplayProof
\]
with $\typctx = \typctx_{\ansthree}$, $\msteps = \msteps_{\ansthree}$, and $\esteps = \esteps_{\ansthree}$. By \ih, we have that $\tderiv_{\ansthree}$ actually has the shape $\tderiv_{\ansthree} \exder \Deri[(0,0)] {} \ansthree {\atype}$. Then $\tderiv$ has the following shape, which satisfies the statement:
\begin{equation*}
\raisebox{\depth}{$
\AxiomC{ $\tderiv_{\anstwo} \exder \Deri[(0,0)] {} \anstwo {\atype}$ }
\AxiomC{ $\tderiv_{\ansthree} \exder \Deri[(0,0)] {} \ansthree {\atype}$ }
\RightLabel{$\ruleMany$}
\UnaryInfC{ $ \Deri[(0,0)] {} \ansthree {\mult\atype}$ }
\RightLabel{$\ruleES$}
\BinaryInfC{ $\Deri[(0, 0)] {}
             {\anstwo \esub\var\ansthree} \atype$ }
\DisplayProof
$}
\tag*{\qedhere}
\end{equation*}
\end{itemize}
\end{proof}


Next, we refine 
subject reduction via tight derivations.

\myparagraph{Tight Subject Reduction.} The proof of tight subject reduction requires a careful treatment because of the intricate context closures of the \cbs strategy (interleaving name $\nctx$ and auxiliary contexts $\actx$ defined in \reffig{cbs-strategy}). We first provide an auxiliary subject reduction statement that is concerned with the name evaluation part of the reduction steps happening in the call-by-silly strategy.


\begin{prop}[Auxiliary subject reduction for non-erasing \cbn steps]
	Let $\tderiv\exder  \Deri[(\msteps, \esteps)]\typctx{\tm}\ltype$ be a derivation.  \label{prop:name-subject-reduction} 
	\begin{enumerate}
		\item \emph{Multiplicative}: 
		if $\tm\tonm\tmtwo$ then $\msteps\geq 1$ and
		there is 
		$\tderivtwo\exder  \Deri[(\msteps-1, \esteps)]\typctx{\tmtwo}\ltype$.
		
		\item \emph{Exponential}: 
		if $\tm\tone\tmtwo$ then $\esteps\geq 1$ and
		there is 
		$\tderivtwo\exder  \Deri[(\msteps, \esteps-1)]\typctx{\tmtwo}\ltype$.
	\end{enumerate}
\end{prop}  

\begin{proof}
	The two points are proved by induction on the reduction steps. For the second point, the exponential case, the root step uses a \cbn linear substitution lemma \cite[Lemma 56]{accattoli2024mirroring}.
\end{proof}

%
%
%


We then prove tight subject reduction for any step of the \cbs strategy. Note that the auxiliary statement is not a subcase of tight subject reduction (\refprop{silly-subject-reduction}): name contexts $\nctx$ are included in general silly contexts $\sictx$ (hence $\tonm\,\subseteq\,\tosim$ and $\tone\,\subseteq\,\tosieyn$) but linear types $\ltype$ are not necessarily tight (the only tight linear type is $\atype$). The auxiliary statement is however necessary (see below the sub-case \emph{Left of an application} in the proof of the multiplicative point).

\begin{prop}[Tight subject reduction for non-erasing \cbs]
  Let $\tderiv\exder  \Deri[(\msteps, \esteps)]\typctx{\tm}\atype$ be a tight derivation.  \label{prop:silly-subject-reduction} 
  \begin{enumerate}
    \item \emph{Multiplicative}: 
 if $\tm\tosim\tmtwo$ then $\msteps\geq 1$ and
  there is 
  $\tderivtwo\exder  \Deri[(\msteps-1, \esteps)]\typctx{\tmtwo}\atype$.
  
  \item \emph{Exponential}: 
 if $\tm\tosieaa\tmtwo$ or $\tm\tosieyn\tmtwo$ then $\esteps\geq 1$ and
  there is 
  $\tderivtwo\exder  \Deri[(\msteps, \esteps-1)]\typctx{\tmtwo}\atype$.
  \end{enumerate}
\end{prop}

\begin{proof}  
\begin{enumerate}
\item \emph{Multiplicative}. We have $\tm = \sictxp\tmthree \tosim \sictxp\tmfour = \tmtwo$. Cases of $\sictx$:
\begin{itemize}
\item \emph{Root step}, that is, $\sictx=\ctxhole$ and $\tm =  \lctxp{\la\var \tmthree} \tmfour \rtom \lctxp{ \tmthree \esub\var\tmfour } = \tmtwo$. The proof is exactly as in the weak case (\refprop{open-subject-reduction}), because in the root case the $\msteps$ index always decreases of exactly 1.

		\item \emph{Left of an application}, that is, $\sictx = \nctx \tmfive$. Then $\tderiv$ has the following form:
		\[
\AxiomC{ $\tderiv_\tmthree \exder\tyjp{(\msteps_\tmthree,\esteps_\tmthree)}{\nctxp{\tmthree}}{\typctx_\tmthree}{\arrowtype{\mtype}{\atype} }$ }
\AxiomC{ $\tderiv_\tmfive \exder\tyjp{(\msteps_\tmfive,\esteps_\tmfive)}{\tmfive}{\typctx_\tmfive}{\mtype \mplus\mult\atype} $ }
\RightLabel{$\ruleAp$}
\BinaryInfC{ $\tyjp{(\msteps_\tmthree + \msteps_\tmfive + 1, \esteps_\tmthree + \esteps_\tmfive)}{\nctxp{\tmthree} \tmfive}{\typctx_\tmthree \mplus \typctx_\tmfive}{\atype}$ }
\DisplayProof
\]
With $\mtype\neq\zero$, $\typctx = \typctx_\tmthree  \mplus  \typctx_\tmfive$, $\msteps = \msteps_\tmthree + \msteps_\tmfive + 1$, and $\esteps = \esteps_\tmthree + \esteps_\tmfive$. The statement then follows by \cbn subject reduction (\refprop{name-subject-reduction}). Note that $\nctxp\tmthree$ does not have type $\atype$, but \refprop{name-subject-reduction} is valid for any linear type, not just $\atype$.

	\item \emph{Left of a substitution}, that is, $\sictx = \sictxtwo \esub{\var}{\tmfive}$. Then $\tderiv$ has the following form:
		\[
\AxiomC{ $\tderiv_\tmthree \exder \tyjp{(\msteps_\tmthree,\esteps_\tmthree)}{\sictxtwop{\tmthree}}{\typctx_\tmthree, \var:\mtype}{\atype}$ }
\AxiomC{ $\tderiv_\tmfive \exder\tyjp{(\msteps_\tmfive,\esteps_\tmfive)}{\tmfive}{\typctx_\tmfive}{\mtype \mplus\mult\atype} $ }
\RightLabel{$\ruleES$}
\BinaryInfC{ $\tyjp{(\msteps_\tmthree + \msteps_\tmfive, \esteps_\tmthree + \esteps_\tmfive)}{\sictxtwop{\tmthree} \esub\var\tmfive}{\typctx_\tmthree \mplus \typctx_\tmfive}{\atype}$ }
\DisplayProof
\]
		With $\mtype\neq\zero$, $\typctx = \typctx_\tmthree  \mplus  \typctx_\tmfive$, $\msteps = \msteps_\tmthree + \msteps_\tmfive$, and $\esteps = \esteps_\tmthree + \esteps_\tmfive$. 
		
By \ih, $\msteps_\tmthree \geq 1$, and so $\msteps \geq 1$, and there exists a derivation $\tderivtwo_\tmthree \exder \tyjp{(\msteps_\tmthree-1,\esteps_\tmthree)}{\sictxtwop{\tmfour}}{\typctx_\tmthree, \var:\mtype}{\atype}$, thus allowing us to construct $\tderivtwo$ as follows:
\[
\AxiomC{ $\tderivtwo_\tmthree \exder\tyjp{(\msteps_\tmthree-1,\esteps_\tmthree)}{\sictxtwop{\tmfour}}{\typctx_\tmthree, \var:\mtype}{\atype}$ }
\AxiomC{ $\tderiv_\tmfive \exder \tyjp{(\msteps_\tmfive,\esteps_\tmfive)}{\tmfive}{\typctx_\tmfive}{\mtype \mplus\mult\atype} $ }
\RightLabel{$\ruleES$}
\BinaryInfC{ $\tyjp{(\msteps_\tmthree -1 + \msteps_\tmfive, \esteps_\tmthree + \esteps_\tmfive)}{\sictxtwop{\tmfour} \esub\var\tmfive}{\typctx_\tmthree \mplus \typctx_\tmfive}{\atype}$ }
\DisplayProof
\]

\item \emph{Right of a substitution}, that is, $\sictx = \ans \esub\var{\sictxtwo}$. The last rule of $\tderiv$ is $\ruleES$, of premise $\tyjp{(\msteps_\ans,\esteps_\ans)}{\ans}{\typctx_\ans}{\atype}$.

By \refprop{silly-normal-forms-forall}, we have that $\typctx_\ans$ is empty and $\msteps_\ans=0=\esteps_\ans$, so that $\tderiv$ actually has the following shape:
\[
\AxiomC{ $\tyjp{(0,0)}{\ans}{}{\atype}$ }
\AxiomC{ $\tyjp{(\msteps,\esteps)}{\sictxtwop\tmthree}{\typctx}{\atype}$ }
\RightLabel{$\ruleMany$}
\UnaryInfC{ $\tyjp{(\msteps,\esteps)}{\sictxtwop\tmthree}{\typctx }{\mult\atype} $ }
\RightLabel{$\ruleES$}
\BinaryInfC{ $\tyjp{(\msteps, \esteps)}{\ans \esub\var{\sictxtwop\tmthree} }{\typctx}{\atype}$ }
\DisplayProof
\]
By \ih, there is a derivation of final judgement $\tyjp{(\msteps-1,\esteps)}{\sictxtwop\tmfour}{\typctx}{\atype}$, and we obtain the following derivation:
\[
\AxiomC{ $\tyjp{(0,0)}{\ans}{}{\atype}$ }
\AxiomC{ $\tyjp{(\msteps-1,\esteps)}{\sictxtwop\tmfour}{\typctx}{\atype}$ }
\RightLabel{$\ruleMany$}
\UnaryInfC{ $\tyjp{(\msteps-1,\esteps)}{\sictxtwop\tmfour}{\typctx }{\mult\atype} $ }
\RightLabel{$\ruleES$}
\BinaryInfC{ $\tyjp{(\msteps-1, \esteps)}{\ans \esub\var{\sictxtwop\tmfour} }{\typctx}{\atype}$ }
\DisplayProof
\]

\end{itemize}

\item \emph{Exponential}. Cases:
\begin{itemize}
\item \emph{Root step of $\tosieyn$}, \ie  $\tm  \rtoep\nctx \tmtwo$. Then it follows from \cbn subject reduction (\refprop{name-subject-reduction}). 

\item \emph{Root step of $\tosieaa$}, \ie  $\tm  \rtoep\sictx \tmtwo$. It goes exactly as the root case for $\tone$ for \cbn subject reduction (\refprop{name-subject-reduction}), except that one applies the right tight linear substitution lemma (see the technical report \cite[Lemma 57]{accattoli2024mirroring}) instead of the \cbn one.
	
		\item \emph{Contextual closure.} As in the $\tom$ case, the only change is that one looks at the $\esteps$ index rather than the $\msteps$ one. Note that indeed those cases do not depend on the details of the step itself, but only on the context enclosing it. \qedhere
\end{itemize}
\end{enumerate}
\end{proof}


The key point of tight subject reduction is that now the indices decrease of \emph{exactly one} at each step, which entails a tight variant of correctness.

\begin{thm}[Tight correctness for \cbs]
Let $\tm$ be a closed term and $\tderiv \exder   \Deri[(\msteps, \esteps)] {}{\tm}{\atype}$ be a tight derivation. 
Then there is a weak normal form $\ntm$ such that $\deriv \colon \tm
  \tosi^* \ntm$ with $\sizep\deriv{\sisym\msym} =\msteps$ and $\sizep\deriv{\sieaasym,\sieynsym} = \esteps$.\label{thm:tight-correctness}
\end{thm}
%
\begin{proof} 
The proof is exactly as for weak correctness (\refthm{weak-correctness}), except that if $\tm$ is normal then the fact that $\msteps = \esteps = 0$ follows from \refprop{silly-normal-forms-forall} and that if $\tm$ is not normal the equality on the number of steps is obtained by  using tight subject reduction (\refprop{silly-subject-reduction}) instead of the quantitative one.
\qedhere
\end{proof}

For tight completeness for closed terms, it is enough to observe that the statement of weak completeness (\refthm{open-completeness}) already gives a tight derivation $\typctx\vdash \tm\hastype\atype$, the type context $\typctx$ of which is empty because $\tm$ is closed (\reflemma{typctx-varocc-tm}), and the indices of which are omitted because they are irrelevant for completeness.




\begin{exa}
We illustrate the tightness of multi types for \cbs with an example. Consider the term $(\la\vartwo\vartwo\vartwo)(\Id \Id)$, the \cbs evaluation of which was given in \refex{evaluation} (page \pageref{ex:evaluation}) and is as follows (the first part coincides with the \cbn evaluation):
\begin{center}\small$
\begin{array}{rll}
	\multicolumn{2}{l}{\textsc{CbN evaluation:}}
	\\
(\la\vartwo\vartwo\vartwo)(\Id \Id) 
\\
\tosim& \vartwo\vartwo\esub\vartwo{\Id \Id} &\tosieyn (\Id \Id)\vartwo \esub\vartwo{\Id\Id}
\\ 
\tosim& (\var\esub\var\Id)\vartwo\esub\vartwo{\Id\Id}
&\tosieyn (\Id\esub\var\Id)\vartwo\esub\vartwo{\Id\Id}
\\
\tosim &\varthree\esub\varthree\vartwo\esub\var\Id\esub\vartwo{\Id\Id} &\tosieyn \vartwo\esub\varthree\vartwo\esub\var\Id\esub\vartwo{\Id\Id}
\\&& \tosieyn \Id \Id \esub\varthree\vartwo\esub\var\Id\esub\vartwo{\Id\Id}
\\
\tosim& \varthreep\esub\varthreep{\Id}\esub\varthree\vartwo\esub\var\Id\esub\vartwo{\Id\Id} &\tosieyn \Id \esub\varthreep{\Id}\esub\varthree\vartwo\esub\var\Id\esub\vartwo{\Id\Id}
\\[5pt]
\hdashline
\\[-5pt]
\multicolumn{2}{l}{\textsc{CbS extension:}}
\\
&& \tosieaa \Id \esub\varthreep{\Id}\esub\varthree{\Id \Id}\esub\var\Id\esub\vartwo{\Id\Id} 
\\
\tosim &\Id \esub\varthreep{\Id}\esub\varthree{\varthreep\esub\varthreep\Id}\esub\var\Id\esub\vartwo{\Id\Id} &\tosieyn \Id \esub\varthreep{\Id}\esub\varthree{\Id\esub\varthreep\Id}\esub\var\Id\esub\vartwo{\Id\Id} 
\\
\tosim &\Id \esub\varthreep{\Id}\esub\varthree{\Id\esub\varthreep\Id}\esub\var\Id\esub\vartwo{\varthreep\esub\varthreep\Id}  &\tosieyn \Id \esub\varthreep{\Id}\esub\varthree{\Id\esub\varthreep\Id}\esub\var\Id\esub\vartwo{\Id\esub\varthreep\Id} 
\end{array}
$\end{center}

If we consider \cbn evaluation, evaluation stops earlier: it would reach a normal form with $\Id \esub\varthreep{\Id}\esub\varthree\vartwo\esub\var\Id\esub\vartwo{\Id\Id}$, that is after $4$ multiplicative and $5$ exponential steps.
\cbs evaluates more than \cbn, hence the rewriting sequence ends only after $6$ multiplicative and $8$ exponential steps.

We now show a tight derivation for the term $(\la\vartwo\vartwo\vartwo)(\Id \Id)$, which is indeed indexed by $(6,8)$, as per tight correctness and completeness (\refthm{tight-correctness} and the text after it). For compactness, we shorten $\atype$ as $\ntype$. Moreover, to fit the derivation in the page, we first give the main derivation, giving a name to some sub-derivations which are shown after the main one. 
\begin{itemize}
\item Main derivation:
\[
\AxiomC{$\tderiv\Deri[(1,3)]{}{\la\vartwo\vartwo\vartwo}{\arrowtype{\mult{\arrowtype{\mult\ntype}\ntype,\ntype,\ntype}}  \ntype}$}
\AxiomC{$\tderivtwo\Deri[(1,2)]{}{\Id\Id}{\arrowtype{\mult\ntype} \ntype}$}
\AxiomC{$(\tderivtwo_\ntype\Deri[(1,1)]{}{\Id\Id}{\ntype})\times3$}
\BinaryInfC{$\Deri[(4,5)]{}{\Id \Id}{\mult{\arrowtype{\mult\ntype}\ntype,\ntype,\ntype}\uplus\mult\ntype}$}
\BinaryInfC{$\Deri[(6,8)]{}{(\la\vartwo\vartwo\vartwo)(\Id\Id)}{ \ntype}$  }
\DisplayProof\]

\item Left premise:
\begin{center}
$\tderiv \defeq$~~~~
\AxiomC{$\Deri[(0,1)]{\vartwo\hastype\mult{\arrowtype{\mult\ntype}\ntype}}{\vartwo}{\arrowtype{\mult\ntype}\ntype}$}
\AxiomC{$\Deri[(0,1)]{\vartwo\hastype\mult\ntype }{\vartwo}{\ntype}$}
\AxiomC{$\Deri[(0,1)]{\vartwo\hastype\mult\ntype }{\vartwo}{\ntype}$}
\BinaryInfC{$\Deri[(0,2)]{\vartwo\hastype\mult{\ntype,\ntype} }{\vartwo}{\mult\ntype\uplus\mult\ntype}$}
\BinaryInfC{$\Deri[(1,3)]{\vartwo\hastype\mult{\arrowtype{\mult\ntype} \ntype,\ntype,\ntype}}{\vartwo\vartwo}{ \ntype}$}
\UnaryInfC{$\Deri[(1,3)]{}{\la\vartwo\vartwo\vartwo}{\arrowtype{\mult{\arrowtype{\mult\ntype} \ntype,\ntype,\ntype}} \ntype}$}
\DisplayProof
\end{center}

\item Middle premise:
\begin{center}
$\tderivtwo_\ntype\defeq$~~~~
\AxiomC{$\tderiv_\ntype \Deri[(0,1)]{}{\Id}{\mult{\arrowtype{\mult\ntype}\ntype}}$}
\AxiomC{}
\RightLabel{$\ruleAxLam$}
\UnaryInfC{$\Deri[(0,0)]{}{\Id}{\ntype}$ }
\AxiomC{}
\RightLabel{$\ruleAxLam$}
\UnaryInfC{$\Deri[(0,0)]{}{\Id}{\ntype}$ }
\RightLabel{$\ruleMany$}
\BinaryInfC{$\Deri[(0,0)]{}{\Id}{\mult{\ntype,\ntype}}$}
\RightLabel{$\ruleAp$}
\BinaryInfC{$\Deri[(1,1)]{}{\Id \Id}{\ntype}$ }
\DisplayProof
\end{center}

\item Right premise:
\begin{center}
$\tderivtwo \defeq$~~~~
\AxiomC{ }
\RightLabel{$\ruleAx$}
\UnaryInfC{$\Deri[(0,1)]{\varthree\hastype\mult{\arrowtype{\mult\ntype}\ntype}}{\varthree}{\arrowtype{\mult\ntype}\ntype}$}
\RightLabel{$\ruleLam$}
\UnaryInfC{$\Deri[(0,1)]{}{\Id}{\arrowtype{\mult{\arrowtype{\mult\ntype}\ntype}} {(\arrowtype{\mult\ntype}\ntype)}}$}
\AxiomC{$\tderiv_\ntype \Deri[(0,1)]{}{\Id}{{\arrowtype{\mult\ntype}\ntype}}$}
\AxiomC{}
\RightLabel{$\ruleAxLam$}
\UnaryInfC{$\Deri[(0,0)]{}{\Id}{\ntype}$ }
\RightLabel{$\ruleMany$}
\BinaryInfC{$\Deri[(0,1)]{}{\Id}{\mult{\arrowtype{\mult\ntype}\ntype,\ntype}}$}
\RightLabel{$\ruleAp$}
\BinaryInfC{$\Deri[(1,2)]{}{\Id \Id}{\arrowtype{\mult\ntype}\ntype}$}
\DisplayProof
\end{center}
\end{itemize}

\end{exa}

\myparagraph{Maximality of the \cbs Strategy.}
Similarly to how we proved uniform normalization for $\tow$, we can prove that on closed terms the \cbs strategy does reach a weak normal form whenever one exists---the key point being that $\tosi$ does not stop too soon. This fact proves that $\eqcsilly$ can equivalently be defined using $\tosi$, as mentioned in \refsect{cbv}. Moreover, by exploiting tight correctness and some of the rewriting properties of \refsect{rewriting_properties}, we prove that the \cbs strategy is maximal.

\begin{prop}
Let $\tm$ be closed and $\tm \tow^h \sans$ with $\sans$ a strict answer.$\label{prop:silly-longest}$
\begin{enumerate}
\item \emph{\cbs is normalizing}: $\tm\tosi^k\sans$ for some $k\in\nat$;
\item \emph{\cbs is maximal}:  $h\leq k$.
\end{enumerate}
\end{prop}
\begin{proof}
\hfill
\begin{enumerate}
\item By completeness (\refthm{open-completeness}), $\tm \tow^h \sans$ implies typability of $\tm$ (a strict answer is in particular a weak normal form), which in turn, by tight correctness (\refthm{tight-correctness}), implies $\tm\tosi^k\sanstwo$ for some $k\in\nat$. By confluence (\refthm{confluence}), $\sans=\sanstwo$.

\item By postponement of $\towgcv$ (\refprop{gc-postponement}) applied to $\tm \tow^h \sans$, we obtain a sequence $\deriv:\tm \townotgcv^{h_1} \tmtwo \towgcv^{h_2} \sans$ with $h_1+h_2\geq h$, for some $\tmtwo$. For the strategy sequence $\tm\tosi^k\sans$ given by Point 1, the \emph{moreover} part of the postponement property gives us a sequence $\derivtwo:\tm \tosinotgcv^{k_1}\tmthree \tosigcv^{k_2} \sans$ with, crucially, $k_1+k_2=k$, for some $\tmthree$.
By completeness of the type system (\refthm{open-completeness}), we obtain a derivation $\vdashp\msteps\esteps \tm \hastype \atype$.
By weak correctness (\refthm{weak-correctness}), we obtain $h_1\leq \msteps +\esteps$. By tight correctness (\refthm{tight-correctness}), $k_1=\msteps +\esteps$. Thus, $h_1\leq k_1$. Now, note that both $\tmtwo$ and $\tmthree$ are $\townotgcv$-normal, because otherwise, by (iterated) strict commutation of $\townotgcv$ and $\towgcv$ (\reflemma{local-strong-commutation}), we obtain that $\sans$ is not $\tow$-normal, against hypothesis. Since $\townotgcv$ is confluent (\refthm{confluence}), $\tmtwo=\tmthree$. Since $\towgcv$ is diamond (\reflemma{local-confluence}), $h_2=k_2$. Then $k=k_1+k_2=k_1+h_2\geq h_1+h_2 =h$. \qedhere
\end{enumerate}
\end{proof}

\myparagraph{Why not the $\l$-Calculus?}
In the $\l$-calculus, it is hard to specify via evaluation contexts the idea behind the \cbs strategy of \emph{evaluating arguments only when they are no longer needed}. The difficulty is specific to weak evaluation. For instance, for $\tm \defeq (\la\vartwo\la\var\vartwo\vartwo) \tmthree$ the \cbs strategy should evaluate $\tmthree$ before substituting it for $\vartwo$, because once $\tmthree$ ends up under $\l\var$ it shall be unreachable by weak evaluation. For $\tmtwo \defeq (\la\vartwo(\la\var \vartwo\vartwo) \tmfour) \tmthree$, instead, the \cbs strategy should not evaluate $\tmthree$ before substituting it, because weak evaluation will reach the two occurrences of $\vartwo$, and one obtains a longer reduction by substituting $\tmthree$ before evaluating it. Note that one cannot decide what to do with $\tmthree$ in $\tm$ and $\tmtwo$ by checking if $\vartwo$ occurs inside the abstraction, because $\vartwo$ occurs under $\l\vartwo$ in both $\tm$ and $\tmtwo$.
	
	Multi types naturally make the right choices for $\tm$ (evaluating $\tmthree$ before substituting it) and $\tmtwo$ (substituting $\tmthree$ before evaluating it), so a strategy matching exactly the bounds given by multi types needs to do the same choices.
		The LSC allows one to bypass the difficulty, by first turning $\beta$-redexes into explicit substitutions, thus exposing $\vartwo$ out of abstractions in $\tmtwo$ but not in $\tm$, and matching what is measured by multi types. In the $\l$-calculus, we have not found natural ways of capturing this aspect (be careful: the example illustrates the problem but solving the problem requires more than just handling the example). 
\section{Conclusions}
\label{sect:conclusions}
We introduce the weak silly substitution calculus, a \cbs abstract machine,  and the \cbs strategy by mirroring the properties of \cbneed with respect to duplication and erasure. Then, we provide evidence of the good design of the framework via a rewriting study of the calculus and the strategy, and by mirroring the semantic analyses of \cbneed via multi types by Kesner \cite{DBLP:conf/fossacs/Kesner16} (qualitative) and Accattoli et al. \cite{DBLP:conf/esop/AccattoliGL19} (quantitative). 

Conceptually, the main results are the operational equivalence of \cbs and \cbv, mirroring the one between \cbn and \cbneed, and the exact measuring of \cbs evaluation lengths via multi types, having the interesting corollary that \cbs is \emph{maximal} in the Weak SSC.
 It would be interesting to show that, dually, the \cbneed strategy computes the shortest reduction in the \cbneed LSC. We are not aware of any such result.

Our work also shows that \cbv contextual equivalence $\eqcvalue$ is completely blind to the efficiency of evaluation. We think that it is important to look for natural refinements of $\eqcvalue$ which are less blind, or, in the opposite direction, exploring what are the minimal extensions of \cbv that make $\eqcvalue$ efficiency-sensitive.

We would also like to  develop categorical and game semantics for both \cbs and \cbneed, to pinpoint the principles behind the wiseness and the silliness of duplications and erasures. 

\bibliographystyle{alphaurl}
\bibliography{main.bbl}

\end{document}

%% file: macros/ben-macros-preamble.tex
\usepackage{ifthen}
 \usepackage{xspace}
 
\newboolean{talk}
\setboolean{talk}{false}
\newboolean{paper}
\setboolean{paper}{false}

\newboolean{IEEEstyle}
\setboolean{IEEEstyle}{false}
\newboolean{lipicsstyle}
\setboolean{lipicsstyle}{false}
\newboolean{eptcsstyle}
\setboolean{eptcsstyle}{false}
\newboolean{entcsstyle}
\setboolean{entcsstyle}{false}
\newboolean{sigplanstyle}
\setboolean{sigplanstyle}{false}
\newboolean{easychairstyle}
\setboolean{easychairstyle}{false}
\newboolean{scrartclstyle}
\setboolean{scrartclstyle}{false}
\newboolean{lncsstyle}
\setboolean{lncsstyle}{false}
\newboolean{acmartstyle}
\setboolean{acmartstyle}{false}

\newboolean{needstheorems}
\setboolean{needstheorems}{false}
\newboolean{withimages}
\setboolean{withimages}{false}
\newboolean{withproofs}
\setboolean{withproofs}{true}
\newboolean{complexity}
\setboolean{complexity}{false}

\newboolean{techr}
\setboolean{techr}{false}

\newboolean{french}
\setboolean{french}{false}

%% file: macros/ben-macros-paper.tex
\ifthenelse{\boolean{talk}}{}{\usepackage[utf8]{inputenc}}
\ifthenelse{\boolean{french}}{\usepackage[french]{babel}}{
  \ifthenelse{\boolean{lipicsstyle}}{}{
  }}
\ifthenelse{\boolean{acmartstyle}}{}{
\usepackage{amssymb}	
}
\usepackage{amsmath}
\usepackage{amsfonts}
\usepackage{graphicx}
\usepackage{bussproofs}
\usepackage{cmll}
\usepackage{stmaryrd}
 \usepackage{multirow}
\ifthenelse{\boolean{talk}}{\usepackage{color}}{
 	\ifthenelse{\boolean{lipicsstyle}}{\usepackage{color}}{
   		\ifthenelse{\boolean{acmartstyle}}{\usepackage{color}}{
			\usepackage[usenames]{xcolor}
		}
	}
}
 \ifthenelse{\boolean{IEEEstyle}}{
 \usepackage[bookmarks={false}]{hyperref}}{
 \usepackage{hyperref}}
 \usepackage{fancybox}
 \usepackage{hhline}
\usepackage{nameref}

\usepackage[normalem]{ulem}
\usepackage{proof}
\usepackage{mathtools}
\ifthenelse{\boolean{lipicsstyle}}{}{\usepackage{tabls}}
\usepackage{cleveref}

\usepackage{arydshln}

\ifthenelse{\boolean{IEEEstyle}}{
\setboolean{needstheorems}{true}}{}

\ifthenelse{\boolean{eptcsstyle}}{
\setboolean{needstheorems}{true}}{}

\ifthenelse{\boolean{scrartclstyle}}{
\setboolean{needstheorems}{true}}{}

\ifthenelse{\boolean{sigplanstyle}}{
\setboolean{needstheorems}{true}}{}

\ifthenelse{\boolean{easychairstyle}}{
\setboolean{needstheorems}{true}}{}

\ifthenelse{\boolean{lipicsstyle}}{
    }{}

\ifthenelse{\boolean{needstheorems}}{
\input{\macrospath/ben-macros-thm-environments}}{}

\newcommand{\myproof}[1]{
\ifthenelse{\boolean{withproofs}}{#1}{}
}

\ifthenelse{\boolean{lipicsstyle}}{
	\renewcommand{\paragraph}[1]{\subparagraph*{#1}}
}{}

\newcommand{\terms}{\Lambda}

\newcommand{\la}[1]{\lambda #1.}

\newcommand{\tm}{t}
\newcommand{\tmtwo}{u}
\newcommand{\tmthree}{s}
\newcommand{\tmfour}{r}
\newcommand{\tmfive}{p}

\newcommand{\var}{x}
\newcommand{\vartwo}{y}
\newcommand{\varthree}{z}
\newcommand{\varfour}{w}

\newcommand{\vartwop}{y'}
\newcommand{\varthreep}{z'}

\newcommand{\dom}[1]{\symfont{dom}(#1)}

\newcommand{\rootRew}[1]{\mapsto_{#1}}
\newcommand{\Rew}[1]{\rightarrow_{#1}}

\newcommand{\Rewp}[1]{\rightarrow_{#1}^+}
\renewcommand{\to}{\Rew{}}

\newcommand{\lRew}[1]{\prescript{}{#1}{\leftarrow}}

\newcommand{\lto}{\lRew{}}

\newcommand{\tostrat}{\Rew{\symfont{str}}}

\renewcommand{\wn}[1]{\mbox{WN}_{#1}}
\newcommand{\sn}[1]{\mbox{SN}_{#1}}

\newcommand{\symfont}[1]{\mathtt{#1}}

\newcommand{\gcsym}{\symfont{gc}}
\newcommand{\gc}{\gcsym}
\newcommand{\gcvsym}{\gcsym\vsym}
\newcommand{\gcv}{\gcvsym}

\newcommand{\esym}{\symfont{e}}
\newcommand{\wsym}{\symfont{w}}

\newcommand{\msym}{\symfont{m}}

\newcommand{\vsym}{\symfont{v}}
\newcommand{\nsym}{\symfont{n}}

\newcommand{\cbv}{{CbV}\xspace}
\newcommand{\cbs}{{CbS}\xspace}
\newcommand{\cbsilly}{\cbs}

\newcommand{\val}{v}
\newcommand{\valtwo}{\val'}

\newcommand{\ctxholep}[1]{\langle #1\rangle}
\newcommand{\ctxhole}{\ctxholep{\cdot}}

\makeatletter
\newsavebox{\@brx}
\newcommand{\llangle}[1][]{\savebox{\@brx}{\(\m@th{#1\langle}\)}%
  \mathopen{\copy\@brx\kern-0.7\wd\@brx\usebox{\@brx}}}
\newcommand{\rrangle}[1][]{\savebox{\@brx}{\(\m@th{#1\rangle}\)}%
  \mathclose{\copy\@brx\kern-0.7\wd\@brx\usebox{\@brx}}}
\makeatother

\newcommand{\ctxholefp}[1]{\llangle #1\rrangle}

\newcommand{\ctx}{C}
\newcommand{\ctxtwo}{\ctx'}

\newcommand{\ctxp}[1]{\ctx\ctxholep{#1}}

\newcommand{\ctxtwofp}[1]{\ctxtwo\ctxholefp{#1}}

\newcommand{\nbvctxtwo}[1]{\nbvctxtwo{#1}}

\newcommand{\sctx}{S}
\newcommand{\sctxtwo}{{\sctx'}}

\newcommand{\sctxp}[1]{\sctx\ctxholep{#1}}

\newcommand{\defeq}{:=}

\newcommand{\grameq}{::=}

\newcommand{\esub}[2]{[#1{\shortleftarrow}#2]}
\newcommand{\isub}[2]{\{#1{\shortleftarrow}#2\}}

\ifthenelse{\boolean{complexity}}{ }{
  
  }

\newcommand{\rtogc}{\rootRew{\gcsym}}
\newcommand{\rtogcv}{\rootRew{\gcv}}

\newcommand{\togcv}{\Rew{\gcv}}

\newcommand{\llbrace}{\{ \kern -0.27em \vert}
\newcommand{\rrbrace}{\vert \kern -0.27em \}}

\newcommand{\streq}{\equiv}

\renewcommand{\l}{\lambda}
\newcommand{\ie}{i.e.\xspace}
\newcommand{\eg}{e.g.\xspace}
\newcommand{\ih}{{\textit{i.h.}}\xspace}
\newcommand{\fv}[1]{\symfont{fv}(#1)}
\newcommand{\ofv}[1]{\symfont{shfv}(#1)}

\ifthenelse{\boolean{entcsstyle}}{}{
	}

\newcommand{\red}[1]{{\color{red} {#1}}}
\newcommand{\blue}[1]{{\color{blue} {#1}}}

\newcommand{\purple}[1]{{\color{purple} {#1}}}
\newcommand{\orange}[1]{{\color{orange} {#1}}}

\definecolor{dgreen}{rgb}{0.0, 0.5, 0.0}
\newcommand{\dgreen}[1]{{\color{dgreen} {#1}}}

\newcommand{\ignore}[1]{}

\newcommand{\colspace}{@{\hspace{.5cm}}}
\newcommand{\myinput}[1]{\ifthenelse{\boolean{withimages}}{\input{#1}}{}}
\newcommand{\inputproof}[1]{\ifthenelse{\boolean{withproofs}}{\input{#1}}{}}

\newcommand{\reflemma}[1]{Lemma~\ref{l:#1}}

\newcommand{\reflemmaeq}[1]{{L.\ref{l:#1}}}

\newcommand{\refthm}[1]{Th.~\ref{thm:#1}}
\newcommand{\refthmp}[2]{Th.~\ref{thm:#1}.\ref{p:#1-#2}}

\newcommand{\refprop}[1]{Prop.~\ref{prop:#1}}
\newcommand{\refpropp}[2]{Prop.~\ref{prop:#1}.\ref{p:#1-#2}} 
\newcommand{\refsect}[1]{Section~\ref{sect:#1}}

\ifthenelse{\boolean{easychairstyle}}{}{
  \ifthenelse{\boolean{lncsstyle}}{
  	
	}{	
	\renewcommand{\refeq}[1]{(\ref{eq:#1})}
  }
}
\newcommand{\reffig}[1]{Fig.~\ref{fig:#1}}
\newcommand{\refcoro}[1]{Corollary~\ref{coro:#1}}

\newcommand{\refdef}[1]{Def.~\ref{def:#1}}

\newcommand{\refex}[1]{Example~\ref{ex:#1}}

\newcommand{\ax}{\symfont{ax}}

\newcommand{\set}[1]{\{#1\}}

\newcommand{\nat}{\mathbb{N}}

\newcommand{\rtom}{\mapsto_\msym}

\newcommand{\tom}{\Rew{\msym}}
\newcommand{\toe}{\Rew{\esym}}
\newcommand{\tow}{\Rew{\wsym}}

\newcommand{\betav}{\beta_\val}

\newcommand{\tobv}{\Rew{\betav}}

\newcommand{\rtobv}{\rootRew{\betav}}

\newcommand{\cbn}{CbN\xspace}
\newcommand{\cbneed}{CbNeed\xspace}

\newcommand{\wctx}{W}
\newcommand{\wctxp}[1]{\wctx\ctxholep{#1}}
\newcommand{\wctxfp}[1]{\wctx\ctxholefp{#1}}
\newcommand{\wctxtwo}{\wctx'}
\newcommand{\wctxtwop}[1]{\wctxtwo\ctxholep{#1}}

\newcommand{\wctxthree}{\wctx''}

\newcommand{\wctxthreefp}[1]{\wctxthree\ctxholefp{#1}}
\newcommand{\wctxfour}{\wctx'''}
\newcommand{\wctxfourp}[1]{\wctxfour\ctxholep{#1}}
\newcommand{\wctxfourfp}[1]{\wctxfour\ctxholefp{#1}}

\newcommand{\size}[1]{|#1|}
\newcommand{\sizep}[2]{|#1|_{#2}}

\newcommand{\env}{e}
\newcommand{\envtwo}{\env'}

\newcommand{\stack}{\pi}

\ifthenelse{\boolean{acmartstyle}}{
	\renewcommand{\state}{\symfont{s}}
	}{
	\newcommand{\state}{\symfont{Q}}
}
\newcommand{\statetwo}{\state'}
\newcommand{\statethree}{\state''}
\newcommand{\statefour}{\state'''}

\newcommand{\decode}[1]{\underline{#1}}
\newcommand{\decodep}[2]{\decode{#1}\ctxholep{#2}}

\newcommand{\cons}{::}

\newcommand{\deriv}{\ensuremath{d}}

\newcommand{\derivtwo}{\ensuremath{e}}
\newcommand{\derivthree}{\ensuremath{f}}
\newcommand{\derivfour}{\ensuremath{g}}

\newcommand{\renamenop}[1]{#1^\alpha}

\newcommand{\mach}{\symfont{M}}

\newcommand{\tomachhole}[1]{\leadsto_{#1}}
\newcommand{\tomach}{\tomachhole{}}
\newcommand{\tomachm}{\tomachhole{\msym}}

\newcommand{\tomache}{\tomachhole{\esym}}

\newcommand{\seasym}{\symfont{sea}}

\newcommand{\tomachsea}{\tomachhole{\seasym}}
\newcommand{\tomachseaone}{\tomachhole{\seasym_1}}
\newcommand{\tomachseatwo}{\tomachhole{\seasym_2}}

\newcommand{\sizepr}[1]{\sizep{#1}{\msym,\esym}}

\newcommand{\withproofs}[1]{\ifthenelse{\boolean{withproofs}}{#1}{}}

\newcommand{\withoutproofs}[1]{\ifthenelse{\boolean{withproofs}}{}{#1}}

\ifthenelse{\boolean{entcsstyle}}{}{
  \ifthenelse{\boolean{lncsstyle}}{}{	
    
  }
 }

\newcounter{numberone}

\newcommand{\med}[2]{
($(#1)!.5!(#2)$)
}

\newcommand{\actx}{A}
\newcommand{\actxp}[1]{\actx\ctxholep{#1}}

\newcommand{\actxtwo}{\actx'}

\newcommand{\mamst}[3]{(#1,#2,#3)}
\newcommand{\glamst}[4]{(#1,#2,#3,#4)}

\newcommand{\itm}{i}
\newcommand{\itmtwo}{\itm'}
\newcommand{\itmthree}{\itm''}

\newcommand{\ccbv}{Closed \cbv}
\newcommand{\ocbv}{Open \cbv}

\newcommand{\tmp}{\tm'}

\newcommand{\mult}{\mathsf{m}} 

\newcommand{\typctx}{\Gamma}
\newcommand{\typctxtwo}{\Delta}

\newcommand{\type}{A}
\newcommand{\typetwo}{B}

\newcommand{\ltype}{L}
\newcommand{\ltypetwo}{\ltype'}

\newcommand{\tderiv}{\pi}
\newcommand{\tderivtwo}{\rho}
\newcommand{\tderivthree}{\sigma}

\newcommand{\sem}[1]{[\![#1]\!]}

\newcommand{\nctx}{N}
\newcommand{\nctxtwo}{\nctx'}

\newcommand{\nctxp}[1]{\nctx\ctxholep{#1}}

\newcommand{\Id}{\symfont{I}}

\newcommand{\mtype}{M}
\newcommand{\mtypetwo}{N}
\newcommand{\mtypethree}{O}
\newcommand{\ttype}{T}
\newcommand{\ttypetwo}{\ttype'}

\newcommand{\hastype}{\,{:}\,}

\newcommand{\tarrow}[2]{#1 \rightarrow #2}
\newcommand{\iI}{{i \in I}}

\renewcommand{\mult}[1]{[ #1 ] }

\makeatletter
\newcommand{\exder}{%
  \def\exderW[##1]{\triangleright_{##1}\ }%
  \def\exderWO{\triangleright\ }%
  \@ifnextchar[\exderW\exderWO%
  }
\makeatother

\renewcommand{\exder}{\,\triangleright}

\newcommand\Deribbase[5]{{#3}\ {\pmb\vdash}_{#2}^{#1} {#4}\  {:}\  {#5}}

\makeatletter

\newcommand{\Deribase}[1]{%
  \def\DeribW[##1]{\Deribbase{##1}{#1}}%
  \def\DeribWO{\Deribbase{}{#1}}%
  \@ifnextchar[\DeribW\DeribWO%
  }

  \newcommand{\Deri}{%
  \def\DeriW_##1{\Deribase{##1}}%
  \def\DeriWO{\Deribase{}}%
  \@ifnextchar_\DeriW\DeriWO%
  }
\makeatother

\newcommand\mplus{\uplus}

\newcommand{\app}{\symfont{app}}

\makeatletter
\newcommand{\appresult}{%
  \def\appresultW<##1>{\app_\result^{##1}}%
  \def\appresultWO{\app_\result}%
  \@ifnextchar<\appresultW\appresultWO%
  }
\makeatother

\newcommand{\sm}{\setminus \!\!\! \setminus}

\newcommand\mydots{\hbox to .6em{.\hss.}}

\newcommand{\mset}[1]{[#1]}

\newcommand{\zero}{\blue{\mathbf{0}}}

\newcommand{\emptymset}{\zero}

\newcommand{\ntm}{n}
\newcommand{\ntmtwo}{\ntm'}

\newcommand{\vctx}{V}
\newcommand{\vctxp}[1]{\vctx\ctxholep{#1}}

\newcommand{\letin}[3]{{\sf let}\ #1=#2\ {\sf in}\ #3}

\newcommand{\lctx}{L}
\renewcommand{\lctx}{S}
\newcommand{\lctxtwo}{\lctx'}

\newcommand{\lctxp}[1]{\lctx\ctxholep{#1}}

\newcommand{\lctxtwop}[1]{\lctxtwo\ctxholep{#1}}

\newcommand{\cameratech}[2]{\ifthenelse{\boolean{techr}}{#2}{#1}}

\newcommand{\octx}{W}

\newcommand{\octxfp}[1]{\octx\ctxholefp{#1}}

\newcommand{\ans}{a}
\newcommand{\anstwo}{\ans'}
\newcommand{\ansthree}{\ans''}

\newcommand{\sans}{\ans_s}
\newcommand{\sanstwo}{\sans'}

\newcommand{\rtosigthree}{\rootRew{\sigma_3}}

\newcommand{\sisym}{\symfont{y}}

\newcommand{\tosip}[1]{\Rew{\sisym{#1}}}
\newcommand{\tosi}{\Rew{\sisym}}
\newcommand{\tosinotgcv}{\tosip{\neg\gcv}}

\newcommand{\sieaasym}{\sisym\esym_{\symfont{AY}}}
\newcommand{\sieynsym}{\sisym\esym_{\symfont{YN}}}

\newcommand{\tosieaa}{\Rew{\sieaasym}}
\newcommand{\tosieyn}{\Rew{\sieynsym}}

\newcommand{\tosim}{\tosip\msym}

\newcommand{\tosigcv}{\tosip\gcv}

\newcommand{\atype}{\mathtt{n}}
\newcommand{\arrowtype}[2]{#1 \rightarrow #2}

\newcommand{\ruleAp}{@}

\newcommand{\ruleLam}{\l}

\newcommand{\ruleAx}{\mathsf{ax}}

\newcommand{\ruleAxLam}{\mathsf{ax}_\l}

\newcommand{\ruleES}{\mathsf{ES}}

\newcommand{\ruleMany}{\mathsf{many}}

\newcommand{\esymp}[1]{\esym_{#1}}
\newcommand{\evsymp}[1]{\esym\vsym_{#1}}
\newcommand{\rtoep}[1]{\rootRew{\esymp{#1}}}
\newcommand{\rtoevp}[1]{\rootRew{\evsymp{#1}}}

\newcommand{\msteps}{m}
\newcommand{\mstepstwo}{\msteps'}

\newcommand{\esteps}{e}
\newcommand{\estepstwo}{\esteps'}

\newcommand{\normal}{\atype}
\newcommand{\tyjp}[4]{{#3} \vdash^{#1} #2 \hastype #4}
\newcommand\bigmplus{\mplus}

\newcommand{\sictx}{Y}
\newcommand{\sictxtwo}{\sictx'}

\newcommand{\sictxp}[1]{\sictx\ctxholep{#1}}
\newcommand{\sictxtwop}[1]{\sictxtwo\ctxholep{#1}}

\newcommand{\sictxfp}[1]{\sictx\ctxholefp{#1}}
\newcommand{\sictxtwofp}[1]{\sictxtwo\ctxholefp{#1}}

\newcommand{\cbvaluesym}{\symfont{value}}
\newcommand{\cbsillysym}{\symfont{silly}}

\newcommand{\equivcp}[1]{\simeq^{\mathrm{ctx}}_{#1}}

\newcommand{\eqcvalue}{\equivcp\cbvaluesym}
\newcommand{\eqcsilly}{\equivcp\cbsillysym}

\newcommand{\ctxeq}{\equivcp{}}

\definecolor{LightGray}{gray}{.80}
\definecolor{DarkGray}{gray}{.60}

\newcommand{\vdashp}[2]{\vdash^{(#1,#2)}}

\newcommand{\wesym}{\wsym\esym}
\newcommand{\towm}{\Rew{\wsym\msym}}
\newcommand{\towe}{\Rew{\wesym}}
\newcommand{\towev}{\Rew{\wesym\vsym}}
\newcommand{\towgc}{\Rew{\wsym\gc}}
\newcommand{\towgcv}{\Rew{\wsym\gcv}}

\newcommand{\nesym}{\nsym\esym}
\newcommand{\tonm}{\Rew{\nsym\msym}}
\newcommand{\tone}{\Rew{\nesym}}
\newcommand{\tongc}{\Rew{\nsym\gc}}
\newcommand{\toname}{\Rew{\nsym}}

\newcommand{\tonotgcv}{\Rew{\neg\gcv}}
\newcommand{\townotgcv}{\Rew{\wsym\neg\gcv}}

\newcommand{\Rule}{\mathsf{r}}

\newcommand{\dgray}[1]{{\color{darkgray} {#1}}}

\newcommand{\wgcv}{\wsym\gcv}

\newcommand{\ntype}{\red{\symfont{n}}}

\newcommand{\tonnotgc}{\Rew{\nsym\neg\gcsym}}

\newcommand{\techrep}[1]{\ifthenelse{\boolean{withproofs}}{#1}{}}
\newcommand{\camerar}[1]{\ifthenelse{\boolean{withproofs}}{}{#1}}

\newcommand{\emptylist}{\epsilon}
\newcommand{\States}{\symfont{States}}
\newcommand{\compilrelsym}{\triangleleft}
\newcommand{\compilrel}[2]{#1\compilrelsym#2}
\newcommand{\fourstate}[4]{(#1,#2,#3,#4)}
\newcommand{\pans}{\symfont{C}}
\newcommand{\prsym}{\symfont{pr}}
\newcommand{\tomachpr}{\tomachhole{\prsym}}

\newcommand{\esterms}{\Lambda_{\symfont{es}}}
\newcommand{\run}{r}
\newcommand{\runtwo}{\run'}
\newcommand{\runthree}{\run''}

\renewcommand{\pans}{\symfont{A}}

\renewcommand{\stack}{\symfont{S}}
\renewcommand{\env}{\symfont{E}}

\newcommand{\eqtype}{\equiv^{\mathrm{typ}}}

\newcommand{\stred}{\Rrightarrow}
\newcommand{\stredp}[1]{\Rrightarrow_{#1}}
\newcommand{\stredapp}{\stred_{@l}}
\newcommand{\stredans}{\stred_{\symfont{ans}}}
\newcommand{\asym}{\symfont{a}}
\newcommand{\bsym}{\symfont{b}}

\newcommand{\myparagraph}[1]{\subsection*{#1}}

\newcommand{\nfo}[1]{\symfont{nf}_{\symfont{sea}}(#1)}

\def\checkmark{\tikz\fill[scale=0.4](0,.35) -- (.25,0) -- (1,.7) -- (.25,.15) -- cycle;}

%% file: macros/ben-macros-diagrams.tex
\usepackage{tikz}
\usetikzlibrary{calc}
\usetikzlibrary{matrix,arrows}
\usetikzlibrary{decorations.shapes}
\usetikzlibrary{decorations.text}
\usetikzlibrary{decorations.pathmorphing}
\usetikzlibrary{decorations.markings}
\usetikzlibrary{positioning}

\tikzset{
node distance=1.3cm, auto,
every node/.style={font=\scriptsize },
ocenter/.style={baseline={([yshift=-.5ex, xshift=-.5ex]current bounding box)}},  
labelBeginAbove/.style={postaction={decorate,decoration={markings,mark=at position 0 with {\node[inner sep= 0.6pt, above=1pt]{\tiny #1};}} } },
labelBeginBelow/.style={postaction={decorate,decoration={markings,mark=at position 0 with {\node[inner sep= 0.6pt, below=1pt]{\tiny #1};}}}},
labelEndAbove/.style={postaction={decorate,decoration={markings,mark=at position 1 with {\node[inner sep= 0.6pt, above=1pt]{\tiny #1};}}}},
labelEndBelow/.style={postaction={decorate,decoration={markings,mark=at position 1 with {\node[inner sep= 0.6pt, below=1pt]{\tiny #1};}}}},
labelEndRight/.style={postaction={decorate,decoration={markings,mark=at position 1 with {\node[inner sep= 0.6pt, right=1pt]{\tiny #1};}}}},
labelEndLeft/.style={postaction={decorate,decoration={markings,mark=at position 1 with {\node[inner sep= 0.6pt, left=1pt]{\tiny #1};}}}}
}


%% file: local-macros.tex
\renewcommand{\atype}{\red{\mathtt{norm}}}
\renewcommand{\arrowtype}[2]{#1 \red\rightarrow #2}
\renewcommand{\ltype}{\red{L}}
\renewcommand{\ltypetwo}{\ltype\red{'}}

\renewcommand{\mtype}{\blue{M}}
\renewcommand{\mtypetwo}{\blue{N}}
\renewcommand{\mtypethree}{\blue{O}}
\renewcommand{\ttype}{\purple{T}}
\renewcommand{\ttypetwo}{\ttype\purple{'}}

\renewcommand{\typctx}{\blue{\Gamma}}
\renewcommand{\typctxtwo}{\blue{\Delta}}

\renewcommand{\mset}[1]{\blue{[}#1\blue{]}}

\renewcommand{\tarrow}[2]{#1 \,\red{\rightarrow}\, #2}
\renewcommand{\mult}[1]{\mset{#1}}